\documentclass[11pt,a4paper]{article}
\usepackage{epsfig}
\usepackage[T1]{fontenc}    
\usepackage{graphics}
\usepackage{graphicx}
\usepackage{pstricks,pst-coil,pst-fill,pst-plot}
\usepackage[fleqn]{amsmath}    
\usepackage{amssymb}    
\usepackage{amsfonts}   
\usepackage{verbatim}   
\usepackage{mathrsfs}   
\usepackage{dsfont}
\usepackage{euscript}
\usepackage{yfonts}
\usepackage{enumerate}     
\usepackage{amsthm}         
\usepackage{txfonts}
\usepackage{marvosym}
\usepackage{stmaryrd}
\usepackage{vmargin}        
\usepackage{wasysym}		

\setmarginsrb{1.8cm}{2cm}{1.8cm}{2cm}{1cm}{1cm}{1cm}{1.6cm}
\makeatletter
\@addtoreset{equation}{section}
\makeatother

\providecommand{\bysame}{\leavevmode\hbox to3em{\hrulefill}\thinspace}
\providecommand{\MR}{\relax\ifhmode\unskip\space\fi MR }

\providecommand{\href}[2]{#2}

\DeclareMathAlphabet{\mathdutchcal}{U}{dutchcal}{m}{n}

\let\tend=\rightarrow

\long\def\symbolfootnote[#1]#2{\begingroup%
	\def\thefootnote{\fnsymbol{footnote}}\footnote[#1]{#2}\endgroup}

\newtheorem{theorem}{Theorem}[section]
\newtheorem{prop}[theorem]{Proposition}
\newtheorem*{theorem*}{Theorem}
\newtheorem{cor}[theorem]{Corollary}

\newtheorem{conj}[theorem]{Conjecture}

\newtheorem{lemme}[theorem]{Lemma}

\newtheorem*{bootstrapaxiomsbound}{Bootstrap axioms with bound states}

\def\Proof{\medskip\noindent {\it Proof --- \ }}

\def\qed{\hfill\rule{2mm}{2mm}}

\usepackage{tikz}
\usetikzlibrary{arrows.meta}
\usetikzlibrary{decorations.markings}

\newcommand\beq{\begin{equation}}
	\newcommand\enq{\end{equation}}
\newcommand\bem{\begin{multline}}
	\newcommand\enm{\end{multline}}

\def\beqa{\begin{eqnarray}}
	\def\eeqa{\end{eqnarray}}
\def\ba{\begin{array}}
	\def\ea{\end{array}}
\def\det{\operatorname{det}}

\newcommand{\f}[2]{{\ensuremath{%
			\mathchoice%
			{\dfrac{#1}{#2}}
			{\dfrac{#1}{#2}}
			{\frac{#1}{#2}}
			{\frac{#1}{#2}}
}}}
\newcommand{\tf}[2]{\ensuremath{#1/#2}}

\def\a{\alpha}

\def\be{\beta}
\def\ga{\gamma}
\def\Ga{\Gamma}

\def\de{\delta}

\def\eps{\epsilon}
\def\veps{\varepsilon}
\def\la{\lambda}
\def\La{\Lambda}

\def\sg{\sigma}

\def\Sg{\Sigma}

\def\ups{\upsilon}
\def\th{\theta}

\def\Om{\Omega}
\def\om{\omega}
\def\vp{\varphi}

\newcommand{\mc}[1]{\ensuremath{\mathcal{#1}}}
\newcommand{\mf}[1]{\ensuremath{\mathfrak{#1}}}
\newcommand{\msc}[1]{\ensuremath{\mathscr{#1}}}

\newcommand{\bs}[1]{\ensuremath{\boldsymbol{#1}}}

\DeclareFontFamily{OT1}{pzc}{}
\DeclareFontShape{OT1}{pzc}{m}{it}{<-> s * [1.10] pzcmi7t}{}
\DeclareMathAlphabet{\mathpzc}{OT1}{pzc}{m}{it}

\def \i{ \mathrm i}

\newcommand{\ov}[1]{\ensuremath{\overline{#1}}}
\newcommand{\wt}[1]{\ensuremath{\widetilde{#1}}}
\newcommand{\wh}[1]{\ensuremath{\widehat{#1}}}

\newcommand{\Int}[2]{\ensuremath{\int\limits_{#1}^{#2}}}

\newcommand{\sul}[2]{\ensuremath{\sum\limits_{#1}^{#2}}}
\newcommand{\pl}[2]{\ensuremath{\prod\limits_{#1}^{#2}}}

\newcommand{\R}{\ensuremath{\mathbb{R}}}
\newcommand{\Cx}{\ensuremath{\mathbb{C}}}
\newcommand{\Bx}{\ensuremath{\mathbb{B}}}

\newcommand{\Dp}[1]{\ensuremath{\partial_{#1}}}

\newcommand{\ex}[1]{\ensuremath{\e{e}^{#1}}}

\def\Res{\operatorname{Res}}

\newcommand{\op}[1]{ \boldsymbol{ \texttt{#1} } }

\newcommand{\abs}[1]{\ensuremath{\left| #1 \right|}}

\newcommand{\dd}{\mathrm{d}}
\newcommand{\e}[1]{\ensuremath{\mathrm{#1}}}

\newcommand{\intff}[2]{\ensuremath{ [  #1 \,; #2 ] }}

\newcommand{\intn}[2]{\ensuremath{[\![ \, #1 \,;\, #2 \,]\!]}}

\begin{document}

	\begin{center}
		\begin{LARGE}
			{\bf Form factors of field exponentials in the 1+1 dimensional Sine-Gordon model}
		\end{LARGE}

		\vspace{1cm}
		
		\vspace{4mm}
		{\large Karol K. Kozlowski \footnote{e-mail: karol.kozlowski@ens-lyon.fr}}%
		\\[1ex]
		Univ Lyon, ENS de Lyon, Univ Claude Bernard Lyon 1, CNRS, Laboratoire de Physique, F-69342 Lyon, France \\[2.5ex]

		{\large Alex Simon \footnote{e-mail: alex.simon@ens-lyon.fr}}%
		\\[1ex]
		Univ Lyon, ENS de Lyon, Univ Claude Bernard Lyon 1, CNRS, Laboratoire de Physique, F-69342 Lyon, France \\[2.5ex]

		\par

		\vspace{40pt}

		\centerline{\bf Abstract} \vspace{1cm}
		\parbox{12cm}{\small   This work provides explicit expressions, in all sectors of the Fock space: soliton, antisoliton and breather, for the fundamental building blocks -the form factors-
		of the generalised integral kernels which represent the exponential of the field operator in the 1+1 dimensional Sine-Gordon quantum field theory. The obtained form factors reproduce the various results obtained earlier in the literature using, in particular, the free field approach.

		}
		
	\end{center}
	
	

	\vspace{40pt}
	
	\tableofcontents

	\section{Introduction}
	
 The bootstrap program for  massive integrable quantum field theories in 1+1 dimensions
may be seen as one possible path, alternative to constructive renormalisation, stochastic quantization or Gaussian multiplicative chaos,
for constructing the correlation function this specific class of low-dimensional quantum field theories. This approach  was initiated in the work of
Karowski and Weisz \cite{KarowskiWeiszFormFactorsFromSymetryAndSMatrices}, later developed by Smirnov
using the quantum Gel'fand-Levitan-Marchenko equations \cite{SmirnovUseGLMEqnsTocomputeSineGordonFF,SmirnovGLMEqnsDerivationForSineGordon,SmirnovUseGLMEqnsForSineGordonFF}
until reaching its present form in the works of Smirnov \cite{SmirnovIntegralRepSolitonFFSineGordonBootstrap}, Kirillov and Smirnov \cite{KirillovSmirnovFirstCompleteSetBootstrapAxiomsForQIFT}.

 The bootstrap program is rooted with the $\op{S}$-matrix program \cite{HeisenbergSomeAspectsofSMatrixIdeas,WheelerFirstIntroConceptSmatrix} and relies on the possibility to
 compute exactly and in closed form the $\op{S}$-matrix of such theories. Indeed, it turns out that, in 1+1 dimensions, the existence of infinitely many local conservation laws
is compatible with the existence of non-trivial $\op{S}$-matrices. On top of the usual properties enjoyed in quantum field theory such as unitarity or
crossing symmetry, $\op{S}$-matrices pertaining to integrable models only involve elastic scattering -no particle production- and
have to comply with the factorizability of the $n$-particle scattering into a concatenation of two-body processes. This implies that such $\op{S}$-matrices
have to satisfy the celebrated Yang-Baxter equation which, in this setting expresses the independence of the decomposition of a three-body scattering into
a sequence of two-body ones.  This exact integrable $\op{S}$-matrix was proposed by Gryanik and Vergeles \cite{GryanikVergelesSMatrixAndOtherStuffForSinhGordon}
for the quantum Sinh-Gordon model where only one type of asymptotic particle occurs, thus resulting in a scalar scattering. The quantum Sine-Gordon model,
corresponding to the quantization of the  $1+1$ dimensional classical Sine-Gordon Lagrangian density
\beq
\mc{L}_{\e{SG}} = \f{1}{2} \Dp{\mu}\vp \Dp{}^{\mu}\vp \, + \, \f{M^2}{g^2} \cos(g \vp) \;, M>0 \quad \e{and} \quad g \in \intff{ 0 }{ \sqrt{8\pi} } \;,
\label{ecriture Lagrangien Sine-Gordon}
\enq
is known to have a much richer
variety of asymptotic particles: solitons and antisolitons \cite{KorepinFaddeevQuantisationOfSolitions}, along with their bound states: the breathers
\cite{DashenHasslacherNeveuIdentificationFromPertThFullPartSPectrSineGordon}. The $\op{S}$-matrix of the soliton/antisoliton sector
was obtained by Zamolodchikov \cite{ZalmolodchikovSMatrixSolitonAntiSolitonSineGordon} and this was extended to the full $\op{S}$-matrix
by Karowski and Thun \cite{KarowskiThunCompleteSMatrixThirring}. Nowadays, $\op{S}$-matrices of a large class of  massive integrable quantum field theories in 1+1 dimensions
have been constructed \cite{ArinshteinFateyevZamolodchikovSMatrixTodaChain,ZalZalBrosFactorizedSMatricesIn(1+1)QFT}.

The $\op{S}$-matrix plays a central role in the bootstrap approach. More precisely, within this approach one constructs the quantum field operators
as linear operators on the model's Fock space built over the set of asymptotic particles. These operators can be described by the generalised functions describing their
distributional integral kernels. These kernels are built from one central object: the operator's form factors. On physical grounds, these are interpreted as matrix
elements of the given quantum field operator taken between the vacuum and some asymptotic particle excited state. The form factors are meromorphic functions
on the physical strip 
\begin{equation}
	\label{def physical strip}
	\msc{S}_{\e{phys}} = \Big\{   z \in \mathbb{C} \, : \, 0 < \Im(z) < \pi \Big\}
\end{equation}
attached to the rapidities of the asymptotic particles and satisfy vector valued coupled Riemann--Hilbert problems in $0, 1, \dots$
variables. This Riemann--Hilbert problem constitutes the founding axiom of the bootstrap program. A given model's $\op{S}$-matrix arises as
one of the data on which the Riemann--Hilbert problem is formulated. We point out that, for theories having a scalar $\op{S}$-matrix, the bootstrap program
Riemann-Hilbert problem was directly recovered by Bostelman and Cadamuro \cite{BostelmannCadamuroFFEqnsInIQFTsWithScalarS}
from the algebraic construction of such theories \cite{LechnerAlgebraicConstructionIQFTSscalarS}.

Over the years, significant efforts have been made to devise techniques allowing
one to solve explicitly the mentioned Riemann--Hilbert problem using its underlying integrability.
First progress has been achieved by Kirillov and Smirnov \cite{KirillovSmirnovFirstCompleteSetBootstrapAxiomsForQIFT,KirillovSmirnovUseOfBootstrapAxiomsForQIFTToGetMassiveThirringFF}
relatively to the Massive Thirring model. Representations of form factors for various operators in the Sine-Gordon model
were obtained
\cite{SmirnovIntegralRepSolitonFFSineGordonBootstrap,SmirnovProofIntegralRepSolitonFFSineGordonBootstrapFollowUpJPhysAPaper,SmirnovReductionsAndClusterPropertyInSineGordonPlusSomeDiscussionsUVLimit},
see also the monograph \cite{SmirnovFormFactors}. In particular, results were given for the form factors of $\ex{\i \f{g}{2} \Psi}$, where $g$
is the coupling constant arising in \eqref{ecriture Lagrangien Sine-Gordon}.
Those works developed techniques allowing one to represent form factors in terms of highly combinatorial expressions involving
entries of an associated six-vertex like monodromy matrix and ingeniously constructed special functions.
More compact and structured expressions for the form factors could have been obtained for models having scalar scattering matrices (Lee-Yang, Sinh-Gordon).
In that case, there appears a significant reduction in the complexity of the
bootstrap equations for the associated form factors. This allowed Zamolodchikov to produce determinant based representations \cite{ZalmolodchikovTwoPointFctsLeeYang} for the form factors,
an approach later generalised in \cite{FringMussardoSimonettiFFSOmeLocalObsSinhGordon,KoubekMussardoFFForMoreOpInSinhGordon}.

Subsequently, various efforts have been made to develop new techniques allowing one to produce much more manageable representations for the form factors.
Here, one should definitely mention the free field, aka. angular quantization, approach introduced by Lukyanov \cite{LukyanovFirstIntroFreeField}
which provided a formalism allowing one to construct the form factors using representation theory. This led to rather very simple
expressions for the form factors of the exponentials of the field operators  $\ex{\i \f{g}{2} \ga \Psi}$, $\ga$ generic,
in the Sinh-Gordon and the Bullough-Dodd models \cite{BrazhnikovLukyanovFreeFieldRepMassiveFFIntegrable}. Furthermore, the free field
techniques also produced \cite{LukyanovConjectureFFExponentialFieldSineGSolASolAndBReather} explicit expressions for the first
few soliton/antisoliton form factors of the quantum Sine-Gordon
operators $\ex{\i \f{g}{2} \ga \Psi}$, while providing all the necessary setting allowing one to also compute the form factors pertaining to
higher soliton/antisoliton excitation sectors. This was later achieved in \cite{PalmaiMultiSolitionFFSineGField}. The resulting expressions were much more explicit
that those obtained earlier on. Later on, it was showed that the free field approach also allows one to produce representations, of a similar kind,
for the form factors of topologically charged, \textit{i.e.} soliton creating, operators in the Sine-Gordon model \cite{LukyanovZamolodchikovSolitonCreatingOpsSineG}.
An alternative construction of form factors was proposed in \cite{BabujianFringKarowskiZapletalExactFFSineGordonBootsstrapI} and further developed in
\cite{BabujianKarowskiExactFFSineGordonBootsstrapII}, ultimately providing form factors of local operators closely connected with the conserved quantities of the theory,
\textit{e.g.} the energy-momentum tensor. The idea of the approach was to use a different basis
for the decomposition of the form factors and to introduce a certain integral transform which eventually allows one to
simplify considerably the form of the coupled Riemann--Hilbert problem one starts with.
All-in-all,  this alternative approach led to rather compact multiple integral-based representations for the form factors.
We shall refer to this technique as the $\mc{K}$-transform method. The approach was also conformed to produce form factors of $\ex{\i \f{g}{2} \ga \Psi}$ in the quantum Sinh-Gordon
model
\cite{BabujianKarowskiBreatherFFSineGordon}, providing eventually the same expressions as those obtained by the free field approach \cite{BrazhnikovLukyanovFreeFieldRepMassiveFFIntegrable}.
We would like to point out that one of the advantages of the $\mc{K}$-transform construction is that
it produces a representation for the form factors in terms of integrals (be it versus a discreet or a continuous measure)
only involving simple fully explicit functions which are symmetric in the integration variables arising in the representation.
We stress that such an explicit symmetric representation is an important advantage of the approach relatively to the possibility of
analysing various other properties of the model, in particular to proving the convergence of form factor expansions for two-point, and possibly multi-point,
functions. Indeed existing concentration of measure techniques are only efficient in estimating, at large $N$, $N$-fold integrals with
explicitly symmetric integrands. We refer to the work \cite{KozConvergenceFFSeriesSinhGordon2ptFcts}.
Finally, more recently, another construction of the Sine-Gordon form factors was proposed in \cite{JimboMiwaSmirnovFormFactorsSineGNewCOnstructionViaHiddenGrassmann}
and built on the use of ingenious auxiliary free fermionic operators, aka. hidden Grassmann structure of the XXZ chain.

The aim of this work is to develop further the $\mc{K}$-transform formalism for the calculation of the Sine-Gordon form factors.
It was observed in \cite{LukyanovConjectureFFExponentialFieldSineGSolASolAndBReather} (although not stated explicitly) that
$\ex{\i \f{g}{2} \ga \Psi}$ has a non-trivial matrix valued mutual locality index \cite{YurovZalmolodchikovFirstIntroMutualLocalityIndex}
in respect to the asymptotic soliton/antisoliton fields. This modifies the form of the bootstrap equations, see \textit{e.g.}
\cite{JimboMiwaSmirnovFormFactorsSineGNewCOnstructionViaHiddenGrassmann}. We conform the $\mc{K}$-transform to solving this alternative
form of the bootstrap equations. Remarkably, the simplest solution that we obtain produces the form factors of   $\ex{\i \f{g}{2} \ga \Psi}$
and of the associated charged operators. The expressions we find are of fundamentally different nature that those obtained within the free field or the
hidden Grassmann structure based approaches. Still, after several handlings, we show that these coincide with the results obtained earlier in the litterature:
the soliton/antisoliton form factors of  $\ex{\i \f{g}{2} \ga \Psi}$ \cite{LukyanovConjectureFFExponentialFieldSineGSolASolAndBReather},
and the $n$-solitonic along with the $n+1$-solitonic/1-antisolitonic  form factors of the associated charged operator proposed in
 \cite{LukyanovZamolodchikovSolitonCreatingOpsSineG}. Furthermore, we demonstrate how to recover the well-known form of the
$n$ soliton/antisoliton form factors of $\ex{\i \f{g}{2} \ga \Psi}$ at the free fermion point of the theory \cite{SchroeTruongFFExpPhyInFreeFermionPtSineG,SmirnovFormFactors}.

The paper is organized as follows. In Section \ref{Section Programme Bootstrap},
we recall the expression of the $\op{S}$-matrix of the Sine-Gordon model, and the construction of quantum field operators within
the form factor based approach. Then, we formulate the twisted bootstrap equations satisfied by the form factors, first in the
purely soliton/antisoliton sector of the theory and then, in full generality, in all sectors of the theory including the presence
of bound states (breathers) for $g$ small enough.  In particular, form factors involving breather contributions are obtained as residues of the
soliton/antisoliton ones at specific points corresponding to fusion angles.
Section \ref{Section Main Results} gathers the main results of this work. Off-shell Bethe Ansatz based representations solving $\gamma$-deformed
form factor bootstrap equations in the soliton/antisoliton sector are presented in Sub-Section \ref{SubSection FF for solitons antisolitons}.
The formulae valid throughout the full Hilbert space of the theory, \textit{i.e.} for the form factors mixing breathers and solitons/antisolitons,
are gathered in Subsection \ref{SubSection FF Full sector with breathers}. Finally, the specialisation of these formulae to the type $1$ breather sector is
outlined in Subsection \ref{Subsection General FF rerstriction to type 1 breathers}.
Lastly, in Section \ref{Section FF exp champ}, we discuss a particular choice of $p$ functions arising in the representations given in Section \ref{Section Main Results} which
gives rise to form factors of the  exponentials of the Sine-Gordon field $\op{O}^{(\ga)} = \ex{\i \f{g\ga}{2}\Psi} $ and its soliton creating analogues.
We compare the resulting form
factors with the literature \cite{BabujianKarowskiBreatherFFSineGordon,BrazhnikovLukyanovFreeFieldRepMassiveFFIntegrable,
BernardLeclairDiffEqnForPIII,LukyanovConjectureFFExponentialFieldSineGSolASolAndBReather,LukyanovZamolodchikovSolitonCreatingOpsSineG},
showing full agreement with the earlier results on the matter, this up to normalization.
The Appendices gather several technical results. Appendix \ref{Appendice Special Functions} recalls the definition of the
special functions used throughout the paper. In particular Sub-Appendix \ref{Appendice SousSection Quantum dilog} reviews all of the useful properties of the
quantum dilogarithm. Appendix \ref{Appendice Preuve des eqns de FF} contains various details relative to the proofs of the bootstrap equation.
Sub-Appendix \ref{Appendice SS section eqn symmetrie i} gives details about the proof of the $\op{S}$-matrix permutation property,
Sub-Appendix \ref{Appendice SS section eqn Watson ii} establishes Watson's equation,
Sub-Appendix \ref{Appendice SS section eqn inv Lorentz iii} establishes the Lorentz boost property while
Sub-Appendix \ref{subsection axiom iv} establishes the kinematic pole axiom. Appendix \ref{Appendix preuve decomposition breathers} contains the derivation of the expression for the breather/soliton-antisoliton form factors, \textit{i.e.} Proposition \ref{Proposition FF dans le secteur breather}.
Further, Appendix \ref{appendix vanish} provides some of the technical details that arise in the course of proving the reduction to
the expressions at the free fermion point. Finally, Appendix \ref{Appendice Resultats auxiliaires} gathers some technical properties
satisfied by various combinations of the operator entries of the monodromy matrix.

		\section{The bootstrap program for the Sine-Gordon model}
\label{Section Programme Bootstrap}
	
	\subsection{The Sine-Gordon model and its $\op{S}$-matrix}

It is convenient to reparameterise the Sine-Gordon coupling constant $g \in \intff{0}{\sqrt{8\pi} }$ \eqref{ecriture Lagrangien Sine-Gordon}
in terms of  $\xi \in \R^{+}$ defined as
	\beq
	\xi \, =  \, \f{ g^2 \pi   }{ 8\pi - g^2} \;. 
	\enq

The quantum field theory in infinite volume associated with the above Lagrangian density contains several asymptotic particles
whose number and kind depends on $\xi$. The model strongly simplifies when $\tf{\pi}{\xi} \in \mathbb{N}$ as then its scattering becomes purely diagonal.
However, despite this apparent simplification, this case demands a quite delicate analysis. We shall thus not consider this situation directly
but rather understand it as a limit, to be taken on the level of the final formulae. Hence, we shall always focus on the case of generic coupling constant when
$\tf{\pi}{\xi}\not\in \mathbb{N}$.  We remind that the model admits two different regimes.
\begin{itemize}
\item The repulsive regime: $\xi>\pi$.  Then the model contains only two
species of asymptotic particles: the soliton ($\mf{s}$) and the antisoliton ($\ov{\mf{s}}$). These particles have equal mass $M$ and are charge conjugates of each other.
Since these particles have equal masses, integrability is compatible with the existence of a non-trivial scattering.

\item The attractive regime: $0<\xi<\pi$. On top of the soliton/antisoliton asymptotic degrees of freedom mentioned above, the model
contains $n_{\xi} = \lfloor \tf{\pi}{\xi} \rfloor$ additional species of asymptotic particles called breathers.
These degrees of freedom will henceforth be denoted as $\mf{b}_1,\dots, \mf{b}_{ n_{\xi} }$.
The breather $\mf{b}_k$ may be interpreted \cite{KarowskiThunCompleteSMatrixThirring} as a soliton/antisoliton bound state, but also as a $\mf{b}_{m}$/$\mf{b}_{\ell}$ bound state,
with $m+\ell=k$. The breathers $\mf{b}_{k}$ have mass $M_k = 2M\sin \left( \frac{k\xi}{2}\right)$. Since all breathers have different masses
and their masses also differ from the one of the soliton and antisoliton, they all have a scalar scattering between each other as well as between the
the solitons  or antisolitons.
\end{itemize}

At given value of the coupling constant $\xi$, the total number of degrees of freedom will be denoted
\beq
N_{\xi} \, = \,  2 \, + \, n_{\xi} \,, \quad n_{\xi} \, = \,  \lfloor \tfrac{\pi}{\xi} \rfloor  \; , \qquad \e{with} \qquad \f{\pi}{\xi}\not\in \mathbb{N} \;.
\label{definition N xi n xi etc}
\enq
It is then useful to gather these internal degrees of freedom into the vector space  $\mf{h}_{\e{tot}} \simeq \Cx^{ N_{\xi} }$ endowed with the canonical basis
$\op{v}^{\mf{s}},\op{v}^{\ov{\mf{s}}},\op{v}^{\mf{b}_1},\cdots, \op{v}^{\mf{b}_{ n_{\xi} } }$. This indexing of vectors will be
also used to label the matrix entries. The Hilbert space associated with the asymptotic particle content of the model is given by the below Fock space
\beq
\mf{h}_{\textrm{SG}}\, = \, \bigoplus\limits_{n =  0}^{+\infty} L^2\Big(\R_{>}^n ,  \mf{h}_{\e{tot}}^{\otimes n} \Big)
\, , \qquad \e{with}  \qquad \R_{>}^n \, = \, \Big\{ \bs{\be}_n = \big(\be_1,\dots,\be_n\big) \in \R^n \, : \, \be_1>\dots > \be_n  \Big\} \;.
\label{definition espace Hilbert Sine G}
\enq

The integrability of the model ensures that the scattering preserves the rapidities of the particles and that it can be seen as
a concatenation of two-body processes in which conservation of momenta holds. Thus, the two-body $\op{S}$-matrix describing the scattering of asymptotic particles with
rapidities $\be_1$ and $\be_2$ only depends on the difference $\th=\be_1-\be_2$ and is thus a matrix
$\op{S}_{\e{tot}}(\th): \mf{h}_{\e{tot}}\otimes\mf{h}_{\e{tot}}\rightarrow  \mf{h}_{\e{tot}}\otimes \mf{h}_{\e{tot}}$.
In the basis
\bem
\bigg\{ \op{v}^{\mf{s}}\otimes \op{v}^{\mf{s}} ,\op{v}^{\mf{s}}\otimes \op{v}^{\ov{\mf{s}}} ,\op{v}^{\ov{\mf{s}}}\otimes \op{v}^{\mf{s}} ,\op{v}^{\ov{\mf{s}}}\otimes \op{v}^{\ov{\mf{s}}} ,
\op{v}^{\mf{s}}\otimes  \op{v}^{\mf{b}_1}, \op{v}^{\ov{\mf{s}}}\otimes  \op{v}^{\mf{b}_1},  \op{v}^{\mf{b}_1}\otimes \op{v}^{\mf{s}}  , \op{v}^{\mf{b}_1}\otimes \op{v}^{\ov{\mf{s}}}
, \cdots,   \op{v}^{\mf{b}_{ n_{\xi} }} \otimes \op{v}^{ \ov{\mf{s}} },  \\
\op{v}^{ \mf{b}_1 } \otimes  \op{v}^{ \mf{b}_1 } , \op{v}^{ \mf{b}_2 } \otimes  \op{v}^{ \mf{b}_1 },  \cdots,
\op{v}^{ \mf{b}_{ n_{\xi} } } \otimes  \op{v}^{ \mf{b}_{ n_{\xi} }  }  \bigg\} \; ,
\end{multline}
$\op{S}_{\e{tot}}(\th)$ takes a block diagonal form:
\beq
\op{S}_{\e{tot}}(\th) \; = \; \left( \ba{ccc}  \op{S}(\th) & 0   \\
0      & \op{S}_{\e{diag}}(\th)
\ea \right)
\label{ecriture matrice S totale}
\enq
where
\beq
\op{S}_{\e{diag}}(\th) \, = \, \e{Diag}\Big(  S_{\mf{s},\mf{b}_1}(\th),  S_{\ov{\mf{s}},\mf{b}_1}(\th),  S_{\mf{b}_1,\mf{s}}(\th),  S_{\mf{b}_1,\ov{\mf{s}}}(\th),
\dots , S_{ \mf{b}_{ n_{\xi}},\ov{\mf{s}} } (\th) ,
S_{\mf{b}_1, \mf{b}_1}(\th), \cdots,  S_{ \mf{b}_{ n_{\xi} }, \mf{b}_{ n_{\xi} } } (\th)   \Big)
\enq
gathers all the diagonal scattering while $\op{S}(\th)$ only describes the soliton-antisoliton scattering
\cite{KarowskiThunCompleteSMatrixThirring,ZalmolodchikovSMatrixSolitonAntiSolitonSineGordon}. We shall now describe these objects in detail.

\subsubsection{The soliton/antisoliton $\op{S}$-matrix}

The soliton/antisoliton $\op{S}$-matrix was first obtained in \cite{ZalmolodchikovSMatrixSolitonAntiSolitonSineGordon}. It will appear useful
to interpret $\op{S}$ as the matrix of an operator acting on the tensor product space $\mf{h}\otimes \mf{h}$, $\mf{h}\simeq \Cx^2$.
It takes the form
\beq
\label{S matrix SG}
\op{S}(\th) \; = \; \left(\ba{cccc}  a(\th) & 0 & 0 & 0  \\ 
0  & b(\th)  & c(\th) & 0  \\ 
0 &  c(\th) &  b(\th)  &  0  \\ 
0 &  0      &   0       &   a(\th)  \ea \right) \;. 
\enq
There, are simply related to $a$
\beq
b(\th) \; = \; \f{ \sinh \big[\tfrac{\pi \th }{\xi} \big]  }{ \sinh\big[ \tfrac{\pi }{\xi} \big(\i\pi - \th  \big) \big] } a(\th) \quad , \quad
c(\th) \; = \; \f{ \sinh \big[ \tfrac{ \i \pi^2 }{\xi} \big] }{ \sinh\big[ \tfrac{\pi }{\xi} \big(\i\pi - \th  \big) \big]  } a(\th) \;.
\enq
In its turn, $a$ admits an integral representation   \cite{FaddeevIrreducibiliteModularDouble,
Kashaev3termIntegralRelationDfcts,KurokawaDoubleSineIntro,RuijsenaarsFirstIntroQuantumRelatToda,WoronowiczQuantumExpFcts}. In fact, it turns out that
the latter may also be expressed in terms of the quantum dilogarithm
\beq
a(\th) \, = \,  - \exp \Bigg\{- \i \Int{0}{+\infty}  \f{\dd t }{ t } \cdot \f{ \sin(\th t) \sinh\big[ \tfrac{\pi-\xi}{2} t \big] }{ \sinh\big[ \tfrac{\xi t}{2} \big] \cosh\big[ \tfrac{\pi t }{2} \big]  }   \Bigg\}
\; = \; - \varpi\Bigg( \ba{cc}  \th +\i \tfrac{\xi}{2} \, , &   \th - \i \tfrac{\xi}{2}  \\ 
\th +\i (\pi - \tfrac{\xi}{2}) \, , &   \th - \i (\pi -\tfrac{\xi}{2})    \ea \Bigg) \;. 
\enq
Above, $\varpi$ stands for the quantum dilogarithm with periods $\xi$ and $2\pi$, see
Appendix \ref{Appendice SousSection Quantum dilog} for more details. The formula is written by using
hypergeometric notations for ratios of products, \textit{viz}.
\beq
\varpi\Bigg( \ba{c} a_1,\dots, a_n \\ b_1,\dots, b_m \ea \Bigg) \; = \; \f{ \pl{k=1}{n} \varpi(a_k)  }{  \pl{k=1}{m} \varpi(b_k)  } \;. 
\enq
The specific choice of periods leads to the following first order shift identities
\beq
\varpi(\la + 2\i\pi) \; = \; 2\i \sinh\Big[ \tfrac{\pi}{\xi} \big( \la + \i [\pi - \tfrac{\xi}{2}] \big) \Big] \cdot \varpi(\la) \qquad \e{and} \qquad
\varpi(\la + \i\xi) \; = \; 2\i \sinh\Big[ \tfrac{1}{2} \big(\la - \i [\pi - \tfrac{\xi}{2}] \big) \Big] \cdot \varpi(\la) \;. 
\enq
The soliton-antisoliton $\op{S}$-matrix enjoys unitarity and crossing symmetry
\beq
\op{S}^{\dagger}(\th) \, = \, \op{S}^{-1}(\th) \, = \, \op{S}(-\th) \qquad \e{and} \qquad
\op{S}(\th) \, = \, \big(\sg^x \otimes \e{id} \big) \op{S}(\i\pi -\th)^{\op{t}_2} \big(\sg^x \otimes \e{id} \big) \;,
\label{ecriture unitarite crossing pour partie soliton antisoliton S matrix}
\enq
where $\op{t}_2$ refers to taking the transposition $\op{t}$ only with respect to the second copy of $\mf{h}$ in the tensor product space $\mf{h}\otimes \mf{h}$. These symmetries are present in any reasonable quantum field theory. On top of these $\op{S}$ also satisfies the celebrated Yang-Baxter equation
\cite{BaxterPartitionfunction8Vertex-FreeEnergy,YangFactorizingDiffusionWithPermutations} which translates
the compatibility of decomposition schemes entering in the factorisability of the multi-particle scattering into two-body process.
It is one of the benchmarks of integrability. To write down this equation, we need to introduce appropriate tensor notations.
Consider the triple tensor product of copies of $\mf{h}$: $\mf{h}_1\otimes \mf{h}_2\otimes \mf{h}_3$, with $\mf{h}_a \simeq \mf{h}$.
Let $\op{S}_{12}(\th) \, = \, \op{S}(\th) \otimes \e{id}$, $\op{S}_{23}(\th) \, = \, \e{id} \otimes \op{S}(\th)$
and $\op{S}_{13}(\th) \, = \, \big( \op{P}\otimes \e{id} \big) \op{S}_{23}(\th)   \big( \op{P}\otimes \e{id} \big)$ with $\op{P}: \mf{h} \otimes \mf{h} \tend  \mf{h} \otimes \mf{h}$
being the permutation operator $\op{P}(\bs{v} \otimes \bs{w}) \,  = \, \bs{w} \otimes \bs{v}$.
Within these notations, it holds
\beq
\op{S}_{12}(\la-\mu)\op{S}_{13}(\la-\nu)\op{S}_{23}(\mu-\nu) \; = \;\op{S}_{23}(\mu-\nu)  \op{S}_{13}(\la-\nu) \op{S}_{12}(\la-\mu) \;.
\label{equation de Yang-Baxter pour S}
\enq
 The above equation should  be understood as an equality
on $\mc{L}\big( \mf{h} \otimes \mf{h} \otimes \mf{h} \big)$.

It will sometimes be useful to use the representation $\op{S}(\th) \; = \;a(\th) \, \wt{\op{S}}(\th)  $ where
\beq
\wt{\op{S}}(\th) \; = \; \left(\ba{cccc}  1 & 0 & 0 & 0  \\
0  & \wt{b}(\th)  & \wt{c}(\th) & 0  \\
0 &  \wt{c}(\th) &  \wt{b}(\th)  &  0  \\
0 &  0      &   0       &   1  \ea \right) \;,
\label{definition matrice S tildee}
\enq
in which
\beq
\wt{b}(\th) \; = \; \f{ \sinh \big[\tfrac{\pi \th }{\xi} \big]  }{ \sinh\big[ \tfrac{\pi }{\xi} \big(\i\pi - \th  \big) \big] }   \quad , \quad
\wt{c}(\th) \; = \; \f{ \sinh \big[ \tfrac{ \i \pi^2 }{\xi} \big] }{ \sinh\big[ \tfrac{\pi }{\xi} \big(\i\pi - \th  \big) \big]  }  \;.
\label{definition tilde b et c}
\enq

Finally, we shall also require to use a twisted variant of the soliton-antisoliton
$\op{S}$-matrix:
\beq
\op{S}^{(\ga)}(\th) \, = \, \ex{ \f{\ga\th}{2} \sg^{z}\otimes \e{id}  } \op{S}(\th) \ex{ -\f{\ga\th}{2}  \sg^{z} \otimes \e{id} } \, = \,
\left(\ba{cccc}  a(\th) & 0 & 0 & 0  \\
0  & b(\th)  & c^{(\ga)}(\th) & 0  \\
0 &  c^{(-\ga)}(\th) &  b(\th)  &  0  \\
0 &  0      &   0       &   a(\th) \ea \right)  \;.
\enq
There, we have introduced
\beq
c^{(\pm \ga)}(\th) \, = \,  \ex{\pm \ga \th }\,  c(\th) \;.
\enq
Also, analogously to the previous notations, we agree that
\beq
\wt{\op{S}}^{(\ga)}_{k\ell}(\th)  \, = \, \f{1}{ a(\th) }\op{S}^{(\ga)}_{k\ell}(\th) \, = \,
\left(\ba{cccc}  1 & 0 & 0 & 0  \\
0  & \wt{b}(\th)  & \wt{c}^{\,(\ga)}(\th) & 0  \\
0 &  \wt{c}^{\, (-\ga)}(\th) &  \wt{b}(\th)  &  0  \\
0 &  0      &   0       &   1 \ea \right)_{[k\ell]} \;.
\enq
Note that it follows readily from the Yang-Baxter equation \eqref{equation de Yang-Baxter pour S} satisfied by the soliton/antisoliton $\op{S}$-matrix
that the $\ga$-twisted $\op{S}$-matrix also satisfies the Yang-Baxter equation
\beq
\op{S}_{12}^{(\ga)}(\la-\mu)\op{S}_{13}^{(\ga)}(\la-\nu)\op{S}_{23}^{(\ga)}(\mu-\nu) \; = \;\op{S}_{23}^{(\ga)}(\mu-\nu)  \op{S}_{13}^{(\ga)}(\la-\nu) \op{S}_{12}^{(\ga)}(\la-\mu) \;.
\label{equation de Yang-Baxter pour S twiste}
\enq

\subsubsection{The soliton/breather and breather/breather $\op{S}$-matrices}

We now write down all of the remaining, diagonal, entries of the total $\op{S}$-matrix. These were first obtained in \cite{KarowskiThunCompleteSMatrixThirring}.
The soliton and antisoliton/breather scattering take the explicit form
\beqa
\label{scattering sol brea}
S_{\ups, \mf{b}_n}(\th) \, = \, S_{\mf{b}_n,\ups}(\th)
%
%
%
& =&(-1)^n \exp\Bigg\{ -2 \i \Int{0}{+\infty} \f{\dd t }{t} \f{ \cosh\big[\tfrac{\xi t }{2 } \big]   \sinh\big[\tfrac{\xi n t }{2 } \big]  }
{ \cosh\big[\tfrac{\pi t }{2 } \big]   \sinh\big[\tfrac{\xi t }{2 } \big] } \sin\big[ \th t \big] \Bigg\} \\
&=& \f{ \tanh\big[\tfrac{1}{2} ( \th + \i \tfrac{\pi}{2} - \i \tfrac{n \xi }{2}) \big]   }{ \tanh\big[\tfrac{1}{2} ( \th - \i \tfrac{\pi}{2} + \i \tfrac{n \xi }{2}) \big] }
\cdot \pl{p=1}{n-1} \tanh^2\big[\tfrac{1}{2} ( \th + \i \tfrac{\pi}{2} + \i (p-\tfrac{n}{2}) \xi) \big] \;. 
\eeqa
%
%
%
We stress that the result does not depend on $\ups \in \{ \mf{s}, \ov{\mf{s}} \}$.
Finally, the breather/breather scattering reads
\beqa
\label{scattering brea brea}
S_{\mf{b}_n, \mf{b}_m}(\th)= S_{\mf{b}_m, \mf{b}_n}(\th) & = &  \exp\Bigg\{ -4 \i \Int{0}{+\infty} \f{\dd t }{t} \f{ \cosh\big[\tfrac{\xi t }{2 } \big]   \sinh\big[\tfrac{\xi n t }{2 } \big]  \cosh\big[\tfrac{\pi - m \xi   }{2 } t \big]  }
{ \cosh\big[\tfrac{\pi t }{2 } \big]   \sinh\big[\tfrac{\xi t }{2 } \big] } \sin\big( \th t \big) \Bigg\} \\
&=& \f{ \tanh\big[\tfrac{1}{2} ( \th  + \i \tfrac{n +m }{2}\xi ) \big]   }{ \tanh\big[\tfrac{1}{2} ( \th - \i \tfrac{n +m }{2}\xi ) \big] }
\cdot \pl{p=1}{n-1}  \f{ \tanh\big[\tfrac{1}{2} ( \th  + \i \tfrac{n +m -2p}{2}\xi ) \big]   }{ \tanh\big[\tfrac{1}{2} ( \th - \i \tfrac{n +m -2p }{2}\xi ) \big] }  
\cdot \pl{p=1}{m-1}  \f{ \tanh\big[\tfrac{1}{2} ( \th  + \i \tfrac{n +m -2p}{2}\xi ) \big]   }{ \tanh\big[\tfrac{1}{2} ( \th - \i \tfrac{n +m -2p }{2}\xi ) \big] } \;. 
\nonumber
\eeqa
The integral representation given in the first line of \eqref{scattering brea brea} is only well defined when $n<m$ and $\th \in \R^{*}$.
However, the finite product representation which follows from it,
already holds for all $n$ and $m$ and $\th$ away from the poles. Furthermore, the product formula is explicitly invariant under the replacement $n \leftrightarrow m$.
Again, the soliton-antisoliton/breather and the breather/breather $\op{S}$-matrices also satisfy unitarity and crossing symmetry (breathers are their own anti-particles):
\beq
S_{\tau, \mf{b}_m}(\th) \,  S_{\tau, \mf{b}_m}( - \th ) \, = \, 1  \qquad \e{and} \qquad S_{ \tau , \mf{b}_m}(\th) \, = \,  S_{\tau, \mf{b}_m}(\i\pi -\th)
\quad \e{with} \quad \tau \in \big\{ \mf{s}, \ov{\mf{s}}, \mf{b}_1,\dots, \mf{b}_{n_{\xi}}  \big\}  \;.
\label{ecriture unitarite crossing pour partie scalaire S matrix}
\enq

It is important to note that the expressions \eqref{scattering sol brea} and \eqref{scattering brea brea} are both obtained from the residues of the regular soliton-antisoliton
scattering matrix \eqref{S matrix SG} by the so-called fusion procedure whose details were given in \cite{KarowskiThunCompleteSMatrixThirring}.
To state the latter, it is convenient to introduce the canonical basis of $\Cx^2$:
\beq
\op{v}^+ \, = \, \left( \ba{c} 1 \\ 0 \ea \right) \quad \e{and} \quad \op{v}^- \, = \, \left( \ba{c} 0 \\ 1 \ea \right) \;,
\enq
which, in view of what will follow, should be interpreted as the rank $2$ projection of $\op{v}^{\mf{s}},  \op{v}^{\ov{\mf{s}}}$ onto $\Cx^2$ whose kernel is spanned by
the vector from the breather sector. Further, let
\beq
\bs{\vp}^{(\mf{b}_k)}\;=\; \f{ 1 }{ \sqrt{2} } \Big\{ \op{v}^+ \otimes \op{v}^- \, +\, (-1)^k    \op{v}^- \otimes \op{v}^+  \Big\}
\enq
A direct calculation shows that $\op{S}$ admits simple poles in the physical strip \eqref{def physical strip} at $\i \mf{u}_k$ with
\beq
\mf{u}_k\, = \, \pi - k \xi\, , \quad k=1,\dots, n_{\xi} \,
\enq
whose matrix residues take the form
\beq
\e{Res}\Big( \op{S}(\th) \, \cdot \, \dd \th , \th = \i \mf{u}_k \Big) \; = \; \f{ R_k }{ \th - \i \mf{u}_k  } \, \bs{\vp}^{(\mf{b}_k)} \cdot  \big( \bs{\vp}^{(\mf{b}_k)} \big)^{\op{t} } \;,
\label{ecriture equation pour poles matrices S}
\enq 
with
\beq
R_k \, = \, - \f{2 \xi}{\pi} a(\i \mf{u}_k ) \sinh\Big( \i \f{\pi^2}{\xi}\Big)\; =\;
(-1)^k 4 \i \cot\Big( \f{ k \xi }{ 2 } \big) \cdot \pl{p=1}{k-1} \cot^2 \Big( \f{p \xi }{ 2 } \Big) \; .
\enq
One can connect these residues to the soliton/antisoliton-breather $\op{S}$-matrix \eqref{scattering sol brea} through the formula
\beq
\bs{\vp}^{(\mf{b}_k)}_{23} \cdot \big( \bs{\vp}^{(\mf{b}_k)}_{23} \big)^{\op{t}} \,  \op{S}_{13}\big( \th + \i \tfrac{\mf{u}_k}{2} \big)
\,  \op{S}_{12}\big( \th - \i \tfrac{\mf{u}_k}{2} \big) \; = \;
 \op{S}_{12}\big( \th - \i \tfrac{\mf{u}_k}{2} \big)
 \,  \op{S}_{13}\big( \th + \i \tfrac{\mf{u}_k}{2} \big) \, \bs{\vp}^{(\mf{b}_k)}_{23} \cdot \big( \bs{\vp}^{(\mf{b}_k)}_{23} \big)^{\op{t}}
\, = \, S_{\! \ups,  \mf{b}_k}(\th) \,   \bs{\vp}^{(\mf{b}_k)}_{23} \cdot \big( \bs{\vp}^{(\mf{b}_k)}_{23} \big)^{\op{t}}\;,
\enq
with $\ups \in \big\{ \mf{s}, \ov{\mf{s}} \big\}$. Moreover, one also has
\beq
 S_{\! \mf{b}_n,  \mf{b}_m}(\th)  \, = \,  S_{\! \ups,  \mf{b}_m}\big( \th + \i \tfrac{\mf{u}_n}{2} \big)  \,
 S_{\! \ups,  \mf{b}_m}\big( \th - \i \tfrac{\mf{u}_n}{2} \big) \;.
\enq

\vspace{3mm}

All of the unitarity, crossing symmetry and Yang-Baxter equations satisfied by the individual $\op{S}$-matrices outlined above
may be lifted to the level of the total $\op{S}$-matrix as given in \eqref{ecriture matrice S totale}.
Indeed, it is a matter of a direct check that
\beq
\op{S}_{\e{tot}}^{\dagger}(\th) \, = \,  \op{S}_{\e{tot}}^{-1}(\th) \, = \, \op{S}_{\e{tot}}(-\th)  \quad
\e{and} \quad
\op{S}_{\e{tot}}(\th)  \, = \, \big( \mf{S}_{1}^x \otimes \e{id} \big)  \op{S}_{\e{tot}}(\i \pi - \th)^{\op{t}_2} \big( \mf{S}_{1}^x \otimes \e{id} \big)
\enq
where $\mf{S}^x=\e{Diag}\big( \sg^x, \op{I}_{n_{\xi}} \big)$ is block diagonal.

Further, by understanding the space indices as now referring to spaces $\mf{h}_{\e{tot}}$ occurring in the triple tensor product
$\mf{h}_{\e{tot}}\otimes \mf{h}_{\e{tot}} \otimes \mf{h}_{\e{tot}}$, one has the Yang-Baxter equation
\beq
\big(\op{S}_{\e{tot}} \big)_{12}(\la-\mu)\big(\op{S}_{\e{tot}} \big)_{13}(\la-\nu)\big(\op{S}_{\e{tot}} \big)_{23}(\mu-\nu) \; = \;
\big(\op{S}_{\e{tot}} \big)_{23}(\mu-\nu)  \big(\op{S}_{\e{tot}} \big)_{13}(\la-\nu) _{12}(\la-\mu) \;.
\enq
Finally,  upon substituting $\op{v}^{+/-} \hookrightarrow  \op{v}^{\mf{s}/\ov{\mf{s}}}$
in $\bs{\vp}^{(\mf{b}_k)}$ one has the fusion equation on the level of the total $\op{S}$-matrix
\beq
\bs{\vp}^{(\mf{b}_k)}_{23} \cdot \big( \bs{\vp}^{(\mf{b}_k)}_{23} \big)^{\op{t}} \,  \big( \op{S}_{\e{tot}} \big)_{13}\big( \th + \i \tfrac{\mf{u}_k}{2} \big)
\,   \big( \op{S}_{\e{tot}} \big)_{12}\big( \th - \i \tfrac{\mf{u}_k}{2} \big) \; = \;
 \op{S}^{(\mf{b}_k)}_{ 1 }(\th) \,   \bs{\vp}^{(\mf{b}_k)}_{23} \cdot \big( \bs{\vp}^{(\mf{b}_k)}_{23} \big)^{\op{t}}
\enq
with
\beq
 \op{S}^{(\mf{b}_k)}(\th)\, = \, \e{Diag}\Big(  S_{\! \mf{s},  \mf{b}_k}(\th),   S_{\! \ov{\mf{s}},  \mf{b}_k}(\th),
 S_{\! \mf{b}_1,  \mf{b}_k}(\th), \dots,  S_{\! \mf{b}_{n_{\xi}},  \mf{b}_k}(\th)   \Big)   \;.
\label{definition matrice S breather k avec le reste}
\enq

\subsection{The space of operators}

We now describe the space of operators of the Sine-Gordon model, what will allow us to put all the premices necessary for stating the bootstrap program equations.
We start by introducing the translation operator by $\bs{y}=(y_0,y_1)\in \R^{1,1}$. A given vector of the Sine-Gordon Hilbert space $\mf{h}_{\e{SG}}$ \eqref{definition espace Hilbert Sine G}
may be written in terms of its components in each of the Fock space $n$-particle sectors
\beq
\op{f} \, = \, \big( \op{f}^{(0)},\dots, \op{f}^{(n)},\dots \big)\;, \quad \e{with} \quad   \op{f}^{(n)} \in L^2\Big(\R_{>}^n ,  \big( \Cx^{N_{\xi}} \big)^{\otimes n} \Big) \;.
\enq
Then, the translation operator $\op{T}_{\bs{y}}$ acts diagonally on each of the Fock space components:
\beq
\op{T}_{\bs{y}}[\op{f}] \, = \, \Big( \op{T}_{\bs{y}}^{(0)}\big[\op{f}^{(0)}\big],\dots, \op{T}_{\bs{y}}^{(n)} \big[\op{f}^{(n)}\big],\dots \Big)
\enq
where, given $\bs{\be}_n\, = \, \big(\be_1,\dots, \be_n\big)$
\beq
\op{T}_{\bs{y}}^{(n)} \big[ \op{f}^{(n)} \big](\bs{\be}_n) \, = \,  \ex{\i \ov{\bs{P}}(\bs{\be}_n) * \bs{y} } \, \op{f}^{(n)}  (\bs{\be}_n) \quad \e{with} \quad
\ov{\bs{P}}(\bs{\be}_n) \; = \;\sul{k=1}{n} \bs{P}_k(\be_k) \;.
\enq
There are several ingredients to the formula. First $*$ is the Minkowski scalar product $\bs{x}*\bs{y}=x_0y_0-x_1y_1$. Next, the momentum matrix takes the form
\beq
\bs{P}(\be) \, = \,\left(\ba{cc} \bs{p}(\be) \op{I}_2 & 0 \\
											0    &    \e{Diag}\big( \bs{p}_1(\be),\dots,\bs{p}_{n_{\xi}}(\be)  \big)  \ea \right)
\quad \e{while} \quad \bs{P}_k(\be)= \underbrace{\e{id}\otimes \dots \otimes \e{id}}_{k-1 \, \e{times} }
\otimes \bs{P}(\be) \otimes  \underbrace{\e{id}\otimes \dots\otimes \e{id}}_{n-k \, \e{times} } \;.
\enq
Above, $ \bs{p}(\be)=M\big(\cosh(\be),\sinh(\be)\big)$ is the soliton/antisoliton $2$-momentum while  $ \bs{p}_k(\be)=M_k\big(\cosh(\be),\sinh(\be)\big)$ is the $k^{\e{th}}$ breather one.

More generally, given any, possibly unbounded, operator $\op{O}(\bs{x})$ on $\mf{h}_{\e{SG}}$, one may represent\symbolfootnote[2]{In principle, such operators are also generalised
function in the space-time variables $\bs{x}$, and so all formulae below should be understood in the distributional sense, \textit{i.e.} up to smearing by a Schwartz test function.}
its action as $\op{O}(\bs{x})[\op{f}] \, = \, \Big( \op{O}^{(0)}(\bs{x})\big[\op{f}\big],\dots, \op{O}^{(m)}(\bs{x}) \big[\op{f}\big],\dots \Big)$, where
\beq
 \op{O}^{(m)}(\bs{x}) \big[\op{f}\big](\bs{\ga}_m) \; = \; \sul{n \geq 0}{} \op{M}_{ m n}^{(\op{O})}\big[\op{f}^{(n)}\big](\bs{x}\mid \bs{\ga}_m) \;.
\enq
Above,  $\op{M}_{ m n}^{(\op{O})}$ are most conveniently described in terms of their distributional integral kernels
\beq
\op{M}_{ m n}^{(\op{O})}\big[\op{f}^{(n)}\big] \big( \bs{x} \, \mid     \bs{\ga}_m \big)  \; =  \Int{\R^n_{>} }{} \hspace{-1mm}  \f{ \dd^n \be  }{ (2\pi)^n } \;
   \ex{ \i \,  \ov{\bs{P}}( \bs{\ga}_m) * \bs{x} } \,
   \mc{M}^{(\op{O})}_{m;n}\big(     \bs{\ga}_{m}  ; \bs{\be}_{n} \big) \,    \ex{ - \i \,  \ov{\bs{P}}( \bs{\be}_n) * \bs{x} }  \cdot
 \op{f}^{(n)}\big( \bs{\be}_n \big) \;.
\enq
The distributional kernels $ \mc{M}^{(\op{O})}_{m;n}\big(     \bs{\ga}_{m}  ; \bs{\be}_{n} \big)$ are $N_{\xi}^{m}\times N_{\xi}^{n}$ matrices with entries taking values in
generalised functions associated with distributions on  $\mc{S}(\R^{m})\times \mc{S}(\R^{n})$.

The bootstrap program consists in setting up an axiomatic set of equations, whose solutions allow one to construct all of the generalised functions
$ \mc{M}^{(\op{O})}_{m;n}\big(     \bs{\ga}_{m}  ; \bs{\be}_{n} \big) $ describing all of the quantum field operators in the model.
In fact, $ \mc{M}^{(\op{O})}_{m;n}\big(     \bs{\ga}_{m}  ; \bs{\be}_{n} \big) $ for any $m$ can be expressed solely in terms of
the fundamental kernels $ \mc{M}^{(\op{O})}_{0;n}\big(  \emptyset ; \bs{\be}_{n} \big) $ owing to the LSZ bootstrap axiom.
We shall not describe this procedure here, and refer to \cite{KirillovSmirnovFirstCompleteSetBootstrapAxiomsForQIFT,KirillovSmirnovUseOfBootstrapAxiomsForQIFTToGetMassiveThirringFF}
for more details. The general principles of $\op{S}$-matrix theory indicate that fundamental building blocks  $ \mc{M}^{(\op{O})}_{0;n}\big(  \emptyset ; \bs{\be}_{n} \big) $
should correspond to ordered $+$ boundary values
\beq
 \mc{M}^{(\op{O})}_{0;n}\big(  \emptyset ; \bs{\be}_{n} \big) \, = \, \lim_{ \substack{ \eps_1>\dots > \eps_n \\ \eps_a \tend 0^+}  }
 \msc{F}_{n}^{(\op{O})}(\bs{\be}_n+\i \bs{\eps}_n) \; , \quad
 \e{for} \quad \bs{\be}_n\in \R^n_{>}
\label{form factor BV}
\enq
of a meromorphic, linear form on $\mf{h}_{\e{tot}}^{\otimes n}$,  $\mf{h}_{\e{tot}} \simeq \Cx^{N_{\xi}}$ valued, function $\msc{F}_{n}^{(\op{O})}$.
This function can thus be represented as
\beq
\msc{F}_{n}^{(\op{O})}\big( \bs{\be}_n \big) \, = \, \sul{ \bs{\tau}_n \in \mf{X}^n }{} \Big[ \msc{F}_{n}^{(\op{O})}\big( \bs{\be}_n \big)  \Big]^{\bs{\tau}_n}
\big( \op{v}^{\bs{\tau}_n} \big)^{\mathtt{t}}
\qquad \e{with} \qquad  \left\{ \ba{ccc}
\mf{X} & = & \big\{ \mf{s}, \ov{\mf{s}}, \mf{b}_1, \dots, \mf{b}_{n_{\xi}} \big\}  \vspace{2mm} \\
 \op{v}^{\bs{\tau}_n}   & = & \op{v}^{\tau_1} \otimes \cdots \otimes \op{v}^{\tau_{n}} \ea \right.  \;,
\label{ecriture decomposition FF total sur base de vecteurs}
\enq
and $\mathtt{t}$ being the transposition.
The scalar generalised functions $\Big[ \msc{F}_{n}^{(\op{O})}\big( \bs{\be}_n \big)  \Big]^{\bs{\tau}_n}$
are called form factors in the physics litterature and are interpreted as formal matrix elements of the quantum field operator $\op{O}(\bs{0})$
taken between the vacuum and appropriate asymptotic states.

\subsection{The bootstrap equations in the soliton/antisoliton sector}

In this subsection, we will state the bootstrap equations in the soliton/antisoliton sector. Note that, in the repulsive regime $\xi > \pi$,
this constitutes already the full bootstrap, since there are no additional poles present issuing from the breathers.

Still, in order to state the bootstrap equations, it is necessary to introduce a few notations.
First of all, consider $n$-copies of $\Cx^2$: $\mf{h}_a\simeq \Cx^2$, $a=1,\dots,n$. The canonical basis of $\mf{h}_{k}$ will be denoted
\beq
\op{v}^{+}_k=\left( \ba{c} 1 \\ 0 \ea \right) \quad \e{and} \quad \op{v}^{-}_k=\left( \ba{c} 0 \\ 1 \ea \right).
\enq
Upon agreeing on the identification in the labels $+=\mf{s}$ and $-=\ov{\mf{s}}$, one then defines
\beq
\label{form factor coordinates}
\mc{F}_{n}^{(\op{O})}(\bs{\be}_n) \; = \; \sul{ \bs{\veps}_n\in \{\pm\}^n }{} \Big[ \msc{F}_{n}^{(\op{O})}\big( \bs{\be}_n \big)  \Big]^{\bs{\veps}_n}
\cdot \big(\op{v}^{\bs{\veps}_n}  \big)^{\mathtt{t}}  \;,
\qquad \op{v}^{\bs{\veps}_n}   \,  = \,  \op{v}^{\veps_1} \otimes  \cdots \otimes \op{v}^{\veps_{n}}
\enq
where we recall that $\mathtt{t}$ is the transposition. In order to lighten the notations, from now on, we shall drop the number of variables index from the form
factors; the latter may always be inferred from the dimensionality of the vector appearing in the argument, \textit{i.e.},
\beq
\mc{F}^{(\op{O})}(\bs{A}) \, \equiv \,  \mc{F}_n^{(\op{O})}(\bs{A}) \quad \e{for} \quad \bs{A} \in \Cx^n \;.
\enq

%
%
%
Further, with $\op{P}$ the permutation operator, we denote $\op{P}_{ij}$ its embedding into $\mc{L}\big(\mf{h}_1\otimes \cdots \otimes \mf{h}_n\big)$
which permutes the $i^{\e{th}}$ and $j^{\e{th}}$ tensor components  of any vector built as a pure tensor product, \textit{viz}.
\beq
\op{P}_{ij} \, \op{v}^{(1)}\otimes \cdots \otimes  \op{v}^{(n)} \; = \; \op{v}^{(1)}\otimes \cdots \otimes  \op{v}^{(i-1) } \otimes \op{v}^{(j)}\otimes  \op{v}^{(i+1) }\otimes \cdots \otimes \op{v}^{(j-1)}\otimes  \op{v}^{(i)} \otimes 
\op{v}^{(j+1)}\otimes \cdots \otimes  \op{v}^{(n)} 
\enq
for $i<j$ and similarly when $i>j$. 

Generally, given an operator $\op{L}$ on $\Cx^2\otimes \Cx^2$, one denotes by $\op{L}_{ij}$ the embedding of $\op{L}$
into $\mc{L}\big(\mf{h}_1\otimes \cdots \otimes \mf{h}_n\big)$  which acts as $\op{L}$ on the components belonging to the tensor product
$\mf{h}_i\otimes  \mf{h}_j$ and as the identity on all the other components of the full tensor product Hilbert space. More precisely,
\beq
\op{L}_{12}=\op{L}\otimes \underbrace{ \e{id}\otimes \cdots \otimes \e{id} }_{n-2\, \e{times} } \qquad \e{and} \qquad \op{L}_{ij} \, = \, \op{P}_{1i}\, \op{P}_{2j} \,  \op{L}_{12} \,  \op{P}_{1i} \, \op{P}_{2j} \;. 
\enq
We also introduce the canonical basis in $\Cx^n$:
\beq
\bs{e}_k \; = \; \big( \underbrace{0,\dots , 0}_{k-1 \; \e{entries} } , 1 , 0,\dots, 0)^{\mathtt{t}} \in \Cx^n \;,
\enq
and the fundamental strip
\beq
\msc{S}\;= \; \Big\{ z \in \Cx  \, : \, 0 < \Im(z) < 2\pi  \Big\} \;.
\enq
Note that the physical strip $\msc{S}_{\e{phys}}$ only corresponds to $0 < \Im(z) < \pi$.

The bootstrap axioms postulate that the form factors $\mc{F}^{(\op{O})}(\bs{\be}_n)$ belong to the space of
linear form valued functions  such that, for each $k \in \intn{1}{n}$ and fixed $\be_a \in \msc{S}$, $a \not=k$,
the maps $\be_k \mapsto \mc{F}^{(\op{O})}(\bs{\be}_n)$ are
\begin{itemize}

\item meromorphic on the fundamental strip $\msc{S}$;

\item admit $+$, resp. $-$, boundary values $\mc{F}_{+}^{(\op{O})}$ on $\R$, resp. $\mc{F}_{-}^{(\op{O})}$ on  $\R+2\i\pi$;

\item bounded, coordinate-wise, at infinity by $ C\cdot \cosh\big( \ell \Re(\be_k) \big)$ for some $n$ and $k$ independent  $\ell$.

\end{itemize}

The axiomatic bootstrap equation correspond to the multi-variable Riemann--Hilbert problem
\begin{itemize}

\item[i)]  $\op{S}$-matrix symmetry
\beq
\mc{F}^{(\op{O})}\big( \bs{\be}_n \big) \;=\;   \mc{F}^{(\op{O})}\big( \bs{\be}_n^{(i+1 i)} \big) \, \op{P}_{i i+1} \,  \op{S}_{i i+1}\big(\be_{i i+1} \big) \;,
\enq
 where
\beq
\be_{ab} = \be_a - \be_b \qquad \e{and} \qquad  \bs{\be}_n^{(i+1 i)} \,= \, \big( \be_1,\dots, \be_{i-1},\be_{i+1},\be_{i},\be_{i+2},\dots, \be_{n} \big) \;. 
\enq
\item[ii)] Watson equation
\beqa
\mc{F}^{(\op{O})}\big( \bs{\be}_n+2\i\pi \bs{e}_1 \big) &= &
\mc{F}^{(\op{O})}\big( \be_2,\dots, \be_n,\be_1 \big) \op{P}_{2 \dots n 1}  \ex{2\i\pi \om_{\op{O}}  +  \i\pi \ga \sg_1^{z} }  \\
& = & \mc{F}^{(\op{O})}\big( \bs{\be}_n \big) \op{S}_{21}(\be_{21})\cdots \op{S}_{n1}(\be_{n1}) \, \ex{2\i\pi \om_{\op{O}} + \i\pi \ga \sg_1^{z} } \, ,
\eeqa
where 
\beq
\op{P}_{2 \dots n 1} \; = \; \op{P}_{n-1n}\,\op{P}_{n-2n-1}\cdots \op{P}_{12} \;, 
\enq
while $\om_{\op{O}}$ is the so-called index of mutual locality of the operator $\op{O}$ in respect to the elementary soliton/antisoliton field $\Phi$.
Also, the second term containing the $\i\pi \ga\sg_1^{z}$ corresponds to taking into account the case where the mutual locality of the operator $\op{O}$ in respect to the elementary soliton/antisoliton field $\Phi$
is different.

\item[iii)] Lorentz boost equation
\beq
\mc{F}^{(\op{O})}\big(\bs{\be}_n+ \th \, \ov{\bs{e}}_n\big) \;=\; \mc{F}^{(\op{O})}\big( \bs{\be}_n \big)\ex{ \th \op{s}_{\op{O}} + \f{\ga}{2}q_{\op{O}} \th}
\enq
where $\op{s}_{\op{O}}$ is the spin of the operator $\op{O}$, $q_{\op{O}} \in \mathbb{Z}$ the operator's charge,
while 
\beq
\ov{\bs{e}}_n \; = \; \sul{k=1}{n} \bs{e}_k \;.
\enq

\item[iv)] Kinematic pole equation
\bem
-\i \e{Res}\Big( \mc{F}^{(\op{O})}\big( \a+\i\pi , \bs{\be}_n, \be \big) \dd \a , \a=\be \Big)  \\
\;= \; \op{C}_{1n+2}\otimes \mc{F}^{(\op{O})}\big( \bs{\be}_n \big)_{2\dots n+1}
\cdot \Big\{ \op{S}_{2\,  n+2}(\be_1-\be)\cdots  \op{S}_{ n+1 n+2}(\be_{n}-\be)  \, - \, \ex{2\i\pi \om_{\op{O}} + \i\pi \ga \sg_1^{z}  } \Big\}
\label{ecriture 1ere eqn cinematique}
\end{multline}
where the labels indicate the local spaces to which the covectors are attached non-trivially while
\beq
\op{C}\; =\; \Big( \op{v}^+\otimes \op{v}^{-} \, + \, \op{v}^{-}\otimes \op{v}^{+} \Big)^{\mathtt{t}} \;. 
\label{definition vecteur conjugaison de charge}
\enq
An equivalent useful form is given by:
\beq
 -\i \e{Res}\Big( \mc{F}^{(\op{O})}\big( \a+\i\pi , \be, \bs{\be}_n \big) \dd \a , \a=\be \Big)\;= \; \op{C} \, \otimes \, \mc{F}^{(\op{O})}\big( \bs{\be}_n \big)
\cdot \Big\{ 1 \, - \, \ex{2\i\pi \om_{\op{O}} + \i\pi \ga\sg_1^{z}  }  \op{S}_{2\,  n+2}(\be-\be_n)\cdots  \op{S}_{ 2 3 }(\be-\be_1)  \Big\} \;.
\enq
\end{itemize}
Note that the indices appearing in the \textit{lhs} of \eqref{ecriture 1ere eqn cinematique} correspond to the indices of the spaces $\mf{h}_{a}^{*}$
on which the given covectors belong.

\subsection{The full bootstrap equation in the attractive regime}
\label{Section Programme Bootstrap avec etats lies}

From now on, we consider the case of a generic coupling constant $\xi < \pi$ such that $\f{\pi}{\xi}\not\in \mathbb{N}$.
The presence of poles at $\i u_k$ in the soliton/antisoliton $\op{S}$-matrix \eqref{ecriture equation pour poles matrices S}
imposes the existence of $n_{\xi}$, \textit{c.f.} \eqref{definition N xi n xi etc}, additional asymptotic particles, the breathers, building up the model's Hilbert spaces.
The bootstrap equation now pertain to the full form factor $\msc{F}^{(\op{O})}$, \textit{c.f.} \eqref{form factor BV}.

\begin{bootstrapaxiomsbound}

Given $q_{\op{O}} \in \mathbb{Z}$ fixed, the form factors $\msc{F}^{(\op{O})}(\bs{\la}_{\ell})$ are such that, for each $k \in \intn{1}{\ell}$,
the maps $\la_k \mapsto \msc{F}^{(\op{O})}(\bs{\la}_{\ell})$ with the variables $\la_a$, $a \not=k$, kept fixed, are
\begin{itemize}

	\item meromorphic on the fundamental strip $\msc{S}$;

	\item admit $+$ boundary values  on $\R$, and  $-$ boundary values  on  $\R+2\i\pi$;

	\item bounded, coordinate-wise, at infinity by $ C\cdot \cosh\big( L \Re(\la_k) \big)$ for some $n$, $q_{\op{O}}$ and $k$ independent  $L$.

\end{itemize}

They satisfy the multi-variable Riemann--Hilbert problem
\begin{itemize}

\item[I)]  $\msc{F}^{(\op{O})}\big(\bs{\la}_{\ell}\big) \;=\;   \msc{F}^{(\op{O})}\big( \bs{\la}_{\ell}^{(i+1 i)} \big)
\, \op{P}_{i i+1} \,   \big(\op{S}_{\e{tot}}\big)_{i i+1}\big(\la_{i i+1} \big)$;
	\item[II)] $\msc{F}^{(\op{O})}\big(\bs{\la}_{\ell}+2\i\pi \bs{e}_1 \big) \; = \;
\msc{F}^{(\op{O})}\big( \la_2,\dots, \la_{\ell},\la_1 \big) \op{P}_{2 \dots \ell 1}  \ex{2\i\pi \om_{\op{O}} \mf{D}_1 + \i\pi \ga \mf{S}_1^{z} }$,

	with $\mf{S}^{z} =  \e{Diag}\big(1,-1,0,\dots,0) $
	where the zero matrix is of size $n_{\xi} \times n_{\xi}$ and translates the fact that the twist is only acting on the (anti)solitonic part,
	$\mf{D}  =  \e{Diag}\big(1,1,0,\dots,0)$;

\item[III)]  $\msc{F}^{(\op{O})}\big(\bs{\la}_{\ell}+ \th \, \ov{\bs{e}}_{\ell}\big) \;=\; \msc{F}^{(\op{O})}\big( \bs{\la}_{\ell} \big)
\ex{ \th \op{s}_{\op{O}} + \f{\ga}{2} q_{\op{O}} \th } $;
\item[IV)] one has 
\bem
-\i \e{Res}\Big( \msc{F}^{(\op{O})}\big( \a+\i\pi , \bs{\la}_{\ell}, \be \big) \dd \a , \a=\be \Big)  \\
 \, = \,  \big( \op{C}_{\e{tot}}  \big)_{1\ell+2} \otimes \msc{F}^{(\op{O})}
\big( \bs{\la}_{\ell} \big)_{2\dots \ell+1}
\cdot \Big\{ \big(\op{S}_{\e{tot}}\big)_{2\,  \ell+2}(\la_1-\be)\cdots  \big(\op{S}_{\e{tot}}\big)_{ \ell+1 \ell+2}(\la_{\ell}-\be)
-  \ex{2\i\pi  \om_{\op{O}} \mf{D}_1 + \i\pi \ga \mf{S}_1^{z}   } \Big\}
\nonumber
\end{multline}

where
\beq
\op{C}_{\e{tot}} \; =\; \bigg( \op{v}^{\mf{s} }\otimes \op{v}^{ \ov{\mf{s}} } \, + \, \op{v}^{ \ov{\mf{s}} }\otimes \op{v}^{\mf{s}}  \, + \, \sul{k=1}{n_{\xi} }
 \op{v}^{\mf{b}_k }\otimes \op{v}^{ \mf{b}_k }   \bigg)^{\mathtt{t}} \;;
\nonumber
\enq
\item[V)] for any  $p, n \geq 0$ one has that the $(p+1)^{\e{st}}$ coordinate projection on the breather sector is obtained by a residue calculation procedure:
\begin{equation}
\label{recursion breather 2}
\msc{F}^{(\op{O})}(\bs{\alpha}_{p+1}, \bs{\be}_n ) \cdot \wh{\op{v}}^{\, \mf{b}_{k_{p+1}}}_{p+1} \, = \,
  \Res \Big( \msc{F}^{(\op{O})}(\bs{\alpha}_p , \bs{\de}_2, \bs{\be}_n) \, \dd \de_{1}
, \, \de_{12}=    \i \mf{u}_{k_{p+1}} \Big)_{\mid  \ov{\de}_{12} = 2 \alpha_{p+1} } \cdot \wh{ {\bs \Ga} }^{(k_{p+1})}_{p+1 p+2}   \;,
\end{equation}
where $\ov{\de}_{12}= \de_1 + \de_2 $ and the $k$-th breather projector reads:
\begin{equation}
\label{exp intertwiner}
{\bs \Ga}^{(k)}  \, =  \, \f{1}{  r_k }    \bs{\vp}^{(\mf{b}_k)} \qquad with \qquad
r_k = 2 \i (-1)^k \sqrt{\cot \Big( \tfrac{k\xi}{2} \Big) } \,  \pl{a=1}{k-1}\cot \Big( \tfrac{a\xi}{2} \Big) \;.
\end{equation}
Furthermore, the \, $\wh{}$ refers to the linear operators acting on the basis elements, \textit{c.f.} \eqref{ecriture decomposition FF total sur base de vecteurs}, as
\beq
\wh{\op{v}}^{\, \mf{b}_{k} }_{p} : \left\{ \ba{ccc }   \big(\mf{h}_{\e{tot}}^{*} \big)^{m+1}  & \tend & \big(\mf{h}_{\e{tot}}^{*} \big)^{m }   \vspace{2mm} \\
(\op{v}^{\bs{\tau}_{m+1}})^{\op{t}}  &   \mapsto     &
\big( \op{v}^{\tau_p}, \op{v}^{\, \mf{b}_{k}} \big)  \;
(\op{v}^{\tau_1})^{\op{t}} \otimes \cdots \otimes (\op{v}^{\tau_{p-1}} )^{\op{t}} \otimes (\op{v}^{\tau_{p+1}})^{\op{t}}\otimes
\cdots \otimes (\op{v}^{\tau_{m+1}})^{\op{t}}
\ea  \right.
\enq
and
\beq
 \wh{\bs{\vp}}^{(\mf{b}_k)}_{ p p+1 }  : \left\{ \ba{ccc }   \big(\mf{h}_{\e{tot}}^{*} \big)^{m+2}  & \tend & \big(\mf{h}_{\e{tot}}^{*} \big)^{m }   \vspace{2mm} \\
(\op{v}^{\bs{\tau}_{m+2}})^{\op{t}}  &   \mapsto     &
\big( \op{v}^{\tau_p}\otimes  \op{v}^{\tau_{p+1}} , \bs{\vp}^{(\mf{b}_k)} \big)  \;
(\op{v}^{\tau_1})^{\op{t}} \otimes \cdots \otimes (\op{v}^{\tau_{p-1}} )^{\op{t}} \otimes (\op{v}^{\tau_{p+2}})^{\op{t}}\otimes
\cdots \otimes (\op{v}^{\tau_{m+1}})^{\op{t}}
\ea  \right.
\enq
\end{itemize}

\end{bootstrapaxiomsbound}

Just as it is the case for the axioms appearing in the purely soliton/antisoliton sector, the origin of the inductive construction of
form factors involving breathers, \textit{viz}. \eqref{recursion breather 2} and \eqref{exp intertwiner},
be understood, heuristically, from the fusion procedure. We once again refer the reader to \cite{BabujianKarowskiExactFFSineGordonBootsstrapII} for more details.
One should also stress that axiom $V)$ fully determines the form factors in the breather sector in terms of their expression in the
soliton/antisoliton sector. Thus, starting from $\mc{F}^{(\op{O})}$, one fully constructs $\msc{F}^{(\op{O})}$ through axiom $V)$.
In fact, one has the
\begin{prop}
 \cite{BabujianKarowskiExactFFSineGordonBootsstrapII}
\label{Proposition Full vs Solitonic Axioms equivalence Karowski Babujian}
Axioms $i)-iv)$ and Axiom $V)$ are equivalent to Axioms $I)-V)$. Namely, under the previously introduced notations,
any solution $\mc{F}^{(\op{O})}$ to equations $i)-iv)$ which defines  $\msc{F}^{(\op{O})}$ through axiom $V)$
produces a solution to equations $I)-V)$. Reciprocally, any solution  $\msc{F}^{(\op{O})}$ to axioms $I)-V)$
produces a solution $\mc{F}^{(\op{O})}$ to  $i)-iv)$ by means of  \eqref{form factor coordinates}.
\end{prop}

In view of the proposition above, it is in fact enough to only fix a "portion" of the full solution $\msc{F}^{(\op{O})}$ to axioms $I)-V)$,
as all of its other values are then deduced from Watson's $\op{S}$-symmetry \textit{viz.} Axiom $I)$.
First of all, owing to Proposition \ref{Proposition Full vs Solitonic Axioms equivalence Karowski Babujian},
its restriction to the soliton/antisoliton sector, \textit{c.f.} \eqref{form factor coordinates}, solves the axioms $i)-iv)$.
Further, by using Axiom $V)$, one may inductively obtain the value of $\msc{F}^{(\op{O})}( \bs{\ga}_{\ell} )$ on a subspace of the full vector space
$\big(\mf{h}_{\e{tot}}^*\big)^{\otimes \ell}$. More precisely, by using the convention introduced earlier in \eqref{ecriture decomposition FF total sur base de vecteurs},
consider the vectors
$\op{v}^{ \bs{\mf{b}}_{s}}\otimes \op{v}^{\bs{\veps}_n} \in \big(\mf{h}_{\e{tot}}^*\big)^{\otimes \ell}$ with $\ell=n+s$
and where $\bs{\mf{b}}_{s} \, = \, \big( \mf{b}_{k_1}, \dots, \mf{b}_{k_s} \big)$, for some $k_a \in \intn{1}{n_{\xi}}$,
while $\bs{\veps}_n = \big(\veps_1, \dots, \veps_n \big)\in \{\mf{s}, \ov{\mf{s}} \,  \}^n$.
This allows one to introduce the $\bs{\mf{b}}_{s}$ breather subsector form factors:
\beq
\Big( \msc{F}^{(\op{O})} \Big)_{1,\dots,n}^{  \bs{\mf{b}}_{s}  } \big( \bs{\a}_s, \bs{\be}_n \big) \; = \;
\sul{ \veps_{n} \in \{\mf{s}, \ov{\mf{s}} \}^n }{}
\Big( \msc{F}^{(\op{O})}  \big( \bs{\a}_s, \bs{\be}_n \big) \cdot \op{v}^{ \bs{\mf{b}}_{s}}\otimes \op{v}^{\bs{\veps}_n} \Big)
\cdot \big( \op{v}^{\bs{\veps}_n}  \big)^{ \mathtt{t} }
\label{definition bs Brether projection}
\enq
which are linear forms on $\mf{h}_1\otimes \cdots \otimes \mf{h}_{n}$ introduced above. These are rather directly computable from Axiom $V)$
by adding one breather rapidity at a time. Finally, the value of $\msc{F}^{(\op{O})}( \bs{\ga}_{\ell} )$ on the full space
is obtained from Axiom $I)$. Namely, given $\bs{\tau}_{\ell} \in \mf{X}^{\ell}$, \textit{c.f.} \eqref{ecriture decomposition FF total sur base de vecteurs},
there exists $n, s \in \mathbb{N}$, $\sg\in \mf{S}_{\ell}$ and $\bs{\mf{b}}_{s}, \bs{\veps}_n$ as above such that
\beq
\bs{\tau}_{\ell}^{\sg} \, = \, \big( \tau_{\sg(1)},\dots, \tau_{\sg(\ell)} \big) \, = \, \big(  \bs{\mf{b}}_{s}, \bs{\veps}_n \big)
\quad \e{and} \quad
 \bs{\ga}_{\ell}^{\sg}\, = \, \big( \bs{\a}_{s} , \bs{\be}_n \big) \;.
\enq
By Axiom $I)$, there exists a directly computable function $\mc{S}\big( \bs{\ga}_{\ell}^{\sg} \mid \bs{\ga}_{\ell} \big)$
such that
\beq
\Big[ \msc{F}^{(\op{O})}( \bs{\ga}_{\ell} ) \Big]^{ \bs{\tau}_{\ell} } \; = \;
\Big[ \msc{F}^{(\op{O})}( \bs{\ga}_{\ell}^{\sg} ) \Big]^{ \bs{\tau}_{\ell}^{\sg} }  \mc{S}\big( \bs{\ga}_{\ell}^{\sg} \mid \bs{\ga}_{\ell} \big)
 \; = \;
\Big[ \msc{F}^{(\op{O})}\big( \bs{\a}_{s} , \bs{\be}_n \big) \Big]^{ \bs{\mf{b}}_{s} , \bs{\veps}_n  }  \mc{S}\big( \bs{\ga}_{\ell}^{\sg} \mid \bs{\ga}_{\ell} \big) \;.
\enq
 The last coordinate appearing in the utmost \textit{rhs} is already contained in \eqref{definition bs Brether projection}.

 \vspace{3mm}

We also stress that the form factors involving breathers also exhibit other poles attached to other particle processes, such as for instance
$\mf{b}_m + \mf{s} \rightarrow \mf{b}_n + \mf{s}$. Just as $I)-IV)$ following from $i)-iv)$ and $V)$, these additional poles and the form of the residues
follow from $i)-iv)$ and $V)$. We do not state these additional poles here though and refer to \cite{SmirnovFormFactors} for more details.

\section{Main results: integral representations for the form factors}
\label{Section Main Results}

\subsection{Solution of the bootstrap equations in the soliton/antisoliton sector}
\label{SubSection FF for solitons antisolitons}

The explicit solutions to the bootstrap equations in the soliton/antisoliton sector utilise a fair amount of auxiliary special functions that we need to introduce first.
While this subsection solely provides a set of solutions to the form factor axioms $i)-iv)$, we stress that its form is \textit{independent} of the value of $\xi$,
\textit{viz}. whether one is in the repulsive $\xi > \pi$ or the attractive $\xi<\pi$ regimes.

\subsubsection{The off-shell Bethe vectors}

Off-shell Bethe vectors provide one with a convenient continuous family of vectors
on whose linear combinations one can conveniently decompose the form factors.
To construct these, one introduces the $\ga$-twisted monodromy matrix
\beq
\label{monodromy matrix}
\op{T}^{(\ga)}_{1,\dots,n;0}(\bs{\be}_n ; \la) \, = \, \op{S}_{1 0}^{(\ga)}(\be_1-\la)\cdots \op{S}_{n 0}^{(\ga)}(\be_n-\la) 
\; = \; \left( \ba{cc} \op{A}_{1,\dots,n}^{(\ga)}(\bs{\be}_n ; \la) & \op{B}_{1,\dots,n}^{(\ga)}(\bs{\be}_n ; \la)  \\ 
\op{C}_{1,\dots,n}^{(\ga)}(\bs{\be}_n ; \la) & \op{D}_{1,\dots,n}^{(\ga)}(\bs{\be}_n ; \la) 		\ea \right) \;. 
\enq
It has $\mf{h}_{\mf{q}}=\mf{h}_1\otimes \cdots \otimes \mf{h}_n$ as its quantum space and $\mf{h}_{0}$ as its auxiliary space, where $\mf{h}_a \simeq \Cx^2$.
The pseudo-vacuum covector $\Om_n\in \mf{h}_{q}^{*}$ takes the form
\beq
\Om_n \, = \, \underbrace{ (\op{v}^+)^{\mathtt{t}}\otimes \cdots \otimes (\op{v}^{+})^{\mathtt{t}} }_{ n-\e{times} } \;.  
\enq
It is a left Eigenvector of the operators $\op{A}_{1,\dots,n}^{(\ga)}(\bs{\be}_n ; \la)$ and $\op{D}_{1,\dots,n}^{(\ga)}(\bs{\be}_n ; \la)$:
\beq
\Om_n \cdot \op{A}_{1,\dots,n}^{(\ga)}(\bs{\be}_n ; \la) \; = \; \Om_n \pl{k=1}{n} a(\be_k-\la) \qquad \e{and} \qquad
\Om_n \cdot \op{D}_{1,\dots,n}^{(\ga)}(\bs{\be}_n ; \la) \; = \; \Om_n \pl{k=1}{n} b(\be_k-\la)  \;. 
\enq
The associated off-shell Bethe vectors takes the form
\beq
\Cx_{1,\dots,n}^{(\ga)}\big( \bs{\be}_n \, ; \, \bs{u}_m\big) \; = \;
\Om_n \cdot \op{C}_{1,\dots,n}^{(\ga)}(\bs{\be}_n ; u_1)  \cdots \op{C}_{1,\dots,n}^{(\ga)}(\bs{\be}_n ; u_m) \;.
\label{definition C vect de Bethe}
\enq
It will also appear convenient to introduce analogous quantities built out of the $\wt{\op{S}}$-matrix,
\textit{viz}. with the $a$ function being factored out, \textit{c.f.} \eqref{definition matrice S tildee}.
Namely, set 
\beq
\wt{\op{T}}_{1,\dots,n;0}^{\, (\ga)}(\bs{\be}_n ; \la) \, = \, \wt{\op{S}}_{1 0}^{\, (\ga)}(\be_1-\la)\cdots \wt{\op{S}}_{n 0}^{\, (\ga)}(\be_n-\la) 
\; = \; \left( \ba{cc} \wt{\op{A}}_{1,\dots,n}^{\, (\ga)}(\bs{\be}_n ; \la) & \wt{\op{B}}_{1,\dots,n}^{\, (\ga)}(\bs{\be}_n ; \la)  \\ 
\wt{\op{C}}_{1,\dots,n}^{\, (\ga)}(\bs{\be}_n ; \la) & \wt{\op{D}}_{1,\dots,n}^{\, (\ga)}(\bs{\be}_n ; \la) 		\ea \right) \;,
\enq
and let 
\beq
\wt{\Cx}_{1,\dots,n}^{  (\ga)}\big( \bs{\be}_n \, ; \, \bs{u}_m\big) \; = \; \Om_n \cdot \wt{\op{C}}_{1,\dots,n}^{\, (\ga)}(\bs{\be}_n ; u_1)  \cdots \wt{\op{C}}_{1,\dots,n}^{\, (\ga)}(\bs{\be}_n ; u_m) \;.  
\label{definition tilde C vect de Bethe}
\enq
It is easy to check that $u_k\mapsto \wt{\Cx}_{1,\dots,n}^{   (\ga)}\big( \bs{\be}_n \, ; \, \bs{u}_m\big)$ is a covector valued meromorphic function on $\Cx$ and which, for generic $\bs{\be}_n$, has only simple poles
located at the points 
\beq
u_k\, = \, \be_a-\i \pi + \i k  \xi \quad \; , \quad k \in \mathbb{Z} \;. 
\enq
These poles stem from the denominators in $\wt{b}(\th)$ and $\wt{c}^{\, (\pm \ga)}(\th)$, \textit{c.f.} \eqref{definition tilde b et c}. We also need to introduce the operator 
\beq
\op{O}^{(\ga)}_{1, \dots, n}(\bs{\be}_n) \, = \, \ex{ \f{ \ga}{2} \bs{\be}_n\cdot \bs{\sg}^z_n } \qquad \e{with} \qquad \bs{\sg}^z_n = \big( \sg_1^z,\dots,\sg_n^z \big) \;. 
\label{definition operateur spin twist}
\enq

\subsubsection{The auxiliary functions}

To start with, one introduces the minimal form factor $\op{F}$, which corresponds to the unique, meromorphic on $\Cx$, solution to the equation 
\beq
\op{F}(\th) = a(\th) \op{F}(-\th) = \op{F}(2\i\pi - \th)
\label{ecriture eqn fnelle FF minimal}
\enq
which 
\begin{itemize}
\item  grows as $ \mp \tfrac{\i}{2} \ex{ \pm \frac{\th}{2}}$ as $\Re(\th) \tend \pm \infty$; 
\item has no poles or zeroes in the fundamental
strip $\Big\{ z \in \Cx \; : \; 0< \Im(z) < 2 \pi \Big\}$, has simple zeroes at $\th=0, 2\i\pi $;
\item is normalised so that $\op{F}(\i\pi)=1$. 
\end{itemize}
By following the procedure pioneered in \cite{KarowskiWeiszFormFactorsFromSymetryAndSMatrices}, one readily gets that 
\beq
\label{def F function}
\op{F}(\th) \; = \; - \i \sinh \big( \tfrac{\th}{2}  \big) \cdot \op{f}_{\e{min}}(\th) 
\enq
where 
\beq
\label{int rep F}
\op{f}_{\e{min}}(\th) \; = \; \exp\Bigg\{  -\Int{0}{+\infty} \f{\dd t }{t} \f{ \sinh\big[\tfrac{\pi -\xi }{2 } t\big]   \sinh^2\big[\tfrac{\pi t }{2 }(1-\tfrac{\th}{\i\pi}) \big]  }
{ \cosh\big[\tfrac{\pi t }{2 } \big]   \sinh\big[\tfrac{\xi t }{2 } \big]  \sinh\big[ \pi t \big] } \Bigg\} \;.
\enq
Note that this integral representation is defined in the range $-\e{min}(\pi,\xi) \, < \,  \Im(\th)  \, < \, 2\pi + \e{min}(\pi,\xi)$. \\
The two other auxiliary functions of interest are built out of $\op{F}$:
\beq
\phi(\th) \, = \, \f{ 1 }{ \op{F}(\th)\op{F}(\th+\i\pi) }  \quad \e{and} \quad \tau(\th) \, = \, \f{1}{\phi(\th)\phi(-\th)} \, = \, 
\op{F}(\th)\op{F}(-\th)\op{F}(\i\pi+\th)\op{F}(\i\pi-\th) \;. 
\label{definition fonction phi}
\enq
Then, a direct calculation yields that 
\beqa
\phi(\th) & = & \f{ \i }{ \sinh(\th) \, \op{F}^2\big( \i \tfrac{\pi}{2} \big) } \cdot
\exp\Bigg\{ 2 \Int{ 0 }{ + \infty } \dd t \f{  \sinh\big[ \tfrac{\pi-\xi}{2} t \big]   }{ t\sinh\big[ \tfrac{\xi t }{2} \big] \sinh[\pi t] } 
\sinh^2\Big[ \tfrac{\pi t }{ 4 } \big( 1 - \tfrac{2 \th }{ \i \pi } \big)  \Big] \Bigg\}  \\
& = & \f{ \i }{ \sinh(\th) \, \op{F}^2\big( \i \tfrac{\pi}{2} \big) \, \varpi^2\big(\i \tfrac{\xi-\pi}{2} \big)} \cdot
\varpi \Bigg( \ba{c}   \th+\i\tfrac{\xi}{2}-\i \pi \\ \th-\i\tfrac{\xi}{2}   \ea \Bigg) \;,  
\label{expression explicite pour phi}
\eeqa
and 
\beq
\tau(\th) \, = \, \sinh \th \, \sinh \big[ \tfrac{\pi \th }{\xi} \big]  \cdot \op{F}^2\big( \i \tfrac{\pi}{2} \big)\,
\op{F}^2\big( - \i \tfrac{\pi}{2} \big) \, \varpi^2\Big(\i \tfrac{\xi + \pi}{2} \, , \,  \i \tfrac{\xi - 3 \pi}{2} \Big) \;. 
\enq
These functions satisfy the below finite difference equations.
\beq
\phi(\th+2\i\pi) \, = \, -\phi(\th) \f{ \sinh\big[ \tfrac{\pi \th }{ \xi }\big]  }{  \sinh\big[ \tfrac{\pi }{ \xi } (\th +\i\pi )\big]   } 
\quad \e{and} \quad 
\phi(\th+\i\xi) \, = \, -\i \phi(\th) \f{ \sinh\big[ \tfrac{ \th }{ 2 }\big]  }{  \cosh\big[ \tfrac{ 1 }{ 2 } (\th +\i\xi )\big]   } \;
\enq
in what concerns $\phi$, while regarding to $\tau$
\beq
\tau(\th \pm 2\i\pi ) \, = \, \tau(\th) \cdot \f{   \sinh\big[ \tfrac{\pi }{ \xi } (\th \pm 2 \i \pi ) \big]    }{  \sinh\big[ \tfrac{\pi \th }{ \xi }\big]   } \;.
\enq
It will also be useful to observe that 
\beq
\th \mapsto a(\th) \cdot \phi(\th) \; = \;  \f{ - \i }{ \sinh(\th) \, \op{F}^2\big( \i \tfrac{\pi}{2} \big) \, \varpi^2\big(\i \tfrac{\xi-\pi}{2} \big)} 
\varpi \Bigg( \ba{c}   \th+\i\tfrac{\xi}{2} \\  \th + \i \pi - \i\tfrac{\xi}{2}    \ea \Bigg)
\label{definition fct a phi}
\enq
has only simple poles, and these are located at 
\begin{itemize}

\item $\th = \i p \xi + 2\i\pi \ell  $, $p, \ell \in \mathbb{N}$;

\item $\th = -\i p \xi - \i\pi (2\ell+1)  $, $p, \ell \in \mathbb{N}$.

\end{itemize}
This can be seen as follows. The simple poles of the quantum dilogarithm are located at 
\beq
\th=-\i\xi (p+1)-\i\pi (2\ell + 1) \quad \e{and} \quad \th = \i\xi (p+1) + 2\i\pi \ell \quad \e{with} \quad  p, \ell \geq 0
\enq
while the hyperbolic sine in the denominator has simple poles at $\th = \i\pi k$, $k\in \mathbb{Z}$. Note that the quantum dilogarithm also has simple zeroes at
\beq
\th= \i p \xi + \i\pi (2\ell + 1) \quad \e{and} \quad \th =  - \i\xi p - 2\i\pi (\ell + 1) \quad \e{with} \quad  p, \ell \geq 0\;. 
\enq
Thus the zeroes of $\th \mapsto \sinh(\th)$ at $\i\pi (2 \ell +1)$ and $-2\i\pi (\ell + 1)$, $\ell \geq 0$ are compensated by the zeroes of the ratios
of dilogarithms. Putting the remaining poles together yields the claim.\\
The function $a\phi$ will play an important role in the computation of  particular cases of form factors of the field exponential and related operators,
\textit{c.f.} Section \ref{Section FF exp champ}.
Several of its properties are detailed in Appendix \ref{Appendice Special Functions}, notably its Fourier transform is given explicitly.

Finally, it will appear handy to introduce the shorthand notation
\beq
g \big(\bs{\be}_n;\bs{u}_{m} \big)  \; = \;
\pl{k<\ell}{n} \op{F}(\be_{k}-\be_{\ell}) \cdot \pl{k=1}{n} \pl{\ell=1}{m} \big( a \cdot \phi)(\be_k-u_{\ell}) \cdot \pl{k<\ell}{m}\tau(u_k-u_{\ell}) \;,
\label{definition fonction g}
\enq
as well as
\beq
\wt{g} \big(\bs{\be}_n;\bs{u}_{m} \big)  \; = \;
\pl{k<\ell}{n} \op{F}(\be_{k}-\be_{\ell}) \cdot \pl{k=1}{n} \pl{\ell=1}{m}  \phi (\be_k-u_{\ell}) \cdot \pl{k<\ell}{m}\tau(u_k-u_{\ell}) \;.
\label{definition fonction tilde g}
\enq

\subsubsection{The integration curve and $p$-function problem}
\label{subsubsection definition of BA integration curve}

Given rapidities $\bs{\be}_n\in \Cx^n$ with generic coordinates, the contour $\msc{C}_{ \bs{\be}_n }$  refers to any curve which satisfies the properties
\begin{itemize}

\item $\msc{C}_{ \bs{\be}_n }$ connects $-\infty$ to $+\infty$;

\item  the points $\be_a-\i p \xi - 2 \i \pi \ell $, $\be_a-\i\pi -\i (p +1) \xi$, $p, \ell \in \mathbb{N}$,  
have index $-1$ in respect to $\msc{C}_{ \bs{\be}_n }$, \textit{viz}. they are located "below" of $\msc{C}_{ \bs{\be}_n }$;

\item  the points $\be_a+\i p \xi + (2 \ell + 1) \i \pi$, $\be_a-\i\pi +\i p \xi$,  $p, \ell \in \mathbb{N}$, 
have index $1$ in respect to $\msc{C}_{ \bs{\be}_n }$, \textit{viz}. they are located "above" of $\msc{C}_{ \bs{\be}_n }$.

\end{itemize}
One should observe that any so defined contour  $\msc{C}_{ \bs{\be}_n }$ separates the poles of
the function $g \big(\bs{\be}_n;\bs{u}_{m} \big) \wt{\Cx}_{1,\dots,n}^{\, (\ga)}\big(\bs{\be}_n;\bs{u}_m \big)$
\beqa
&& u_b = \be_a-\i\pi +\i k \xi \quad k \in \mathbb{Z} \label{series de pole covecteur des C} \; ,\\
&& u_b = \be_a +\i p \xi +\i\pi (2 \ell + 1) \; , \quad u_b = \be_a - \i p \xi  - 2 \i\pi \ell \quad p, \ell \in \mathbb{N} \;,
\label{series de pole a phi}
\eeqa
with $(a,b) \in \intn{1}{n}\times \intn{1}{m}$, in two disjoint groups.

We do stress that contours $\msc{C}_{ \bs{\be}_n }$ are well-defined for generic values of the coupling constant $\xi$, even if their precise form is
simpler in the repulsive sector $\xi > \pi$. This, in particular, makes the solution presented in the next subsubsection
to be valid irrespectively of the value of $\xi$, as long as it remains generic.

To close the definition section, we introduce the $p$-function problem. The family of $p$-functions, $p: \Cx^n \times \Cx^m \tend \Cx$ with $n,m \in \mathbb{N}$
solves the $p$-function problem if all of the below equations are fullfilled:
\begin{itemize}
\item[a)] $ p \big(\bs{\be}_n;\bs{u}_m \big)$ is a symmetric function of $\bs{\be}_n$ and $\bs{u}_{m}$;
\item[b)]  for any $k$, $u_k \mapsto p \big(\bs{\be}_n;\bs{u}_m \big)$ is a polynomial in the variables $\ex{\pm u_k}$;
\item[c)] $ p \big(\bs{\be}_n + \th \, \ov{\bs{e}}_n;\bs{u}_m + \th \, \ov{\bs{e}}_m \big)\, = \,  p \big(\bs{\be}_n ,\bs{u}_m \big) \ex{\th \op{s} } $ with $\op{s} \in \R$;
\item[d)] $p \big(\bs{\be}_n + 2\i \pi \bs{e}_k ; \bs{u}_m  \big) \, = \,  p \big(\bs{\be}_n  ; \bs{u}_m  \big) \cdot \ex{2\i\pi \om } $ with $\om \in \R$;
\item[e)] there exists two symmetric functions $\vp^{(\pm)}$ of $n-1$ variables such that given
\beq
\varkappa \; = \;  \f{  2\i\xi  \ex{\f{\i\pi \ga}{2} } }{  \Big\{ \op{F}(\i\tfrac{\pi}{2}\big) \cdot \varpi\big(\i \tfrac{\xi-\pi}{2} \big)  \Big\}^4   } \;,
\label{definition cste varkappa}
\enq
it holds
\beqa
p \big( (\be_n+\i\pi, \bs{\be}_n^{\prime} ) ; (\bs{u}_{m-1},\be_n)  \big) & = & \f{ \i }{m \varkappa }  p \big( \bs{\be}_{n-1}^{\prime} ; \bs{u}_{m-1}  \big) \; + \; \de_{n,2m} \vp^{(+)}\big( \bs{\be}_{n-1}^{\prime} \big) \\
p \big( ( \bs{\be}_{n-1}, \be_1-\i\pi ) ; (\bs{u}_{m-1},\be_1)  \big) & = & \f{  \i   }{m \varkappa }  p \big( \bs{\be}_{n-1}^{\prime} ; \bs{u}_{m-1}  \big) \cdot  \ex{2\i\pi \om }
\; + \; \de_{n,2m} \vp^{(-)}\big( \bs{\be}_{n-1}^{\prime} \big)
\eeqa
where $\bs{\be}^{\prime}_n\, = \, \big( \be_2, \dots, \be_{n} \big)$.

\end{itemize}

\subsubsection{A Bethe Ansatz family of solutions to axioms $i)-iv)$}
\label{subsubsection solution bootstrap}

\begin{theorem}  
Let
\beq
q=n-2m
\enq
be fixed and $p$ be a family of functions solving the $p$-function problem of Subsubsection \ref{subsubsection definition of BA integration curve}
associated with spin $\op{s}$ and scalar mutual locality index $\om$.  Then, the covector valued functions
\beq
\label{general soliton form factor}
\bs{\Psi}_{1,\dots,n}\big[ p\big]\big( \bs{\be}_n \big) \, = \hspace{-2mm} \Int{ (\msc{C}_{\bs{\be}_n})^m }{} \hspace{-2mm} \dd^{m} u \;
\wt{\Cx}_{1,\dots,n}^{\, (\ga)}\big(\bs{\be}_n;\bs{u}_m \big) \cdot \op{O}_{1,\dots,n}^{(\ga)}(\bs{\be}_n)  \cdot \big( g \cdot p \big)\big(\bs{\be}_n;\bs{u}_{m} \big)  
\enq
solve the form factor axioms $\mathrm{i)}-\mathrm{iv)}$ associated with a spin $\op{s}$, scalar mutual locality index $\om$ and vector  mutual locality index $\ga$.

\end{theorem}

The proof of the theorem, which is based on the concept of off-shell Bethe vector
\cite{BabujianKarowskiZapletalSomeDvpmtofOffShellBA,ReshetikhinOffShellBAforKZ,SmirnovFormFactors}, is gathered in Appendix
\ref{Appendice Preuve des eqns de FF}. The Watson $\op{S}$-symmetry equation is proven in Appendix \ref{Appendice SS section eqn symmetrie i},
the Watson monodromy equation in Appendix \ref{Appendice SS section eqn Watson ii}, the Lorentz invariant equation in Appendix
\ref{Appendice SS section eqn inv Lorentz iii} and the kinematic pole axiom in Appendix \ref{subsection axiom iv}.

There exists an alternative, complementary, representation for the solution of the bootstrap axioms which turns out to be handy when computing two-point functions.
To state it, we first introduce the dual pseudo-vacuum vector
\beq
\ov{\Om}_n \, = \, \underbrace{ (\op{v}^{-})^{\mathtt{t}}\otimes \cdots \otimes (\op{v}^{-})^{\mathtt{t}} }_{ n-\e{times} } \;,  
\enq
and the associated off-shell Bethe vector
\beq
\wt{\Bx}_{1,\dots,n}^{  (\ga)}\big( \bs{\be}_n \, ; \, \bs{u}_m\big) \; = \; \ov{\Om}_n \cdot \wt{\op{B}}_{1,\dots,n}^{\, (\ga)}(\bs{\be}_n ; u_1)
\cdots \wt{\op{B}}_{1,\dots,n}^{\, (\ga)}(\bs{\be}_n ; u_m) \;.
\label{definition tilde B vect de Bethe}
\enq

\begin{cor}

Given  fixed $q=n-2m$, the covector valued functions 
\beq
\ov{\bs{\Psi}}_{1,\dots,n}\big[ \ov{p}  \big]\big( \bs{\be}_n \big) \, = \, \Int{(\msc{C}_{\bs{\be}_n})^m }{} \dd^{m} u \; 
\wt{\Bx}_{1,\dots,n}^{\, (\ga)}\big(\bs{\be}_n;\bs{u}_m \big) \cdot \op{O}_{1,\dots,n}^{(\ga)}(\bs{\be}_n)  \cdot \big( g \cdot \ov{p} \big)\big(\bs{\be}_n;\bs{u}_{m} \big)  
\label{definition solution de base bar Psi avec op B}
\enq
satisfy the form factor axioms $\mathrm{i)}-\mathrm{iv)}$ associated with a spin $\op{s}$, mutual locality index $\om$ and vector index $\ga$,
provided that the family of functions $\ov{p}$ satisfies the $p$-function problem equations
$\mathrm{a)-e)}$ associated with a spin $\op{s}$, mutual locality index $\om$ and vector index $-\ga$, \textit{i.e.}
in which one proceeds to the replacement $\ga \mapsto -\ga$ in the value of the constant $\varkappa$.

Moreover, for any  $\gamma$-independent family of functions $p$, not necessarily solving $\mathrm{a)-e)}$, one has the relation
\beq
\ov{\bs{\Psi}}_{1,\dots,n}\big[ p  \big]\big( \bs{\be}_n \big) \, = \,
\bs{\Psi}_{1,\dots,n}\big[ p \big]\big( \bs{\be}_n \big)_{\mid \ga \hookrightarrow - \ga}  \cdot  \Sg_{1,\dots,n}^{x} \;.
\enq

\end{cor}

\Proof 
Owing to the symmetry
\beq
\sg_{a}^x \,  \sg_{b}^{x} \,  \op{S}_{ab}(\la) \, \sg_{a}^x \,  \sg_{b}^{x} \; = \; \op{S}_{ab}(\la) 
\enq
one has the exchange relation
\beq
\op{T}^{(\ga)}_{1,\dots,n;0}\big( \bs{\be}_n;\la \big) \cdot \Sg_{1,\dots,n}^{x}  \, = \,
\Sg_{1,\dots,n}^{x} \cdot  \sg_{0}^{x} \, \op{T}^{(-\ga)}_{1,\dots,n;0}\big( \bs{\be}_n;\la \big)  \, \sg_{0}^{x}
\enq
in which $\Sg_{1,\dots,n}^{x} \, = \, \sg_1^{x} \cdots \sg_{n}^{x}$. In particular, one has 
\beq
\op{C}^{(\ga)}_{1,\dots,n}\big( \bs{\be}_n;\la \big) \cdot  \Sg_{1,\dots,n}^{x}  \, = \, \Sg_{1,\dots,n}^{x}  \cdot  \op{B}^{(-\ga)}_{1,\dots,n}\big( \bs{\be}_n;\la \big) \;. 
\label{ecriture conjugaison par Sg x de C et B}
\enq
Further, one also has $\op{O}_{1,\dots,n}^{(\ga)}(\bs{\be}_n) \cdot  \Sg_{1,\dots,n}^{x} \, = \, \Sg_{1,\dots,n}^{x}  \cdot \op{O}_{1,\dots,n}^{(-\ga)}(\bs{\be}_n)$.\\
Thus, it follows that 
\beq
\ov{\bs{\Psi}}_{1,\dots,n}\big[ \ov{p} \big]\big( \bs{\be}_n \big) \; = \; \bs{\Psi}_{1,\dots,n}^{(-\ga)}\big[ \ov{p} \big]\big( \bs{\be}_n \big)  \cdot  \Sg_{1,\dots,n}^{x}  \;. 
\enq
Here, we introduced a temporary notation which allows one to insist on the $\ga$-dependence of the $\bs{\Psi}_{1,\dots,n}$ transform
\beq
\bs{\Psi}_{1,\dots,n}^{(\ga)} \big[ p\big]\big( \bs{\be}_n \big) \, =   \hspace{-2mm}  \Int{ (\msc{C}_{\bs{\be}_n})^m }{} \hspace{-2mm}\dd^{m} u \;
\wt{\Cx}_{1,\dots,n}^{\, (\ga)}\big(\bs{\be}_n;\bs{u}_m \big) \cdot \op{O}_{1,\dots,n}^{(\ga)}(\bs{\be}_n)  \cdot \big( g \cdot p \big)\big(\bs{\be}_n;\bs{u}_{m} \big)  \;. 
\enq
However, we do stress that there is no $\ga\hookrightarrow -\ga$ replacement in the expression for $\ov{p}$. 
One then has that 
\beqa
\ov{\bs{\Psi}}_{1,\dots,n}\big[ \ov{p} \big]\big( \bs{\be}_n \big)  & = &
\bs{\Psi}_{1,\dots,n}^{(-\ga)}\big[ \ov{p} \big]\big( \bs{\be}_n^{(i+1 i)} \big) \cdot \op{P}_{i i+1} \,  \op{S}_{i i+1}\big(\be_{i i+1} \big) \cdot  \Sg_{1,\dots,n}^{x} \\ 
& = & \ov{\bs{\Psi}}_{1,\dots,n}\big[ \ov{p} \big]\big( \bs{\be}_n \big) \big( \bs{\be}_n^{(i+1 i)} \big) \cdot  \op{P}_{i i+1} \,  \op{S}_{i i+1}\big(\be_{i i+1} \big) \;.
\eeqa
This corresponds to equation $\mathrm{i)}$ of the bootstrap axioms. Further,  
\beqa
\ov{\bs{\Psi}}_{1,\dots,n}\big[ \ov{p} \big]\big( \bs{\be}_n+2\i\pi \bs{e}_1 \big)  & = &
\bs{\Psi}_{1,\dots,n}^{(-\ga)}\big[ \ov{p} \big]\big( \be_2,\dots, \be_n,\be_1 \big) \op{P}_{2 \dots n 1}  \ex{2\i\pi \om - \i\pi \ga \sg_1^{z} } \cdot
\Sg_{1,\dots,n}^{x} \\
& = & \bs{\Psi}_{1,\dots,n}^{(-\ga)}\big[ \ov{p} \big]\big( \bs{\be}_n \big) \op{S}_{21}(\be_{21})\cdots \op{S}_{n1}(\be_{n1})
\, \ex{2\i\pi \om - \i\pi \ga \sg_1^{z} } \cdot  \Sg_{1,\dots,n}^{x} \\
& = &
\ov{\bs{\Psi}}_{1,\dots,n}\big[ \ov{p} \big]\big( \bs{\be}_n \big) \op{S}_{21}(\be_{21})\cdots \op{S}_{n1}(\be_{n1})
\, \ex{2\i\pi \om + \i\pi \ga \sg_1^{z} } \;.
\eeqa
Hence, axiom $\mathrm{ii)}$ follows. Axiom $\mathrm{iii)}$ is trivial to check. 
Finally, 
\bem
-\i \e{Res}\Big( \ov{\bs{\Psi}}_{1,\dots,n+2}\big[ \ov{p} \big]\big( \a+\i\pi , \bs{\be}_n, \be \big) \dd \a , \a=\be \Big)    \\
\;= \; \op{C}_{1n+2}\otimes \bs{\Psi}_{2\dots n+1}^{(-\ga)}\big[ \ov{p} \big]\big( \bs{\be}_n \big)
\cdot \Big\{ \op{S}_{2\,  n+2}(\be_1-\be)\cdots  \op{S}_{ n+1 n+2}(\be_{n}-\be)  \, - \, \ex{2\i\pi \om - \i\pi \ga \sg_1^{z}  } \Big\} \cdot  \Sg_{1,\dots,n+2}^{x} \\
\;= \; \op{C}_{1n+2}\otimes \ov{\bs{\Psi}}_{2\dots n+1}\big[ \ov{p} \big]\big( \bs{\be}_n \big)
\cdot \Big\{ \op{S}_{2\,  n+2}(\be_1-\be)\cdots  \op{S}_{ n+1 n+2}(\be_{n}-\be)  \, - \, \ex{2\i\pi \om + \i\pi \ga \sg_1^{z}  } \Big\}  \;. 
\end{multline}
Here, we used that $\op{C}_{ab}=\op{C}_{ab} \sg_a^{x} \sg_b^{x}$, so that axiom $\mathrm{iv)}$ follows.  \qed




\subsection{Solution of the bootstrap equations in the full sector, with bound states}
\label{SubSection FF Full sector with breathers}

Using the expression we obtained for the purely solitonic/antisolitonic form factors \eqref{general soliton form factor} as well as the recursion relation
\eqref{recursion breather 2}, one is now able to construct the full set of form factors involving breather particles.
The expression involves several auxiliary functions which we shall now define. Just as for the $\op{S}$-matrices, the calculation of breather residues
generates new functions expressing the interactions between particles. First, one has the $\op{F}$-factors associated with the new kinds of scattering which take place
\begin{itemize}

\item soliton-antisoliton/breather $\op{F}$-factor:
\begin{equation}
\label{minimal bs form factor}
\op{F}_{\mf{b}_k, \mf{s} } (\th) \, = \, \op{F}\big(\th +\tfrac{\i }{2} \mf{u}_k \big) \cdot \op{F}\big(\th - \tfrac{\i }{2} \mf{u}_k \big) \cdot
\big( a \phi\big)\big( \tfrac{\i }{2} \mf{u}_k -\th \big)  \cdot
\pl{p=1}{k} \bigg\{ \f{ \sinh\big( \tf{\th}{2}\big) }
{ \sinh\big[ \tfrac{1}{2} \big(\th- \i \tfrac{\mf{u}_k}{2} - \i p \xi \big) \big] }    \bigg\}
\;;
\end{equation}

\item breather/breather  $\op{F}$-factor:
\begin{equation}
\label{minimal bb form factor}
\op{F}_{\mf{b}_{k_1}, \mf{b}_{k_2}} (\th)\, = \,
  \f{ \op{F}_{\mf{b}_{k_1}, \mf{s} } \big(\th +\tfrac{\i }{2} \mf{u}_{k_2} \big)
  \cdot \op{F}_{\mf{b}_{k_1}, \mf{s} } \big(\th -\tfrac{\i }{2} \mf{u}_{k_2}\big)   }
{\sinh\big[ \tfrac{1}{2} \big(\th + \i \tfrac{\pi-\xi}{2}\big) \big]\sinh\big[ \tfrac{1}{2} \big(\th - \i \tfrac{\pi-\xi}{2}\big) \big] }
\cosh^2\big( \tfrac{\th}{2}\big) \;.
\end{equation}
\end{itemize}

One may check that these $\op{F}$-factors solve the below set of equations
\beq
\op{F}_{\mf{b}_k, \mf{s} } (\th) \, = \, S_{\mf{b}_k, \mf{s} }(\th)  \op{F}_{\mf{b}_k, \mf{s} } (-\th) \, = \,  \op{F}_{\mf{b}_k, \mf{b}_{\ell} } (2\i\pi-\th)  \qquad \e{and} \qquad
\op{F}_{\mf{b}_k, \mf{b}_{\ell} } (\th) \, = \, S_{\mf{b}_k, \mf{b}_{\ell} }(\th) \op{F}_{\mf{b}_k, \mf{b}_{\ell} } (-\th) \, = \,  \op{F}_{\mf{b}_k, \mf{b}_{\ell} } (2\i\pi-\th)  \;.
\label{ecriture eqns F factors des breathers solitons et breathers breathers}
\enq

\noindent Further, new sets of functions arise which take into account the remainder of the interactions between the asymptotic breather particles:
\begin{itemize}

\item self-breather interaction:
\begin{equation}
\label{chi function}
\chi_{ \ell ; k}(\th) \, =  \,  \f{  (-1)^{\ell} \, \i^k \,   \sinh\big( \th-  \tfrac{\i }{2} \mf{u}_k -\i \ell \xi\big) }
{2\sinh\big[ \tfrac{1}{2} \big( \th-  \tfrac{\i }{2} \mf{u}_k -\i k \xi\big) \big]
\sinh\big[ \tfrac{1}{2} \big( \th-  \tfrac{\i }{2} \mf{u}_k \big) \big]}
\prod_{p=1}^{k-1} \coth\big[ \tfrac{1}{2} \big( \th-  \tfrac{\i }{2} \mf{u}_k -\i p \xi\big) \big]\;,
\end{equation}
and
\begin{equation}
\label{rho function}
\rho_{\ell;k}(\th)  = \i^{\ell}
\pl{p=1}{\ell} \bigg\{ \f{\cosh\big[ \tfrac{1}{2} \big( \th - \tfrac{\i }{2} \mf{u}_k -\i(p-1)\xi \big) \big]}
{ \sinh\big( \tf{\th}{2}\big) }  \bigg\}
\pl{p=\ell+1}{k}   \bigg\{ \f{ \sinh\big[ \tfrac{1}{2} \big(\th- \i \tfrac{\mf{u}_k}{2} - \i p \xi \big) \big] }{ \sinh\big( \tf{\th}{2}\big) }
   \bigg\}  \;;
\end{equation}
\item breather-breather interaction:
\bem
\label{def lambda}
\La\Big( \,  ^{k_1,k_2}   _{\ell_1,\ell_2}  \mid \th\Big) = \rho_{\ell_1;k_1}\big(\th +\tfrac{\i }{2} \mf{u}_{k_2}\big) \cdot
\rho_{\ell_1;k_1}\big(\th - \tfrac{\i }{2} \mf{u}_{k_2} \big) \cdot \chi_{\ell_1;k_1}\big(\th + \tfrac{\i }{2} \mf{u}_{k_2} + \i \ell_2 \xi\big)  \\
\times  \f{\sinh\big[ \tfrac{1}{2} \big(\th + \i \tfrac{\pi-\xi}{2}\big) \big]\sinh\big[ \tfrac{1}{2} \big(\th - \i \tfrac{\pi-\xi}{2}\big) \big] }
{\cosh^2\big( \tfrac{\th}{2}\big) }  \;.
\end{multline}
\end{itemize}

All of these functions arise in the definition of the $\mf{h}_{\e{tot}}$ analogue
of the function $g$ \eqref{definition fonction g} appearing in the description of solutions in the
soliton/antisoliton sector. The latter depends on the collection of breather indices $\bs{k}_s=(k_1,\dots, k_s)$, breather
rapidities $\bs{\a}_s$, soliton/antisoliton rapidities $\bs{\be}_n$ and integration variables $\bs{u}_m$:
\bem
\mathdutchcal{h}_{\bs{\ell}_s;\bs{k}_s}\big( \bs{\a}_s, \bs{\be}_n ; \bs{u}_m\big)  \, = \,   g(\bs{\be}_{n};\bs{u}_{m})
\pl{j<t}{s} \op{F}_{\mf{b}_{k_j}, \mf{b}_{k_{t}}} (\alpha_j - \alpha_{t})
\pl{j=1}{s} \pl{c=1}{n} \op{F}_{\mf{b}_{k_j}, \mf{s} } (\alpha_j - \be_c)
\pl{j<t}{s} \La\Big( \, _{\ell_j, \, \ell_{t}} ^{k_j,\, k_{t}} \mid \alpha_j - \alpha_{t} \Big)
		 \\
\times  \pl{j=1}{s} \pl{b=1}{n} \rho_{\ell_j;k_j}(\alpha_j - \be_b)  \, \pl{j=1}{s} \pl{t=1}{m}   \chi_{\ell_j;k_j}(\alpha_j - u_{t})  \, .
\label{definition fonction h des breathers}
\end{multline}

\subsubsection{The integration curve $\msc{C}_{ \bs{\a}_s,\bs{\be}_n}^{(\bs{k}_s)}$}
\label{SousSousSectionCourbeIntegration}

Given rapidities $(\bs{\a}_s,\bs{\be}_n)\in \Cx^s \times \Cx^n$ with generic coordinates,
and a vector of integers $\bs{k}_s \in \mathbb{N}^s$
the contour $\msc{C}_{ \bs{\a}_s,\bs{\be}_n}^{(\bs{k}_s)}$  refers to any curve which satisfies the properties
\begin{itemize}

\item $\msc{C}_{ \bs{\a}_s,\bs{\be}_n}^{(\bs{k}_s)}$ connects $-\infty$ to $+\infty$;

\item  the points $\be_a-\i p \xi - 2 \i \pi \ell $, $\be_a-\i\pi -\i (p +1) \xi$, $p, \ell \in \mathbb{N}$,
have index $-1$ in respect to $\msc{C}_{ \bs{\a}_s,\bs{\be}_n}^{(\bs{k}_s)}$, \textit{viz}. they are located "below" of $\msc{C}_{ \bs{\a}_s,\bs{\be}_n}^{(\bs{k}_s)}$;

\item  the points $\be_a+\i p \xi + (2 \ell + 1) \i \pi$, $\be_a-\i\pi +\i p \xi$,  $p, \ell \in \mathbb{N}$,
have index $1$ in respect to $\msc{C}_{ \bs{\a}_s,\bs{\be}_n}^{(\bs{k}_s)}$, \textit{viz}. they are located above of $\msc{C}_{ \bs{\a}_s,\bs{\be}_n}^{(\bs{k}_s)}$.

\item for $t \in \intn{1}{s}$,  the points $\a_t +\tfrac{\i}{2} \mf{u}_{k_t} -\i p \xi - 2 \i \pi \ell $,
$\a_t -\tfrac{\i}{2} \mf{u}_{k_t} -\i p \xi - 2 \i \pi (\ell+1)$, $\a_t \pm \tfrac{\i}{2} \mf{u}_{k_t} - \i\pi -\i (p +1) \xi$,
$p, \ell \in \mathbb{N}$, and  $\a_t -\tfrac{\i}{2} \mf{u}_{k_t} -\i p \xi $ with $p \geq k_t+1$
have index $-1$ in respect to $\msc{C}_{ \bs{\a}_s,\bs{\be}_n}^{(\bs{k}_s)}$, \textit{viz}. they are located "below" of $\msc{C}_{ \bs{\a}_s,\bs{\be}_n}^{(\bs{k}_s)}$;

\item for $t \in \intn{1}{s}$, the points $\a_t \pm \tfrac{\i}{2} \mf{u}_{k_t}  +\i p \xi + (2 \ell + 1) \i \pi$, $\a_t \pm \tfrac{\i}{2} \mf{u}_{k_t} -\i\pi +\i p \xi$,  $p, \ell \in \mathbb{N}$
and $\a_t -\tfrac{\i}{2} \mf{u}_{k_t} -\i p \xi $ with $p \in \intn{0}{k_t}$
have index $1$ in respect to $\msc{C}_{ \bs{\a}_s,\bs{\be}_n}^{(\bs{k}_s)}$, \textit{viz}. they are located "above" of $\msc{C}_{ \bs{\a}_s,\bs{\be}_n}^{(\bs{k}_s)}$.

\end{itemize}

\subsubsection{The general solution}

\begin{prop}
\label{Proposition FF dans le secteur breather}
Let $q = n-2m \in \mathbb{Z}$, $\op{s}$ and $\om$ be all fixed and $p$ a family of functions solving the associated $p$-function problem.
Then, the unique solution $\bs{\Phi}[p](\bs{\la}_{\ell})$, $\ell \in \mathbb{N}$ to the axioms $I)-V)$ whose restriction to the soliton/antisoliton,
\textit{c.f.} \eqref{form factor coordinates}, is given by  $\bs{\Psi}_{1,\dots,n}[p](\bs{\be}_{n})$, $n \in \mathbb{N}$, is fully characterised by
its projection on the breather coordinates, \textit{c.f.} \eqref{definition bs Brether projection},
\beq
\bs{\mf{b}}_{s} \, = \, \big( \mf{b}_{k_1}, \dots, \mf{b}_{k_s} \big)\;, \quad s \in \mathbb{N} \;, \; k_a \in \intn{ 1 }{ n_{\xi} }
\quad and \quad
\bs{k}_{s} \, = \, \big( k_1,\dots, k_s \big)
\label{ecriture coordonnees breather generales}
\enq
which is defined as
\beq
\bs{\Phi}_{1,\dots,n}^{\bs{\mf{b}}_{s}} [p]  \big( \bs{\a}_s, \bs{\be}_n \big) \; = \;
\sul{ \bs{\veps}_{n} \in \{ \pm  \}^n }{}
\Big( \bs{\Phi}[p] \big( \bs{\a}_s, \bs{\be}_n \big) \cdot \op{v}^{ \bs{\mf{b}}_{s}}\otimes \op{v}^{  \wt{\bs{\veps}}_n} \Big)
\cdot \big( \op{v}^{\bs{\veps}_n}  \big)^{ \mathtt{t} } \;,
\label{definition projection solution generale sur coordonnees breather}
\enq
with $\wt{\veps}_a = \mf{s}/\ov{\mf{s}}$ if $ \veps_a = +/-$. Given $\mc{I}_{\bs{k}_s}=\intn{0}{k_1}\times \cdots \times \intn{0}{k_s}$,
these are explicitly expressed as
	\beq
		\label{exp generale form factor}
\bs{\Phi}_{1,\dots,n}^{\bs{\mf{b}}_{s}} [p]  \big( \bs{\a}_s, \bs{\be}_n \big)   =
  \sul{ \bs{\ell}_s \in \mc{I}_{\bs{k}_s} }{ }    (-1)^{\ov{\bs{\ell}}_s}  \hspace{-3mm}
\Int{ \big( \msc{C}_{ \bs{\a}_s,\bs{\be}_n}^{(\bs{k}_s)} \big)^m  }{} \hspace{-3mm} \dd^{m} u
 \, \wt{\Cx}_{1,\dots,n}^{\,(\ga)}\big( \bs{\be}_n \, ; \, \bs{u}_m\big)
\cdot \op{O}_{1,\dots,n}^{(\ga)}(\bs{\be}_n) \cdot  \mathdutchcal{h}_{ \bs{\ell}_s  ;  \bs{k}_s  }\big( \bs{\a}_s, \bs{\be}_n ; \bs{u}_m\big)
\cdot p_{ \bs{\ell}_s  ;  \bs{k}_s  } \big(\bs{\alpha}_s , \bs{\be}_n;\bs{u}_m \big) \;.
\enq
Above, we have set $\ov{\bs{\ell}}_s=\sul{t=1}{s}\ell_t$ while the soliton/antisoliton integration contour
is as introduced in Subsubsection \ref{SousSousSectionCourbeIntegration}.
Finally, $s$-breathers p function appearing above takes the explicit form
\beq
\label{s breathers p function}
p_{ \bs{\ell}_s  ;  \bs{k}_s  }  \big(\bs{\alpha}_s, \bs{\be}_n;\bs{u}_m \big) = \f{(m+s)!}{m!}
p \Big(\bs{\mu}_{n+2s} ; \alpha_1 - \tfrac{ \i }{ 2 } \mf{u}_{k_1} -\i \ell_1 \xi,
\dots,\alpha_s - \tfrac{ \i }{ 2 } \mf{u}_{k_s} -\i \ell_s \xi,\bs{u}_m \Big)
 \pl{j=1}{s} d_{\ell_j;k_j} \;.
\enq
Its definition involves the concatenated variable
\beq
\bs{\mu}_{n+2s} \, = \, \Big(\alpha_1 + \tfrac{ \i }{ 2 } \mf{u}_{k_1} ,\alpha_1 - \tfrac{ \i }{ 2 } \mf{u}_{k_1} ,\dots,
\alpha_s +   \tfrac{ \i }{ 2 } \mf{u}_{k_s} ,\alpha_s - \tfrac{ \i }{ 2 } \mf{u}_{k_s} ,  \bs{\be}_n  \Big) \;,
\enq
and the scalar factors
\beq
\label{d factors}
d_{ \ell; k} = \i\varkappa \op{F}\big(\i \mf{u}_{k}\big) \ex{ \f{\i\ga\xi}{2} (2\ell-k) }
\f{  \Big\{ \tfrac{1}{2} \tan \big( \tfrac{k\xi}{2} \big)  \Big\}^{ \f{1}{2} } \prod\limits_{p=1}^{k-1}\tan \Big( \tfrac{p\xi}{2} \Big)}
{\cos\Big(\tfrac{ \ell \xi}{2}\Big)\cos\Big(\tfrac{ (k-\ell) \xi}{2}\Big)}
\cdot \prod_{p=1}^{\ell}\cot\Big( \tfrac{ p \xi}{2}\Big) \prod_{p=1}^{k-\ell}\cot\Big( \tfrac{ p \xi}{2}\Big) \;;
\enq

\end{prop}

We postpone the proof of this result to Appendix \ref{Appendix preuve decomposition breathers}.
When the twist parameter $\ga$ is set to zero, the above formula was obtained in the case of $\bs{k}_s=(1,\dots,1)$ in \cite{BabujianKarowskiExactFFSineGordonBootsstrapII}.

\subsection{The restriction to the first breather sector}
\label{Subsection General FF rerstriction to type 1 breathers}

As observed in \cite{BabujianKarowskiExactFFSineGordonBootsstrapII}, the
formula for the general off-shell Bethe Ansatz form factors takes a particularly simple form when focusing on the theory's subsector built-up solely
in terms of lowerst order breathers, \textit{i.e.}  when one sets $n=0$ and $\bs{k}_s=(1,\dots,1)$, $\bs{\mf{b}}_s=(\mf{b}_1,\dots,\mf{b}_1)$ in
\eqref{exp generale form factor}. Observe that, formally speaking, the Sine-Gordon Lagrangian goes to the Sinh-Gordon Lagrangian
when rotating the coupling constant by $\i$, \textit{i.e.} $g \hookrightarrow \i g$. This transformation reflects itself in the
replacement $\xi \hookrightarrow -2\pi \mf{b}$ of the Sine-Gordon coupling constant $\xi$ in terms of the Sinh-Gordon one $\mf{b} \in \intff{0}{\tf{1}{2}}$.
Now, the matter is that under this replacement, the type-1 breather/ type-1 breather scattering matrix $S_{\mf{b}_1, \mf{b}_1}$ reduces to the one 1+1 dimensional Sinh-Gordon model:
\beq
S_{\mf{b}_1, \mf{b}_1}(\th)=  \f{ \tanh\big[\tfrac{1}{2} ( \th  + \i  \xi ) \big]   }{ \tanh\big[\tfrac{1}{2} ( \th - \i \xi ) \big] } \;\;
\hookrightarrow \; \;  \f{ \tanh\big[\tfrac{1}{2} ( \th  - 2\i \pi   \mf{b} ) \big]   }{ \tanh\big[\tfrac{1}{2} ( \th  +2\i \pi   \mf{b}  ) \big] }
\enq
Since, the analytic continuation of $S_{\mf{b}_1, \mf{b}_1}$ in $\xi$ does not alter the functional form of the
$\op{F}$-factor equations \eqref{ecriture eqns F factors des breathers solitons et breathers breathers}, the analytic continuation of
$\op{F}_{\mf{b}_1,\mf{b}_1}$ in $\xi \hookrightarrow -2\pi \mf{b}$ should thus produce the $\op{F}$-factor associated with the
Sinh-Gordon $\op{S}$-matrix. This will indeed be the case. Since the comparison of the results restricted to the  type-1 breather sector
plays a role in the identification of our form factors with those of the exponential of the field, we reproduce the proof of the reduction here.
Note that the difference with the representation found in \cite{BabujianKarowskiExactFFSineGordonBootsstrapII} stems from the
incorporation of the twist effects into the reduced $p$-function.

\begin{prop}
Given $s \in \mathbb{N}$,  the $s$ type 1 breathers form factor associated with a solution $p$ to the $p$-function problem takes the explicit form:
	\begin{equation}
\bs{\Phi}^{\bs{\mf{b}}_{s}} [p]  \big( \bs{\a}_s \big) \; = \; \pl{j<t}{s} \op{F}_{\mf{b}_1, \mf{b}_1} (\alpha_j - \alpha_{t})
 \sul{\bs{\ell}_s \in \{0,1\}^s}{}   (-1)^{ \ov{\bs{\ell}}_s}
 \pl{j<t}{s} \bigg\{  1 + \i  \f{ (\ell_j-\ell_{t})\sin(\xi) }{\sinh(\alpha_j - \alpha_{t})} \bigg\}
\cdot p_{\bs{\ell}_s ; \bs{k}_s} \big(\bs{\alpha}_s; \emptyset \big) \;,
	\end{equation}
$\bs{k}_s=(1,\dots,1)\in \mathbb{N}^s$, and where the minimal form factor $ \op{F}_{\mf{b}_1, \mf{b}_1}$ is expressible in terms of Barnes' $G$ functions
\begin{equation}
\op{F}_{\mf{b}_{1}, \mf{b}_{1}} (\th) = \f{1}{\Ga \Big( 1-\tfrac{\i\th}{2\pi} \, , \, \tfrac{\i\th}{2\pi} \Big) }
\cdot G\Bigg( \ba{cccc} 1 + \nu_-\, ,  &   2 + \nu_+ \,,  &   \tfrac{1}{2} - \nu_+ \, ,
				&   \tfrac{3}{2} - \nu_-     \\
 - \nu_+  \, ,  &   1 - \nu_-   \, ,  &   \tfrac{1}{2} +  \nu_-  \, ,
 &   \tfrac{3}{2} + \nu_+   \ea \Bigg) \;, \quad  \nu_{\pm} \; = \,  \f{\xi \pm \i\th}{2\pi} \;.
\label{ecriture minimal FF b1b1}
\end{equation}
which reduces to the $F$-factor of the Sinh-Gordon theory upon the very same replacement of $\xi\hookrightarrow -2\pi \mf{b}$.
\end{prop}

One can readily check that \eqref{ecriture minimal FF b1b1} reduces to the Sinh-Gordon $\op{F}$ factor upon the substitution $\xi \hookrightarrow -2\pi \mf{b}$.
In fact, the expression taken as a whole reduces upon the very same substitution to the summand proposed, within the angular quantization approach, in
\cite{BrazhnikovLukyanovFreeFieldRepMassiveFFIntegrable} for the form factors of the exponential of the field, this up to the form of the $p$-function contribution which we shall discuss later on.
This as well reproduces the $\mc{K}$-transform of $p$-functions representation for the form factors of the Sinh-Gordon model developed in \cite{BabujianKarowskiBreatherFFSineGordon}.
Note that the whole dependence on $\ga$ is now only contained in the reduced $p$-function.

\proof
One starts from the general formula for form factors \eqref{exp generale form factor} without any solitons, and $s$ identical type-1 breathers, \textit{viz}.
$\bs{k}_s=(1,\dots,1)$. This yields:
\begin{equation}
\bs{\Phi}^{\bs{\mf{b}}_{s}} [p]  \big( \bs{\a}_s \big) \; = \; \pl{j<t}{s} \op{F}_{\mf{b}_1, \mf{b}_1} (\alpha_j - \alpha_{t})
 \sul{\bs{\ell}_s \in \{0,1\}^s}{}   (-1)^{ \ov{\bs{\ell}}_s}
 \pl{j<t}{s} \La\Big( \, _{\ell_j, \, \ell_{t}} ^{1 \, 1} \mid \alpha_j - \alpha_{t} \Big)
\cdot p_{\bs{\ell}_s ; \bs{k}_s} \big(\bs{\alpha}_s; \emptyset \big) \;.
\end{equation}
To simplify the raw form of $ \La\Big( \, _{\ell_j, \, \ell_{t}} ^{1 \, 1} \mid \th \Big)$ as given by \eqref{def lambda},
one observes that, owing to $\ell_1$ only taking values 0 or 1, one has:
\begin{equation}
	\chi_{\ell_1 ; 1}\big(\th + \tfrac{ \i }{ 2 } \mf{u}_{1} + \i \ell_2 \xi \big) \, = \,
\i (-1)^{\ell_1} \f{\cosh\big[ \tfrac{1}{2}( \th + \i (\ell_2 -\ell_1) \xi)\big] }{\sinh\big[ \tfrac{1}{2}( \th + \i (\ell_2 + \ell_1 -1 ) \xi)\big]} \;.
\end{equation}
Moreover, the two remaining building blocks, \textit{c.f.} \eqref{rho function}, upon disjoining $\ell=0$ and $\ell=1$, eventually simplify to:
\begin{equation}
	\rho_{\ell_1;1}\big(\th +\tfrac{ \i }{ 2 } \mf{u}_{1} \big) \cdot \rho_{\ell_1;1}\big(\th -\tfrac{ \i }{ 2 } \mf{u}_{1} \big) \, =  \,
-\i (-1)^{\ell_1} \f{\sinh\big[ \tfrac{1}{2}(\th +  (2\ell_1-1)\i\xi) \big] \cosh\big( \tfrac{\th}{2} \big) }
{ \sinh\big[ \tfrac{1}{2} \big(\th + \i \tfrac{\pi-\xi}{2}\big) \big] \sinh\big[ \tfrac{1}{2} \big(\th - \i \tfrac{\pi-\xi}{2}\big) \big] }   \;.
\end{equation}
Then, by distinguising whether $\ell_1 = \ell_2$ or $\ell_1 = 1-\ell_2$, one gets that
\begin{equation}
 \La\Big( \, _{\ell_j, \, \ell_{t}} ^{1 \, 1} \mid \th \Big) \, = \,  1 + \i (\ell_1-\ell_2) \f{ \sin(\xi) }{ \sinh(\th) } \;.
\end{equation}
It remains to get the explicit expression for $\op{F}_{\mf{b}_1, \mf{b}_1}$. One first needs to recast \eqref{minimal bs form factor}. By using the functional equation for
$\op{F}$ \eqref{ecriture eqn fnelle FF minimal}
and the definition of $\phi$ \eqref{definition fonction phi}, one gets:
\begin{equation}
\op{F}_{\mf{b}_1, \mf{s}} (\th) \, = \,
%
%
%
\f{ \op{F} \big(\th + \tfrac{ \i }{ 2 } \mf{u}_{1} \big)  }{ \op{F}\big( \th +\i \tfrac{\pi+\xi}{2} \big)  }	  \cdot
\f{ \sinh\big( \tfrac{\th}{2}\big) }{ \sinh\big[ \tfrac{1}{2} \big(\th - \i \tfrac{\pi+\xi}{2}\big) \big] } \;.
\end{equation}
Inserting the above representation into \eqref{minimal bb form factor} and using \eqref{def F function}-\eqref{int rep F} leads to:
\begin{equation}
\op{F}_{\mf{b}_{1}, \mf{b}_{1}} (\th) \, =  \,
\f{ \sinh\big[ \tfrac{1}{2}(\th +\i\pi - \i\xi)\big]\sinh\big(\tfrac{\th}{2}\big) }{\sinh\big[ \tfrac{1}{2}(\th + \i\xi)\big]\sinh\big[ \tfrac{1}{2}(\th- \i\xi)\big] }
\cdot \exp\Bigg\{  -\Int{0}{+\infty} \f{\dd t }{t} \f{ \sinh[ (\pi -\xi ) t]  \sinh\big[\tfrac{\pi t }{2 }(1-\tfrac{2\th}{\i\pi}) \big]  }
	{ \cosh\big[\tfrac{\pi t }{2 } \big]   \sinh [ \pi t  ] } \Bigg\} \;.
\end{equation}
The exponential factor can now be recast solely in terms of ratios of Barnes-$G$ functions by using their integral integral representation given in Appendix \ref{Appendice Special Functions}:
\beq
 \exp\Bigg\{  -\Int{0}{+\infty} \f{\dd t }{t} \f{ \sinh[ (\pi -\xi ) t]  \sinh\big[\tfrac{\pi t }{2 }(1-\tfrac{2\th}{\i\pi}) \big]  }
	{ \cosh\big[\tfrac{\pi t }{2 } \big]   \sinh [ \pi t  ] } \Bigg\}  \; = \;
 G\Bigg( \ba{cccc} \nu_-\, ,  &   1 + \nu_+ \,,  &   \tfrac{3}{2} - \nu_+ \, ,
				&   \tfrac{3}{2} - \nu_-     \\
 2- \nu_+  \, ,  &   1 - \nu_+   \, ,  &   \tfrac{1}{2} +  \nu_+  \, ,
 &   \tfrac{1}{2} + \nu_-   \ea \Bigg) \;,
\enq
 with $\nu_{\pm}$ as in \eqref{ecriture minimal FF b1b1}. This then yields  \eqref{ecriture minimal FF b1b1}
 upon using the shift properties of $G$ and Euler's reflection formula for $\Ga$.  \qed




\section{Form factors of the exponential of the field and of the related soliton creating operators}
\label{Section FF exp champ}

We first introduce a family of $p$-functions. Given $\bs{\be}_n \in \mathbb{C}^n$ and $\bs{u}_m \in \mathbb{C}^m$, let
\begin{equation}
p^{(\ga)}(\bs{\be}_n ; \bs{u}_m) = \f{1}{m!}\Big( \f{\i}{\varkappa} \Big)^m
\label{p function SG}
\end{equation}
in which the $\ga$-dependent constant $\varkappa$ is as introduced in \eqref{definition cste varkappa}. It is direct to see that $p^{(\ga)}$ satisfies all the $p$-functions axioms a)-e).

\begin{conj}
\label{Conjecture FF exponentielle du champ}
Let $n,m \in \mathbb{N}$ and $\ga\in \Cx$ be such that $\abs{\Re(\ga)} \leq 1$. Then, given $p^{(\ga)}$ as in \eqref{p function SG}, set
\beq
\bs{\Psi}_{1,\dots,n}\big[ p^{(\ga)}\big]\big( \bs{\be}_n \big) \, =   \hspace{-3mm} \Int{ (\msc{C}_{\bs{\be}_n})^m }{} \hspace{-2mm}  \dd^{m} u \;
\wt{\Cx}_{1,\dots,n}^{\, (\ga)}\big(\bs{\be}_n;\bs{u}_m \big) \cdot \op{O}_{1,\dots,n}^{(\ga)}(\bs{\be}_n)  \cdot \big( g \cdot p^{(\ga)} \big)\big(\bs{\be}_n;\bs{u}_{m} \big)  \;.
\label{form factor exponential}
\enq
The above provides one with the form factors
\begin{itemize}
\item[i)] of the exponential of the field operator $ \ex{\i \f{g\ga}{2} \Psi}$ in the $n$ soliton/antisoliton sector for $n \in 2\mathbb{N}$ and
$m=n/2$, all its other form factors being zero;
\item[ii)] of the topologically charged, soliton creating, operators connected with the exponential of the field in the sense of \cite{LukyanovZamolodchikovSolitonCreatingOpsSineG},
with charge $q=n-2m$.

\end{itemize}

\end{conj}

From this expression, one readily obtains the form factors involving breather particles, using the general expression \eqref{exp generale form factor}.
The remainder of this section is devoted to arguing the validity of the conjecture by comparing the above general representation
with the various representations which have been proposed for i) or ii) in the litterature for specific values of $n$ and $m$
along with the results obtained at the free fermion point $\xi=\pi$. We stress that the representation appears, however, as the first one that covers systematically and explicitly all
sectors of soliton/antisoliton excitations.
Below, we compare with the formulae argued by Lukyanov \cite{LukyanovConjectureFFExponentialFieldSineGSolASolAndBReather}
for the form factors of $ \ex{\i \f{g\ga}{2} \Psi} $ in the two soliton/antisoliton sector along with the
formulae obtained in \cite{BernardLeclairDiffEqnForPIII,SmirnovFormFactors} at the free fermion point $\xi=\pi$ this for any
soliton/antisoliton sector.
Finally, we also compare with the formulae for the soliton creating operators argued in \cite{LukyanovZamolodchikovSolitonCreatingOpsSineG}
within the angular quantization approach.

\subsection{The  type-$1$ breather sector}

\begin{lemme}
Let $\bs{\a}_s \in \mathbb{C}^s$ and $\bs{k}_s =(1,\dots, 1)\in \mathbb{N}^s$, Then, upon choosing the base $p$-function as in \eqref{p function SG},
it holds
\begin{equation}
p^{(\ga)}_{\bs{\ell}_s;\bs{k}_{s} }(\bs{\a}_s;\emptyset) \,  =  \, (-1)^s  \Big\{ 2\sin(\tfrac{\xi}{2}) \Big\}^{-\f{s}{2} }
 \exp\bigg\{ \f{-s}{2\pi} \Int{0}{\xi} \f{t \, \dd t}{\sin t} \bigg\}   \cdot \pl{r=1}{s} \ex{-\i \f{\xi}{2} \ga (-1)^{\ell_{r}}} \;.
\end{equation}
in which $p^{(\ga)}_{\bs{\ell}_s;\bs{k}_{s} }$  is as defined through \eqref{s breathers p function}.

\end{lemme}

The result of the lemma thus shows that upon the Sine to Sinh-Gordon formal correspondence $\xi \hookrightarrow -2\pi \mf{b}$ the type-$1$ breather sector $p$-function
associated with the Sine-Gordon exponential of the field maps directly onto the expression for the $p$ function representing the Sinh-Gordon
exponential of the field as argued within the angular quantization approach \cite{BrazhnikovLukyanovFreeFieldRepMassiveFFIntegrable}
and within the $\mc{K}$-transform approach to the Sinh-Gordon model's bootstrap program.

\proof

By explicitating the expression for $p^{(\ga)}_{\bs{\ell}_s;\bs{k}_{s} }$ in the case of interest,
one gets:
\begin{equation}
p^{(\ga)}_{\bs{\ell}_s;\bs{k}_{s} }(\bs{\a}_s;\emptyset) \, =  \,
\Bigg\{ \f{ - \op{F}(\i\pi - \i \xi) }{   \big[ \sin (\tfrac{\xi}{2}) \cos (\tfrac{\xi}{2}) \big]^{\tf{1}{2}} } \Bigg\}^s
\pl{r=1}{s}  \ex{\i \f{\ga \xi}{2} (-1)^{\ell_r} }   \;.
\end{equation}
Note that above, we used that since $\ell_{r} \in \{0,1\}$, $\ell_{r} \, = \,  \tfrac{ 1 }{2} \big( 1-(-1)^{\ell_{r}} \big)$.
By using the integral representation for $\op{F}$ \eqref{def F function}-\eqref{int rep F},  one gets
\beq
 \op{F}(\i\pi - \i \xi)  \, = \, \cos (\tfrac{\xi}{2})
\exp\Bigg\{  -  2 \Int{0}{+\infty} \f{\dd t }{t} \f{ \sinh\big[ \tfrac{t}{4} \big(  1 - \tfrac{ \xi }{\pi } \big) \big]   \sinh\big[\tfrac{t \xi }{4 \pi  }\big] \sinh\big[\tfrac{t }{4 }\big]  }
{    \sinh^2\big[\tfrac{ t }{2 } \big]  } \Bigg\} \;.
\enq
This can be recast in terms of Barnes $G$-functions by using \eqref{ecriture rep int Barnes} what yields
\beq
 \op{F}(\i\pi - \i \xi)  \, = \,  \cos (\tfrac{\xi}{2}) \,  G\Bigg( \ba{cccc}  1+ \nu\, ,  &   \tfrac{3}{2} - \nu \,,  &   \tfrac{1}{2}   \\
 1- \nu  \, ,  &  \tfrac{1}{2} +  \nu  \, , &   \tfrac{3}{2}    \ea \Bigg) \;.
\enq
It then remains to invoke \eqref{ecriture rep int ratio barnes} so as to reduce the above to
\beq
	\op{F}(\i \pi - \i \xi) \, =  \, \Big\{ \cos\big(\tfrac{\xi}{2}\big)  \Big\}^{\f{1}{2}} \cdot \exp\bigg\{ \f{-1}{2\pi} \Int{0}{\xi} \f{t\dd t}{\sin t} \bigg\} \;.
\enq
This entails the claim. \qed

\subsection{The free fermion case}

In this subsection, we prove that in the free fermion limit $\xi \searrow \pi$ the form factors given in Conjecture \ref{Conjecture FF exponentielle du champ}
indeed yield the free fermion form factor of the exponential of the field as obtained in \cite{BernardLeclairDiffEqnForPIII,SchroeTruongFFExpPhyInFreeFermionPtSineG,SmirnovFormFactors}.

\begin{prop}
\label{prop free fermion limit}
Let $n=2m$ and $\ga\in \Cx$ such that $\abs{\Re(\ga)} \leq 1$. Then, it holds that
\begin{equation}
\bs{\Psi}_{1,\dots,n}\big[p^{(\ga)}\big]\big(\bs{\be}_n) \underset{\xi  \searrow  \, \pi}{\longrightarrow}  \bs{\Psi}^{\textrm{FF}}_{1,\dots,n}\big[ p^{(\ga)} \big]\big(\bs{\be}_n) \;,
\end{equation}
where $\bs{\Psi}^{FF}_{1,\dots,n}\big[p^{(\ga)}\big]$ is the form factor obtained in the free field case \cite{BernardLeclairDiffEqnForPIII,SchroeTruongFFExpPhyInFreeFermionPtSineG,SmirnovFormFactors}:
\begin{equation}
\label{def form factor free fermion}
\bs{\Psi}^{\textrm{FF}}_{1,\dots,n}\big[p^{(\ga)}\big]\big(\bs{\be}_n) \, =  \,
(-1)^{\f{m(m-1)}{2}} \Big\{\i\sin\Big(\tfrac{\ga\pi}{2}\big)\Big\}^m
\hspace{-4mm} \sul{ \bs{\veps}_n \in \{\pm \}^n \in \mc{E}_n  }{}  \hspace{-3mm} \,   \ex{\f{\ga}{2} \bs{\veps}_n \cdot \bs{\be}_n } \cdot
\f{\pl{\substack{b>a \\ \veps_b=\veps_a = +}}{}\sinh\left(\tfrac{\be_{ba}}{2}\right)
\pl{\substack{b>a \\ \veps_b=\veps_a = -}}{}\sinh\left(\tfrac{\be_{ba}}{2}\right)    }
{ (-1)^{\# \mc{E}^{-+}_{<} } \pl{\substack{b=1 \\ \veps_b = +}}{n} \pl{\substack{a=1 \\ \veps_a = -}}{n} \cosh\left(\tfrac{\be_{ba}}{2}\right)}
\Big( \op{v}^{\veps_1}_1  \dots \op{v}^{\veps_n}_n \Big)^{\op{t}} \;.
\end{equation}
Here we have introduced $\mc{E}^{-+}_{<} \, = \, \Big\{ (a,b) \, : \, \veps_a = - \, , \;  \veps_b=+ \, , \, a<b \Big\} $ and
\beq
\mc{E}_n \, = \, \Big\{  \bs{\veps}_n \in \{\pm \}^n \, : \, \# \{a \, : \, \veps_a=+\} \, = \,   \# \{a \, : \, \veps_a=-\}  = m \Big\} \;.
\label{definition En}
\enq

\end{prop}

\proof

Observe that the integration contour $	\msc{C}_{ \bs{\be}_n }$ is pinched by the poles at $\be_a-\i(\pi-\xi)$ (index $1$) and  $\be_a$ (index $-1$)
as well as $\be_a-\i\xi$ (index $-1$) and $\be_a-\i\pi$ (index $1$) in the $\xi \tend \pi$ limit. Similarly as in the proof of the kinematic poles residues in Appendix \ref{subsection axiom iv}, we deal with those pinchings by introducing the
auxiliary contour $\widetilde{\msc{C}}_{\bs{\be}_n}$:
\begin{itemize}
	
	\item $\widetilde{\msc{C}}_{\bs{\be}_n}$ connects $-\infty$ to $+\infty$

	\item  the points  $\be_a-\i\pi -\i (p +1) \xi$, $p, \ell \in \mathbb{N}$,  
	have index $-1$ in respect to $\widetilde{\msc{C}}_{\bs{\be}_n}$,
	
	\item the points $\be_a-\i p \xi - 2 \i \pi \ell $, $p, \ell \in \mathbb{N}$, have index $1$ if $(p,\ell)=(1,0)$ and index $-1$ otherwise,
	
	\item  the points $\be_a+\i p \xi + (2 \ell + 1) \i \pi$,  $p, \ell \in \mathbb{N}$, 
	have index $1$, 
	
	\item the points $\be_a-\i\pi +\i p \xi$, $p \in \mathbb{N}$ have index $1$ if $p \neq 1$ and index $-1$ if $p=1$.
\end{itemize}
$\widetilde{\msc{C}}_{\bs{\be}_n}$ has the advantage of not separating poles that will collapse to each other in the $\xi \tend \pi$ limit:
$\be_a - \i \pi + \i \xi \rightarrow \be_a$ and $\be_a - \i \xi \rightarrow \be_a - \i \pi$. We thus have the homotopy relation:
\begin{equation}
\msc{C}_{ \bs{\be}_n } \setminus \widetilde{\msc{C}}_{\bs{\be}_n}   \; \simeq \;
\bigcup\limits_{a=1}^{n}   \Big\{ \ \Dp{}\mc{D}_{\be_a-\i \pi + \i  \xi ,\eps}
\cup   - \Dp{}\mc{D}_{\be_a- \i \xi ,\eps}    \Big\}  \;.
\end{equation}
By using the symmetry in $\bs{u}_m$ of the integrand in \eqref{form factor exponential}, we thus get
\begin{equation}
	\label{decompos succ res}
\bs{\Psi}_{1,\dots,n}\big[p^{(\ga)}\big]\big(\bs{\be}_n) = \sul{r=0}{m} \f{ m!  \, (2\pi\i)^r  }{(m-r)!}
\sul{  \bs{\ell}_r \in \intn{1}{n}^r }{}
\Int{ \big( \widetilde{\msc{C}}_{\bs{\be}_n} \big)^r }{}  \hspace{-3mm} \dd^{m-r} \nu
\, \mc{R}_r\big( \bs{\ell}_r ;\bs{\be}_n; \bs{\nu}_{m-r} \big)
\end{equation}
where we have defined the successive residues:
\begin{equation}
\mc{R}_r\big( \bs{\ell}_r ;\bs{\be}_n; \bs{\nu}_{m-r} \big) \, = \,  \sul{ \bs{j}_{r} \in \{0,1\}^r }{}  (-1)^{ \ov{\bs{j}}_r } \,
\e{Res}\Big( \wh{\Cx}_{1,\dots,n}^{(\ga)}\big( \bs{\be}_n \, ; \, \bs{u}_m\big) \dd^r u \;, \; u_s= \be_{\ell_s}^{(j_s)}, s=1,\dots r  \Big)
 \;,
\label{def succ res}
\end{equation}
in which, we  denoted for short
\beq
\wh{\Cx}_{1,\dots,n}^{(\ga)}\big( \bs{\be}_n \, ; \, \bs{u}_m\big) \; = \; \wt{\Cx}_{1,\dots,n}^{(\ga)}\big( \bs{\be}_n \, ; \, \bs{u}_m\big)
\cdot  \op{O}_{1,\dots,n}^{(\ga)}(\bs{\be}_n) \cdot  \big( g \cdot p \big)\big( \bs{\be}_n \, ; \, \bs{u}_m\big) \;,
\label{definition hat C}
\enq
so that 
\beq
\bs{\Psi}_{1,\dots,n}\big[ p\big]\big( \bs{\be}_n \big)  \; = \; \Int{ (\msc{C}_{\bs{\be}_n})^m }{} \hspace{-1mm} \dd^m u \, \wh{\Cx}_{1,\dots,n}^{(\ga)}\big( \bs{\be}_n \, ; \, \bs{u}_m\big)  \,. 
\enq
and have set
\beq
\be_a^{(0)}\, = \, \be_a - \i(\pi-\xi) \qquad \e{and} \qquad \be_a^{(1)}\, = \, \be_a - \i\xi \;.
\label{definition beta 0 et beta 1}
\enq
It is shown in Appendix \ref{appendix vanish} that
\beq
\lim_{\xi \tend \, \pi} \, \Int{ \big( \widetilde{\msc{C}}_{\bs{\be}_n} \big)^r }{}  \hspace{-3mm} \dd^{m-r} \nu
\, \mc{R}_r\big( \bs{\ell}_r ;\bs{\be}_n; \bs{\nu}_{m-r} \big) \; = \; 0
\enq
if $r<m$. From now on, we thus focus on $r=m$, \textit{i.e.} when solely residues are being computed.
In order to compute these residues in view of taking the $\xi \tend \pi$ limit, we need to expand various building blocks to
the first order in $\xi-\pi$. The computation of the first type of residue starts with:
\beq
\e{Res}\Big( \op{C}_{1,\dots,n}^{(\ga)}\big( \bs{\be}_n \, ; \,u\big) \cdot  \dd u \, , \, u= \be_{\ell} - \i\pi + \i\xi  \Big)  \, = \,
-\i (\xi-\pi)  \,  \op{E}^{12}_{\ell}  \cdot \Big( 1 + \e{O}\big(\xi-\pi\big) \Big)  \,,
\label{residu C}
\enq
in which $\op{E}^{12}$ is the $2\times 2$ elementary matrix whose sole non-zero entry is $1$ located in the top right corner. Further,
\begin{equation}
g \big(\bs{\be}_n;\bs{u}_{m-1}, \be_{\ell} - \i\pi + \i\xi \big)  \, =  \,
 2 \f{ g \big(\bs{\be}_n;\bs{u}_{m-1} \big) }{\xi-\pi} \cdot \pl{ \substack{k=1 \\ k \neq \ell} }{n} (a\phi)\big(\be_{k\ell}) \cdot
\pl{t=1}{m-1} \tau\big(u_t-\be_{\ell} )  \cdot \Big( 1 + \e{O}\big(\xi-\pi\big) \Big)\,,
\label{evaluation g ff}
\end{equation}
while obviously for any $u_m$ it holds
\begin{equation}
	p^{(\ga)} \big(\bs{\be}_n;\bs{u}_{m}\big) \, =  \, \f{\i}{m\varkappa}  p^{(\ga)} \big(\bs{\be}_n;\bs{u}_{m-1} \big)  \,.
\label{evaluation p ff}
\end{equation}
Combining \eqref{residu C}, \eqref{evaluation g ff} and \eqref{evaluation p ff} while using \eqref{definition operateur spin twist} yields
\begin{multline}
2\pi\i\e{Res}\Big( \wh{\Cx}_{1,\dots,n}^{(\ga)}\big( \bs{\be}_n \, ; \, \bs{u}_m\big) \cdot  \dd u_m \, , \, u_m= \be_{\ell_m} - \i\pi + \i\xi  \Big)  \\ \, = \,
\f{ 4\pi\i\ex{-\ga\be_{\ell_m}} }{ m \varkappa } \cdot
\pl{\substack{k=1 \\ k \neq \ell_m}}{n} (a\phi)\big(\be_{k\ell_m}) \cdot \pl{t=1}{m-1} \tau\big(u_t-\be_{\ell_m})
	\,  \wh{\Cx}_{1,\dots,n}^{(\ga)}\big( \bs{\be}_n \, ; \, \bs{u}_{m-1}\big) \cdot \op{E}^{12}_{\ell_m} \cdot \Big( 1 + \e{O}\big(\xi-\pi\big) \Big)\,.
\label{contribution type 1}
\end{multline}
Now, in what pertains to the second type of residues of interet:
\begin{equation}
\op{C}_{1,\dots,n}^{(\ga)}\big( \bs{\be}_n \, ; \,\be_{\ell}-\i\xi\big) = - \ex{\i\ga\pi} \, \op{E}^{12}_{\ell} \cdot \Big( 1 + \e{O}\big(\xi-\pi\big) \Big) \,,
\label{evaluation C ff}
\end{equation}
while
\begin{equation}
\e{Res}\Big(g\big( \bs{\be}_n \, ; \, \bs{u}_m\big) \cdot  \dd u_m \, , \, u_m= \be_{\ell} - \i\xi  \Big) \, = \,
-2\i g \big(\bs{\be}_n;\bs{u}_{m-1} \big) \cdot \pl{\substack{k=1 \\ k \neq \ell}}{n} (a\phi)\big(\be_{k\ell} + \i\pi)
\cdot \pl{t=1}{m-1} \tau\big(u_t-\be_{\ell}+\i\pi) \cdot \Big( 1 + \e{O}\big(\xi-\pi\big) \Big)\,,
\label{residu g ff}
\end{equation}
and just like before
\begin{equation}
p^{(\ga)} \big(\bs{\be}_n;\bs{u}_{m-1}, \be_{\ell} - \i\xi \big) = p^{(\ga)} \big(\bs{\be}_n;\bs{u}_{m-1} \big) \cdot \f{\i}{m\varkappa} \,.
\label{evaluation p ff 2}
\end{equation}
Thus, \eqref{evaluation C ff}, \eqref{residu g ff} and \eqref{evaluation p ff 2} all together yield:
\begin{multline}
2\pi\i\e{Res}\Big( \wh{\Cx}_{1,\dots,n}^{(\ga)}\big( \bs{\be}_n \, ; \, \bs{u}_m\big) \cdot  \dd u_m \, , \, u_m= \be_{\ell_m} -\i\xi  \Big)  \\
 \, = \,
 - \f{ 4\pi\i }{ m \varkappa }   \ex{\ga (\i\pi - \be_{\ell_m} ) }
\pl{\substack{k=1 \\ k \neq \ell_m}}{n} (a\phi)\big(\be_{k\ell_m}+\i\pi) \cdot \pl{t=1}{m-1} \tau\big(u_t-\be_{\ell_m}+\i\pi)
 \, \wh{\Cx}_{1,\dots,n}^{(\ga)}\big( \bs{\be}_n \, ; \, \bs{u}_{m-1}\big) \, \op{E}^{12}_{\ell_m}   \Big( 1 + \e{O}\big(\xi-\pi\big) \Big)\,.
	\label{contribution type 2}
\end{multline}
Therefore, upon recalling the notations \eqref{definition beta 0 et beta 1}, the weighted contribution of the residues can be recast as
\begin{multline}
  \sul{\ell_m=1}{n}\sul{j_m=0}{1 } 2\pi\i\e{Res}\Big( \wh{\Cx}_{1,\dots,n}^{(\ga)}\big( \bs{\be}_n \, ; \, \bs{u}_m\big) \cdot  \dd u_m \, , \, u_m= \be_{\ell_m}^{(j_m)} \Big)
\, = \, \f{4\pi\i}{m\varkappa}  \cdot \wh{\Cx}_{1,\dots,n}^{(\ga)}\big( \bs{\be}_n \, ; \, \bs{u}_{m-1}\big) \\
	\times \sul{\ell_m=1}{n} \sul{j_m=0}{1}
\ex{-\ga(\be_{\ell_m}-j_m\i\pi)}  \pl{\substack{k=1 \\ k \neq \ell_m}}{n} (a\phi)\big(\be_{k\ell_m} + j_m\i\pi) \cdot \pl{t=1}{m-1} \tau\big(u_t-\be_{\ell_m}+j_m\i\pi)
\, \op{E}^{12}_{\ell_m}   \Big( 1 + \e{O}\big(\xi-\pi\big) \Big)\,.
\end{multline}
The only dependence of the terms appearing on the \textit{rhs} under the sum on $\bs{u}_{m-1}$ is through $\tau$, which is an entire function. This
allows one to readily iterate the residue calculation and, upon using the identity
\begin{equation}
	\wh{\Cx}_{1,\dots,n}^{(\ga)}\big( \bs{\be}_n \, ; \emptyset \big)   \, =  \,   \pl{b<a}{n} \op{F}(\be_{ba}) \cdot \pl{b=1}{n} \ex{\f{\ga}{2}\be_b} \, \Om_n  \;,
\end{equation}
arrive to the representation
\begin{multline}
\label{eq intermediaire free FF}
(2\pi\i)^m  \hspace{-3mm} \sul{  \bs{\ell}_m \in \intn{1}{n}^m }{} \hspace{-2mm} \mc{R}_m\big( \bs{\ell}_m ;\bs{\be}_n; \emptyset \big)\, = \,
\f{ 1 }{ m! }  \Big( \f{4\pi\i}{\varkappa} \Big)^m \pl{b>a}{n} \op{F}(\be_{ba})\cdot \pl{k=1}{n} \ex{ \f{\ga}{2} \be_k }
 \sul{  \bs{\ell}_m \in \intn{1}{n}^m }{}  \sul{ \bs{j}_m \in \{0,1\}^m }{ }  \\
\times \pl{t=1}{m} \pl{ \substack{t=1 \\ k \neq \ell_{t}} }{ n } (a\phi)\big(\be_{k\ell_{t}} + j_{t}\i\pi)
\pl{t=1}{m}  \ex{-\ga(\be_{\ell_{t}}-j_t \i \pi) }  \pl{t>k}{m} \tau\big(\be_{\ell_{t}\ell_k}-\i \pi j_{tk} \big)
\cdot \Om_n \pl{t=1}{m}  \op{E}^{12}_{\ell_t}  \cdot \Big( 1 + \e{O}\big(\xi-\pi\big) \Big) \,.
\end{multline}
First of all, we can already observe that the summand in the \textit{rhs} of \eqref{eq intermediaire free FF}
is a symmetric function of $\bs{\ell}_m$ that vanishes on the diagonal $\ell_j = \ell_{p}$ with $p \neq j$.
The vanishing issues as much from the product of the nilpotent matrices as from the zeroes of the function $\tau$.
This allows one to restrict the summation to the ordered domain $\intn{1}{n}^m_{<} \, = \,  \big\{ \bs{\ell}_m \, : \, 1\leq  \ell_1 < \dots < \ell_m \leq n  \big\}$.
To simplify further, one notices that for generic $\th$:
\begin{equation}
\op{F}(\th) = -\i\sinh\left(\tfrac{\th}{2}\right) \Big( 1 + \e{O}\big(\xi-\pi\big) \Big) \;, \qquad
(a\phi)(\th) = \f{-2\i}{\sinh(\th)} \Big( 1 + \e{O}\big(\xi-\pi\big) \Big)  \;, \qquad \tau(\th) = \f{\sinh^2(\th)}{4}  \Big( 1 + \e{O}(\xi-\pi) \Big)  \;.
\end{equation}
This leads to
\begin{multline}
(2\pi\i)^m  \hspace{-3mm} \sul{  \bs{\ell}_m \in \intn{1}{n}^m }{} \hspace{-2mm} \mc{R}_m\big( \bs{\ell}_m ;\bs{\be}_n; \emptyset \big) \, = \,
\left(\f{4\pi\i}{\varkappa}\right)^m \pl{b>a}{n} \left\{-\i\sinh\left(\tfrac{\be_{ba}}{2}\right) \right\}\cdot \pl{b=1}{n} \ex{\f{\ga}{2}\be_b}
	\cdot \sul{  \bs{\ell}_m \in \intn{1}{n}^m_{<} }{}  \sul{ \bs{j}_m \in \{0,1\}^m }{ }  \\
\pl{t=1}{m}  \bigg\{ \ex{-\ga(\be_{\ell_{t}}-j_t\i\pi)} \bigg\}
\pl{t=1}{m}  \pl{\substack{k=1 \\ k \neq \ell_{t}}}{n} \bigg\{  \f{-2\i (-1)^{j_{t}} }{\sinh\big(\be_{k\ell_{t}}\big)}   \bigg\}
	\pl{ k < t }{ m }  \bigg\{ \f{\sinh^2\big(\be_{\ell_{k}\ell_t}\big)}{4} \bigg\}
\cdot \Om_n \pl{t=1}{m}  \op{E}^{12}_{\ell_t}  \cdot \Big( 1 + \e{O}\big(\xi-\pi\big) \Big)  \,.
\end{multline}
Taking into account that $n$ is even, one obtains:
\begin{equation}
\sul{ \bs{j}_m \in \{0,1\}^m }{ }
\pl{t=1}{m} \Big\{ \ex{\ga j_t\i\pi} (-1)^{(n-1)j_t} \Big\}  \, = \,
\Big\{-2\i\ex{\f{\i\ga\pi}{2}} \sin\big(\tfrac{\ga\pi}{2}\big)  \Big\}^m \;.
\end{equation}
Furthermore,  by using $n=2m$ one gets:
\begin{equation}
(-\i)^m (-\i)^{\f{n(n-1)}{2}} (-\i)^{m(n-1)} = \i^m \;.
\end{equation}
Finally, it holds
\begin{equation}
\varkappa \; = \;  \f{  2\i\xi  \ex{\i\ga \f{\pi}{2} } }{  \Big\{ \op{F}(\i\tfrac{\pi}{2}\big) \cdot \varpi\big(\i \tfrac{\xi-\pi}{2} \big)  \Big\}^4   } \, = \,
8\i\pi\ex{\i\ga \f{\pi}{2} }  \Big( 1 + \e{O}(\xi-\pi) \Big) \;.
\end{equation}
All-in-all, this leads to
\begin{multline}
(2\pi\i)^m  \hspace{-3mm}  \sul{  \bs{\ell}_m \in \intn{1}{n}^m }{} \hspace{-2mm} \mc{R}_m\big( \bs{\ell}_m ;\bs{\be}_n; \emptyset \big) \, = \,
\left\{\i\sin\left(\tfrac{\ga\pi}{2}\right)\right\}^m \pl{b>a}{n} \Big\{ \sinh\Big(\tfrac{\be_{ba}}{2}\Big) \Big\}\cdot \pl{b=1}{n} \ex{\f{\ga}{2}\be_b}   \\
\times \sul{  \bs{\ell}_m \in \intn{1}{n}^m_{<} }{}
\pl{t=1}{m} \Bigg\{ \ex{-\ga\be_{\ell_{t}}} \cdot
\pl{\substack{c=1 \\ c \neq \ell_{t} } }{n}          \f{2}{\sinh\big(\be_{c\ell_{t}}\big)} \cdot
\pl{k=1}{t-1} \f{\sinh^2\big(\be_{\ell_{k}\ell_t}\big)}{4} \Bigg\}
	\cdot  \Om_n \pl{t=1}{m}  \op{E}^{12}_{\ell_t}  \cdot \Big( 1 + \e{O}\big(\xi-\pi\big) \Big)  \,.
\end{multline}
It remains to simplify the products of the $\sinh$ factors. For fixed $\bs{\ell}_m \in \intn{1}{n}^m_{<}$, let $\wt{\bs{\ell}}_m \in \intn{1}{n}^m_{<}$ be such that
$ \{\ell_a \}_1^m \sqcup   \{ \wt{\ell}_a \}_1^m \, = \, \intn{1}{n}$. Then, it holds
\begin{equation}
\pl{t=1}{m} \Bigg\{ \ex{-\ga\be_{\ell_{t}}} \cdot \pl{\substack{c=1 \\ c \neq \ell_{t} } }{n}          \f{2}{\sinh\big(\be_{c\ell_{t}}\big)} \cdot
\pl{k=1}{t-1} \f{\sinh^2\big(\be_{\ell_{k}\ell_t}\big)}{4} \Bigg\}
\, = \,  (-1)^{ \f{m(m-1)}{2} } \pl{t, s=1}{m} \bigg\{        \f{2}{\sinh\big(\be_{\wt{\ell}_s\ell_t}\big)}  \bigg\}
\end{equation}
and then
\begin{equation}
\pl{b>a}{n} \Big\{\sinh\big(\tfrac{\be_{ba}}{2}\big) \Big\} \pl{t, s=1}{m} \bigg\{        \f{2}{\sinh\big(\be_{\wt{\ell}_s\ell_t}\big)}  \bigg\}
\, = \, (-1)^{ \#\{t<s \, : \, \ell_t < \wt{\ell}_s \} }
 \f{  \pl{t>s}{m} \Big\{  \sinh\big(\tfrac{1}{2}\be_{\ell_t\ell_s}\big)  \sinh\big(\tfrac{1}{2} \be_{ \wt{\ell}_t\wt{\ell}_s} \big) \Big\} }
{ \pl{t,s}{m} \Big\{  \cosh\big(\tfrac{1}{2} \be _{\wt{\ell}_t\ell_s} \big)  \Big\}  }  \;.
\end{equation}
At this stage, it remains to observe that given $\bs{\ell}_m \in \intn{1}{n}^m_{<}$ and thus the associated $\wt{\bs{\ell}}_m \in \intn{1}{n}^m_{<}$
introduced above, one may define $\bs{\veps}_n\in \{0,1\}^n$ as $\veps_{\ell_t}=-$  and  $\veps_{\wt{\ell}_t}=+$, this for $t=1,\dots,m$. Reciprocally,
given any $\bs{\veps}_n\in \mc{E}_n$ as defined in \eqref{definition En}, by the above, one unambigously builds a sequence   $\bs{\ell}_m \in \intn{1}{n}^m_{<}$.
Upon changing to these new coordinates and taking $\xi \tend \pi$, the claim follows. \qed

\subsection{The chargeless two soliton/antisoliton sector}

In this subsection, we focus on the particular $(n,m)=(2,1)$ sector in the absence of breathers $\xi > \pi$ and
recast the proposed expression for the form factor of the exponential of the field operator
in a form which allows for a direct comparison with the results of \cite{LukyanovConjectureFFExponentialFieldSineGSolASolAndBReather}.
That work, building on the free field approach to the resolution of the bootstrap axioms,
proposed an approach allowing one to compute the form factors of the field exponentials
through certain vacuum expectation values of products of free fields. This was explicitated there in the two soliton/antisoliton
sector. That conjectural approach was tested in various limits: free fermions, UV CFT limit, etc.

\begin{prop}
Let $(n,m)=(2,1)$, $\xi>\pi$ and $\ga \in \mathbb{C}$ be  such that $\abs{\Re(\ga)} \leq 1$. Then, the two-point form factors of the operator
$\mc{O}^{(\ga)} = \ex{\i\f{g\ga}{2} \Psi}$ as defined by Conjecture \ref{Conjecture FF exponentielle du champ}
are such that the non-zero components admit the following representation, c.f. notations of \eqref{form factor coordinates}:
\begin{equation}
\label{-+ integral FF}
\Big[ \bs{\Psi}_{1,2}\big[ p^{(\ga)}\big]\big( \bs{\be}_2 \big) \Big]^{-+} \,  =  \,
\op{F}(\be_{12})\, \ex{-\f{\pi}{2\xi} (\be_{21} + \i\pi )}
\, \Big\{ \ex{-\i\ga\pi} \mc{I}_{\ga} ( \be_{21}) \,  +  \,  \mc{I}_{\ga} ( - 2\pi \i - \be_{21} ) \Big\} \;,
\end{equation}
\begin{equation}
	\label{+- integral FF}
\Big[\bs{\Psi}_{1,2}\big[ p^{(\ga)}\big]\big( \bs{\be}_2 \big) \Big]^{+-} \,  =  \, \op{F}(\be_{12}) \, \ex{\f{\pi}{2\xi} (\be_{21} + \i\pi )} \,
\Big\{ \ex{\i\ga\pi} \mc{I}_{\ga} ( \be_{21} )   \, +  \, \mc{I}_{\ga} ( - 2\pi \i - \be_{21}) \Big\} \;.
	\end{equation}
Above, given generic $\th \in \mathbb{C}$, we agree upon :
	\begin{equation}
	\label{def integrale lukyanov}
\mc{I}_{\ga} (\th) = \f{\i\ex{\i\f{\ga\pi}{2}}}{\varkappa}  \hspace{-1mm} \Int{ \msc{C}_{\bs{w}_{\th} } }{} \hspace{-1mm}  \dd x
\, \big(a\phi\big)\Big(-\tfrac{\th}{2} - x\Big) \, \big(a\phi\big)\Big(\tfrac{\th}{2} - x\Big)  \ex{-  (\f{\pi}{\xi}+\ga ) (x - \f{\i\pi}{2} ) } \;, \quad
\quad \bs{w}_{\th} \, = \,   \big(- \f{\th}{2}, \f{\th}{2}  \big)
\end{equation}
in which the integration contour is as defined in Subsubsection \ref{subsubsection solution bootstrap}. Note that, when $\th \in \R$, the simpler contour integral representation holds:
\begin{equation}
\label{rep Lukyanov integral R}
\mc{I}_{\ga} (\th) \, = \, \f{ \i\ex{\i\f{\ga\pi}{2}} }{ \varkappa } \Int{\R}{} \dd x  \,\big(a\phi\big)\Big(x + \tfrac{\th}{2} -\i\tfrac{\pi}{2} \Big)
\big(a\phi\big)\Big(x - \tfrac{\th}{2} -\i\tfrac{\pi}{2} \Big)  \ex{ - (\f{\pi}{\xi}+\ga ) x}  \;.
\end{equation}
\end{prop}

\proof When $\th \in \R$,  we first establish the equivalence of \eqref{rep Lukyanov integral R} and \eqref{def integrale lukyanov}.
Thus, starting from \eqref{rep Lukyanov integral R}, one deforms the contour $\msc{C}_{\bs{w}_{\th} }$ to $\R + \i\tfrac{\pi}{2}$
what is licit since no poles of the integrand are crossed in the process. This results in
\begin{equation}
\mc{I}_{\ga} (\th) \, = \,  \f{ \i\ex{\i\f{\ga\pi}{2}} }{ \varkappa } \Int{\R}{} \dd x  \, \big(a\phi\big)\Big(\tfrac{\th}{2} - x -\i\tfrac{\pi}{2} \Big)
\, \big(a\phi\big)\Big(- \tfrac{\th}{2} - x -\i\tfrac{\pi}{2} \Big) \,  \ex{- (\f{\pi}{\xi}+\ga )x} \;.
\end{equation}
One then concludes by invoking the property $\big(a\phi\big)(\la) = \big(a\phi)(-\i\pi - \la)$.\\
We now establish the two identities \eqref{-+ integral FF}-\eqref{+- integral FF} along with the vanishing of the $++$ and $--$ components.
Going back to formula \eqref{form factor exponential} at $n=2$:
\beq
\bs{\Psi}_{1,2}\big[ p^{(\ga)}\big]\big( \bs{\be}_2 \big) \, = \, \Int{ \msc{C}_{\bs{\be}_2} }{} \dd u \; 
\wt{\Cx}_{1,2}^{\, (\ga)}\big(\bs{\be}_2;u \big) \cdot \op{O}_{1,2}^{(\ga)}(\bs{\be}_2)  \cdot \big( g \cdot p^{(\ga)} \big)\big(\bs{\be}_2;u \big)  \;,
\enq
one gets by definition of the monodromy matrix \eqref{monodromy matrix}:
\begin{equation}
\wt{\Cx}_{1,2}^{\, (\ga)}\big(\bs{\be}_2;u \big) \cdot \op{O}_{1,2}^{(\ga)}(\bs{\be}_2) \, = \,
\wt{b}(\be_1-u) \cdot \wt{c}^{\, (\ga)}(\be_2-u) \cdot \ex{\f{\ga}{2}\be_{12}} \cdot (\op{v}^+_1 \op{v}^-_2)^{\op{t}} \,  + \,
\wt{c}^{\, (\ga)}(\be_1-u)  \cdot \ex{\f{\ga}{2}\be_{21}} \cdot (\op{v}^-_1 \op{v}^+_2)^{\op{t}} \;,
\end{equation}
while
\begin{equation}
\big( g \cdot p^{(\ga)} \big)\big(\bs{\be}_2;u \big) = \f{\i}{\varkappa}   \cdot \op{F}(\be_{12}) \cdot \big( a \phi \big) (\be_1 - u) \cdot \big( a \phi \big) (\be_2 - u) \;.
\end{equation}
The above already yields that $ \Big[\bs{\Psi}_{1,2}\big[ p^{(\ga)}\big]\big( \bs{\be}_2 \big) \Big]^{\pm\pm}  = 0$, while
\begin{equation}
\Big[ \bs{\Psi}_{1,2}\big[ p^{(\ga)}\big]\big( \bs{\be}_2 \big) \Big]^{-+} = \f{\i}{\varkappa} \ex{\f{\ga}{2} \ov{\bs{\be}}_2} \cdot \op{F}(\be_{12}) \cdot \Int{ \msc{C}_{\bs{\be}_2} }{} \dd u \;  \wt{c}(\be_1-u) \cdot \ex{-\ga u} \cdot \big( a \phi \big) (\be_1 - u) \cdot \big( a \phi \big) (\be_2 - u) \;,
\label{diff contour MP}
\end{equation}
and
\begin{equation}
\Big[ \bs{\Psi}_{1,2}\big[ p^{(\ga)}\big]\big( \bs{\be}_2 \big) \Big]^{+-}  = \f{\i}{\varkappa} \ex{\f{\ga}{2}\ov{\bs{\be}}_2 } \cdot \op{F}(\be_{12}) \cdot \Int{ \msc{C}_{\bs{\be}_2} }{} \dd u
\; \wt{b}(\be_1-u) \cdot \wt{c}(\be_2-u) \cdot \ex{-\ga u} \cdot \big( a \phi \big) (\be_1 - u) \cdot \big( a \phi \big) (\be_2 - u) \;.
\label{diff contour PM}
\end{equation}
At this stage, one implements the change of variables $u \mapsto v + \f{1}{2}\ov{\bs{\be}}_2$:
\begin{equation}
\Big[ \bs{\Psi}_{1,2}\big[ p^{(\ga)}\big]\big( \bs{\be}_2 \big) \Big]^{-+} \, = \,  \f{\i}{\varkappa} \cdot \op{F}(\be_{12}) \cdot
\Int{\msc{C}_{ \bs{w}_{ \be_{21} } } }{} \dd v \;
\wt{c}\left(\tfrac{\be_{12}}{2} - v\right) \cdot \ex{-\ga v} \cdot \big(a\phi\big)\left(\tfrac{\be_{12}}{2} - v\right)\big(a\phi\big)\left(\tfrac{\be_{21}}{2} - v\right) \;,
	\label{-+ coordinate}
\end{equation}
\begin{equation}
\Big[ \bs{\Psi}_{1,2}\big[ p^{(\ga)}\big]\big( \bs{\be}_2 \big) \Big]^{+-} \, = \, \f{\i}{\varkappa} \cdot \op{F}(\be_{12}) \cdot
\Int{  \msc{C}_{ \bs{w}_{ \be_{21} } } }{} \dd v \; \wt{b}\left(\tfrac{\be_{12}}{2} - v\right) \cdot \wt{c}\left(\tfrac{\be_{21}}{2} - v\right) \cdot \ex{-\ga v}
\cdot \big(a\phi\big)\left(\tfrac{\be_{12}}{2} - v\right)\big(a\phi\big)\left(\tfrac{\be_{21}}{2} - v\right) \;.
\label{+- coordinate}
\end{equation}
We first simplify \eqref{-+ coordinate}, making use of the identities
\begin{equation}
\label{useful lemma}
\wt{c}(\la) \, = \, \ex{\f{\pi}{\xi}\la} \, \Big\{ 1 + \ex{-\i\f{\pi^2}{\xi}} \cdot \wt{b}(\la) \Big\}  \qquad \e{and} \qquad
	(a\phi) (\la - 2\pi\i) = (a\phi)(\la) \cdot \wt{b}(\la)\;,
\end{equation}
in order to obtain
\begin{multline}
\Big[ \bs{\Psi}_{1,2}\big[ p^{(\ga)}\big]\big( \bs{\be}_2 \big) \Big]^{-+}  = \f{\i}{\varkappa} \ex{-\i\f{\ga\pi}{2}}\ex{-\f{\pi}{2\xi}(\be_{21}+\i\pi)} \cdot \op{F}(\be_{12}) \\
\times \Int{  \msc{C}_{ \bs{w}_{ \be_{21} } } }{} \dd v \; \ex{-(\ga+\f{\pi}{\xi}) (v - \i\f{\pi}{2})} \cdot
\bigg\{  \big(a\phi\big)\Big(\tfrac{\be_{12}}{2} - v\Big) \,+ \,  \ex{-\i\f{\pi^2}{\xi}} \cdot\big(a\phi\big)\Big(\tfrac{\be_{12}}{2} - v - 2\pi\i\Big) \bigg\}
\cdot \big(a\phi\big)\Big(\tfrac{\be_{21}}{2} - v\Big) \;.
\end{multline}
The contribution associated with the first integrand of the second line exactly corresponds to $\ex{-\i\ga\pi} \mc{I}_{\ga} (\be_{21})$
as defined through \eqref{def integrale lukyanov}. As for the contribution of the second integrand, we implement the change of variables  $v=t-\i\pi$, what yields
\begin{multline}
\f{\i}{\varkappa} \ex{-\i\f{\ga\pi}{2}} \hspace{-2mm} \Int{  \msc{C}_{ \bs{w}_{ \be_{21} } } }{} \hspace{-2mm}  \dd v \;
\ex{-(\ga+\f{\pi}{\xi}) (v - \i\f{\pi}{2})} \cdot \ex{-\i\f{\pi^2}{\xi}}
\cdot\big(a\phi\big)\Big(\tfrac{\be_{12}}{2} - v - 2\pi\i\Big)\cdot \big(a\phi\big)\Big(\tfrac{\be_{21}}{2} - v\Big)  \\
\; = \;
\f{\i}{\varkappa} \ex{\i\f{\ga\pi}{2}} \hspace{-3mm} \Int{  \msc{C}_{ \bs{w}_{ \be_{21} } } + \i\pi }{} \hspace{-3mm} \dd v \;
\ex{-(\ga+\f{\pi}{\xi}) (t - \i\f{\pi}{2})} \cdot\big(a\phi\big)\left(\tfrac{\be_{12}}{2} - t - \i\pi\right)\cdot \big(a\phi\big)\left(\tfrac{\be_{21}}{2} - t + \i \pi\right) \;.
\end{multline}
One then concludes by noticing that, taken the pole structure of the integrand, the contour of integration $\msc{C}_{ \bs{w}_{ \be_{21} } } + \i\pi$  can be deformed to
$\msc{C}_{ \bs{w}_{ \be_{21} -2\i\pi} }$, thus yielding \eqref{-+ integral FF}.\\
We now move on to the computation of the other nonzero coordinate \eqref{+- coordinate}. Using the same identities as before, we end up with
\begin{multline}
\Big[ \bs{\Psi}_{1,2}\big[ p^{(\ga)}\big]\big( \bs{\be}_2 \big) \Big]^{+-}  = \f{\i}{\varkappa} \ex{\i\f{\ga\pi}{2}}\ex{\f{\pi}{2\xi}(\be_{21}+\i\pi)} \cdot \op{F}(\be_{12}) \\
	\times \Int{ \msc{C}_{ \bs{w}_{ \be_{21} } }  }{} \dd v \; \ex{-(\ga+\f{\pi}{\xi}) (v + \i\f{\pi}{2})} \cdot
\bigg\{  \big(a\phi\big)\big(\tfrac{\be_{21}}{2} - v\big) + \ex{-\i\f{\pi^2}{\xi}} \cdot\big(a\phi\big)\big(\tfrac{\be_{21}}{2} - v - 2\pi\i\big) \bigg\}\cdot
\big(a\phi\big)\big(\tfrac{\be_{12}}{2} - v - 2\pi\i\big) \;.
\end{multline}
The first contribution was already dealt with in the previous computation, thus leading to:
\begin{multline}
\Big[ \bs{\Psi}_{1,2}\big[ p^{(\ga)}\big]\big( \bs{\be}_2 \big) \Big]^{+-} =   \ex{\f{\pi}{2\xi}(\be_{21}+\i\pi)} \cdot \op{F}(\be_{12})  \\
\times \Biggr[ \mc{I}_{\ga} (-2\pi\i - \be_{21}) + \f{\i}{\varkappa} \ex{\i\f{\ga\pi}{2}} \Int{ \msc{C}_{ \bs{w}_{ \be_{21} } }  }{} \dd v \;
\ex{-(\ga+\f{\pi}{\xi}) (v + \i\f{\pi}{2})} \cdot \ex{-\i\f{\pi^2}{\xi}}
\cdot\big(a\phi\big)\big(\tfrac{\be_{12}}{2} - v - 2\pi\i\big)\cdot \big(a\phi\big)\big(\tfrac{\be_{21}}{2} - v - 2\pi\i\big) \Biggr]\;.
\end{multline}
The pole structure of the remaining integrand allows to deform the contour from $\msc{C}_{ \bs{w}_{ \be_{21} } }$ to
$\msc{C}_{ \bs{w}_{ \be_{21} } } +2\pi\i$ without crossing any poles. This exactly leads to the sought result for the form factors. \qed

For generic values of $\ga$, expressions \eqref{-+ integral FF} and \eqref{+- integral FF} do not appear to simplify further.
However, as shown in  \cite{LukyanovConjectureFFExponentialFieldSineGSolASolAndBReather}, simpler expressions can be obtained for $\ga \in \big\{ \pm  1 \, , 0  \big\}$
as well as for $\Dp{\ga}\, _{\mid\ga=0}$ thus reproducing the previously established expressions \cite{KarowskiWeiszFormFactorsFromSymetryAndSMatrices,SmirnovFormFactors} for the two-soliton
antisoliton form factors of $ \ex{ \pm \i\f{g}{2} \Psi} $,  the identity operator as well as of $     \i\tfrac{g}{2} \Psi $.
More precisely, one gets
\begin{equation}
\bs{\Psi}_{1,2}\big[ p^{(1)}\big]\big( \bs{\be}_2 \big) = \f{-2\pi\i\cdot \op{F}(\be_{12})}{\xi \cdot \sinh\left(\tfrac{\pi}{\xi}(\i\pi-\be_{12})\right)}
\cdot \bigg\{   \ex{  \pm \f{ \pi(\be_{12}-\i\pi) }{2\xi} } \cdot \big( \op{v}^+_1 \op{v}^-_2\big)^{\op{t} } \, + \,
\ex{ \pm \f{ \pi (\i\pi-\be_{12}) }{ 2\xi } } \cdot  \big(\op{v}^-_1 \op{v}^+_2\big)^{\op{t}} \bigg\}    \;,
\end{equation}
as well as
\begin{equation}
\bs{\Psi}_{1,2}\big[ p^{(\ga)}\big]\big( \bs{\be}_2 \big) \underset{\ga \rightarrow 0}{=}\i \f{g\ga}{2} f_{1,2}(\be_{12}) + \e{o}\big(\ga\big)
\,\quad \e{with} \quad
 f_{1,2}(\th) \, = \, - \f{ \pi \op{F}(\th) \Big\{ \big( \op{v}^+_1\op{v}^-_2 \big)^{\op{t}} \, - \, \big( \op{v}^-_1\op{v}^+_2 \big)^{\op{t}} \Big\} }
 {g\cosh\big[\tfrac{\pi}{2\xi}(\i\pi - \th)\big] \cosh\big(\tfrac{\th}{2}\big) }  \;.
\end{equation}
We refer to \cite{LukyanovConjectureFFExponentialFieldSineGSolASolAndBReather} for the details of the derivation.

\subsection{The charged case}
In this subsection, we specialise the general expression \eqref{form factor exponential} from Conjecture \ref{Conjecture FF exponentielle du champ}
so as to reproduce the expressions obtained in \cite{LukyanovZamolodchikovSolitonCreatingOpsSineG}
for some form factors of soliton creating operators. These coincide, albeit with different normalization factors.
\begin{prop}
Let $n \geq 2$ and $q=n$. Then, the $n$-solitons form factor is given by:
\begin{equation}
\bs{\Psi}_{1,\dots,n}\big[ p^{(\ga)}\big]\big( \bs{\be}_n \big) = \pl{i<j}{n} \op{F}\big(\be_i - \be_j\big) \cdot \pl{a=1}{n}\ex{\f{\ga}{2}\be_a}
\cdot \big( \op{v}^+_1\cdots \op{v}^+_n\big)^{\op{t}} \;.
\end{equation}
\end{prop}
\proof
This is a direct consequence of the formula \eqref{general soliton form factor} applied to the case $m=0$. In that case, there are no integrations,
and $p^{(\ga)}\big(\bs{\be}_n ; \emptyset\big) = 1$. \qed

\vspace{2mm}
We now focus on the case of when one soliton-antisoliton pair is added, while conserving the total charge.
\begin{prop}
Let $n \geq 2$ and $q = n$. Then, the $n+1$-solitons-$1$-antisoliton form factor is given by:
\begin{multline}
\bs{\Psi}_{1,\dots,n+2}\big[ p^{(\ga)}\big]\big( \bs{\be}_{n+2} \big) = \f{\Big\{ \op{F}(\i\tfrac{\pi}{2}\big) \cdot \varpi\big(\i \tfrac{\xi-\pi}{2} \big)  \Big\}^4}{2\xi}
\pl{s=1}{n+2}\ex{\f{\ga}{2}\be_s} \pl{k<\ell}{n+2} \op{F}\big(\be_{k\ell} \big) \sul{k=1}{n+2} \op{V}_k^{\op{t}} \, \ex{\f{\pi}{\xi}\be_k}
\hspace{-4mm} \Int{ \msc{C}_{\bs{\be}_{n+2}} + \i\f{\pi}{2} }{} \hspace{-3mm} \dd v \; \ex{- (\ga + \f{\pi}{\xi} ) v } \\
\times
\sul{ \veps=\pm }{}  \Big\{ \ex{ \i\f{\pi^2}{2\xi} \veps } \big(a\phi)\big( \veps( \be_k - v) + \i\tfrac{\pi}{2} \big)  \Big\}
\pl{a=1}{k-1} \big(a\phi)\big(\be_a - v + \i\tfrac{\pi}{2} \big) \pl{a=k+1}{n+2} \big(a\phi)\big(\i\tfrac{\pi}{2} + v - \be_a \big) \;.
\end{multline}
where $\op{V}_k \, = \, \op{v}^+_1 \cdots \op{v}^+_{k-1} \op{v}^-_k\op{v}^+_{k+1} \cdots \op{v}^+_{n+2}$.

\end{prop}

\proof 
Starting from \eqref{general soliton form factor} specialised to the present setting, it is clear that the only
non-zero components of the covector will be along the covectors $\op{V}_k^{\op{t}}$ for some $k \in \intn{1}{n+2}$. Moreover, one has
\begin{equation}
\bs{\Psi}_{1,\dots,n+2}\big[ p^{(\ga)}\big]\big( \bs{\be}_{n+2} \big) \cdot \op{V}_k \, = \,
\f{\i}{\varkappa}\ex{-\ga\be_k} \pl{a=1}{n+2}\ex{\f{\ga}{2}\be_a} \pl{ k<\ell}{n+2} \op{F}\big(\be_{k\ell}\big) \Int{ \msc{C}_{\bs{\be}_{n+2}} }{} \hspace{-3mm} \dd u \;
\Om_{n+2} \wt{\op{C}}_{1,\dots,n+2}^{\, (\ga)}\big(\bs{\be}_{n+2};u \big) \cdot\op{V}_k \; \pl{a=1}{n+2} \big(a\phi)\big(\be_a - u \big) \;.
\end{equation}
The scalar product may be recast in terms of $\op{S}$-matrices and then one directly computes their action on $\op{v}^+_0 \op{V}_k$ until reaching $\wt{\op{S}}_{k 0}^{\, (\ga)}$:
\begin{equation}
\Om_{n+2} \wt{\op{C}}_{1,\dots,n+2}^{\, (\ga)}\big(\bs{\be}_{n+2};u \big) \cdot\op{V}_k \, = \,
\Om_{n+2} \cdot \big(\op{v}^-_0\big)^{\op{t}} \wt{\op{S}}_{1 0}^{\, (\ga)}(\be_1-u)\cdots \wt{\op{S}}_{k-1 0}^{\, (\ga)}(\be_{k-1}-u)
\cdot \wt{\op{S}}_{k 0}^{\, (\ga)}(\be_k-u) \cdot \op{v}^+_0 \op{V}_k \;.
\end{equation}
One then observes that
\beq
\big(\op{v}^+_k\big)^{\op{t}} \wt{\op{S}}_{k 0}^{\, (\ga)}(\th)  \op{v}^+_0 \op{v}^-_k \, =  \,  \wt{c}^{(\ga)}\big(\th\big)   \op{v}^-_0\qquad \e{and} \qquad
 \big(\op{v}^+_k\big)^{\op{t}} \wt{\op{S}}_{k 0}^{\, (\ga)}(\th)  \op{v}^-_0 \op{v}^+_k \, = \,  \wt{b}\big(\th\big) \op{v}^-_0 \;,
\enq
thus leading to
\begin{equation}
\Om_{n+2} \wt{\op{C}}_{1,\dots,n+2}^{\, (\ga)}\big(\bs{\be}_{n+2};u \big) \cdot\op{V}_k \, = \,
\ex{-\ga u} \cdot \ex{\ga \be_k} \cdot\wt{c}\big(\be_k - u\big) \cdot \pl{a=1}{k-1} \wt{b}\big(\be_a - u\big)\;.
\end{equation}
By using the second identity in  \eqref{useful lemma}, one then recasts the form factor as
\begin{multline}
\bs{\Psi}_{1,\dots,n+2}\big[ p^{(\ga)}\big]\big( \bs{\be}_{n+2} \big) \cdot\op{V}_k \, = \, \f{\i}{\varkappa} \pl{a=1}{n+2}\ex{\f{\ga}{2}\be_a}
\pl{ k<\ell}{n+2} \op{F}\big(\be_{k\ell}\big) \\
\times \Int{ \msc{C}_{\bs{\be}_{n+2}} }{} \hspace{-2mm} \dd u \; \ex{-\ga u} \cdot \wt{c}\big(\be_k - u\big) \cdot \big(a\phi)\big(\be_k - u \big)
\, \pl{a=1}{k-1} \big(a\phi)\big(\i\pi + u - \be_a  \big) \, \pl{a=k+1}{n+2} \big(a\phi)\big(\be_a - u \big) \;.
\end{multline}
Upon implementing the first identity in \eqref{useful lemma}, one gets
\begin{multline}
\bs{\Psi}_{1,\dots,n+2}\big[ p^{(\ga)}\big]\big( \bs{\be}_{n+2} \big) \cdot\op{V}_k  \, = \,
\f{\i}{\varkappa} \ex{\f{\pi}{\xi}\be_k}\pl{a=1}{n+2}\ex{\f{\ga}{2}\be_a} \pl{ k<\ell}{n+2} \op{F}\big(\be_{k\ell}\big) \\
\times \Int{ \msc{C}_{\bs{\be}_{n+2}} }{} \hspace{-2mm} \dd u \; \ex{- (\ga + \f{\pi}{\xi}) u} \,
\Big[  \big(a\phi)\big(\be_k - u \big) + \ex{-\i\f{\pi^2}{\xi}} \big(a\phi)\big(\i\pi + u - \be_k \big) \Big] \,
\pl{a=1}{k-1} \big(a\phi)\big(\i\pi + u - \be_a  \big)  \, \pl{a=k+1}{n+2} \big(a\phi)\big(\be_a - u \big) \;.
\end{multline}
Setting $v = u + \i \f{\pi}{2}$ finally entails
\begin{multline}
\bs{\Psi}_{1,\dots,n+2}\big[ p^{(\ga)}\big]\big( \bs{\be}_{n+2} \big)  \cdot\op{V}_k \, = \,
\f{\Big\{ \op{F}(\i\tfrac{\pi}{2}\big) \cdot \varpi\big(\i \tfrac{\xi-\pi}{2} \big)  \Big\}^4}{2\xi} \ex{\f{\pi}{\xi}\be_k}
\pl{a=1}{n+2}\ex{\f{\ga}{2}\be_a} \pl{i<j}{n+2} \op{F}\big(\be_i - \be_j\big) \\
\times \Int{ \msc{C}_{\bs{\be}_{n+2}} + \i\f{\pi}{2} }{} \hspace{-3mm} \dd v \; \ex{-\big(\ga + \f{\pi}{\xi}\big) v}
\cdot \sul{\veps=\pm}{} \Big\{ \ex{\i \veps \f{\pi^2}{2\xi}} \big(a\phi)\big( \veps(\be_k - v) + \i\tfrac{\pi}{2} \big)   \Big\}
\cdot \pl{a=1}{k-1} \big(a\phi)\big(u - \be_a + \i\tfrac{\pi}{2} \big)
\cdot \pl{a=k+1}{n+2} \big(a\phi)\big(\i\tfrac{\pi}{2} +  \be_a -u \big) \;,
\nonumber
\end{multline}
from which the claim follows upon using that $\big(a\phi)\big(\i\tfrac{\pi}{2} +  u \big)= \big(a\phi)\big(\i\tfrac{\pi}{2} -  u \big)$, see \eqref{definition fct a phi}. \qed

\section{Conclusion}

This work proposed a new expression for the form factors associated with the exponentials of the Sine-Gordon field in all regimes, within the framework of the $p$-function and using
Bethe Ansatz techniques. This computation was possible by introducing a twist factor in the bootstrap equations relative to the two different soliton and antisoliton components. The
results were shown to be consistent with similar results obtained by different techniques. Our results pave the way towards the computation of the multipoint correlation functions
as well as proving the convergence of two-point functions of the
Sine-Gordon model

which should generalize the techniques we used in a previous paper for the Sinh-Gordon model.


\section*{Acknowledgment}

The work of KKK and AS is supported by the ERC Project LDRAM : ERC-2019-ADG Project 884584. KKK acknowledges support from CNRS and from
the joint AND-DFG TSF24 project ANR-24-CE92-0033. The authors thank S. Lukyanov and F. Smirnov for discussions.




\appendix

\section{Special functions of interest}
\label{Appendice Special Functions}

\subsection{The quantum dilogarithm}
\label{Appendice SousSection Quantum dilog}

The quantum dilogarithm $\varpi$ of periods $\om_1,\om_2 \in \R^+$ is a meromorphic function on $\Cx$ which admits the integral representation
\beq
\varpi(\la) \; = \;  \exp\Bigg\{  \pm \f{ \i \pi }{ 2 \om_1 \om_2 } \cdot \Big( \la^2 \,+\, \f{ \om^2_1+\om^2_2 }{ 12 }   \Big)
\; - \;  \Int{ \R  \pm  \i 0^+}{}  \f{ \dd t }{ 4 t } \f{ \ex{-2\i \la t}  }{ \sinh\big(t\om_1 \big) \cdot \sinh\big(t\om_2 \big)  }  \Bigg\} \;,
\label{definition quantum dilog}
\enq
valid for $|\Im(\la)| \, < \, \tf{ \Om }{2}$, where $\Om=\om_1+\om_2$. 

This function is invariant under the transformation $\om_1\leftrightarrow \om_2$ and satisfies the first order finite difference equations
\beq
\varpi\big(\la+\i\om_2\big) \, = \, 2\i  \sinh\Big[ \f{\pi}{\om_1}\big(\la+\i\f{\tau}{2} \big) \Big] \cdot \varpi(\la) \quad \e{and} \quad
\varpi\big(\la+\i\om_1\big) \, = \, 2\i  \sinh\Big[ \f{\pi}{\om_2}\big(\la - \i\f{\tau}{2} \big) \Big] \cdot \varpi(\la)  \;,
\label{ecriture eqn diff finite dilog}
\enq
with $\tau \, = \, \om_2 - \om_1$.
From there one entails that
\bem
\varpi\Big(\la-\i \tfrac{ \Om }{ 2 } + \i \ell \om_1 + \i k \om_2 \Big) \, = \, (-1)^{k\ell} \big( -2\i \big)^{\ell + k}  \cdot
\pl{p=0}{k-1}   \sinh\Big[ \f{\pi}{\om_1}\big(\la + \i p \om_2 \big) \Big]   \\
\times \pl{p=0}{\ell-1}  \sinh\Big[ \f{\pi}{\om_2}\big(\la + \i p \om_1 \big) \Big]   \cdot
\varpi\Big(\la-\i \tfrac{ \Om }{ 2 } \Big)
\label{ecriture recurrence generale direct dilogarithme}
\end{multline}
and symmetrically,
\bem
\varpi\Big(\la-\i \tfrac{ \Om }{ 2 } - \i \ell \om_1 - \i k \om_2 \Big) \, = \, (-1)^{k\ell} \Big( \f{\i}{2}  \Big)^{\ell + k}  \cdot
\pl{p=1}{k} \Big\{ \sinh\Big[ \f{\pi}{\om_1}\big(\la - \i p \om_2 \big) \Big] \Big\}^{-1}  \\
\times \pl{p=1}{\ell } \Big\{ \sinh\Big[ \f{\pi}{\om_2}\big(\la - \i p \om_1 \big) \Big] \Big\}^{-1} \cdot
\varpi\Big(\la-\i \tfrac{ \Om }{ 2 } \Big)  \;.
\label{ecriture recurrence generale inverse dilogarithme}
\end{multline}
The quantum dilogarithm has only simple poles and zeroes. These are located at
\beq
\varpi(x)=0 \quad \e{iff} \quad x \in  \i\f{\Om }{2} + \i\mathbb{N} \om_1+ \i\mathbb{N} \om_2 \qquad \e{and} \qquad
\varpi^{-1}(x)=0 \quad \e{iff} \quad x \in  -\i\f{\Om }{2} - \i\mathbb{N} \om_1 - \i\mathbb{N} \om_2 \;.
\enq
$\varpi$ satisfies to the inversion identity $\varpi(\la)\varpi(-\la)=1$ and $\big(\varpi(\la^*)\big)^*=\varpi^{-1}(\la)$.
One can also establish that
\beq
\e{Res}\Big(\varpi\big( \la - \i \tfrac{\Om}{2}\big) \cdot  \dd \la, \la=0 \Big) \, = \, \f{\i}{2\pi} \sqrt{\om_1 \om_2}  \qquad \e{and} \qquad
\varpi( \tfrac{\i}{2} \tau) \; = \; \sqrt{ \f{\om_2}{\om_1 } } \;,
\enq
as well as
\beq
\varpi(0)=1 \; , \quad \varpi\big( \i \tfrac{\om_1}{2} \big) \, = \,  \varpi\big( \i \tfrac{\om_2}{2} \big) \, = \,  \sqrt{2} \;.
\enq

The above entails that, for $(k,\ell)\in \mathbb{N}^2$,
\beq
\e{Res}\Big(\varpi\big( \la - \i \tfrac{\Om}{2}\big) \cdot  \dd \la, \la=-\i \ell \om_1 -\i k \om_2  \Big) \, = \, \f{\i }{2  \pi} \sqrt{\om_1 \om_2}
(-1)^{k\ell} \Big( \f{ 1 }{2 \i } \Big)^{\ell + k}  \cdot \bigg\{
\pl{p=1}{k}   \sinh\Big[ \i p \pi \f{ \om_2 }{\om_1} \Big]   \cdot \pl{p=1}{\ell}  \sinh\Big[ \i p \pi \f{\om_1 }{\om_2} \Big] \bigg\}^{-1} \;.
\label{ecriture formule residu general du dilogarithme}
\enq

\subsection{The Barnes-G function}

The Barnes-G function is an entire function with the following integral representation for $\Re{z}>0$:
\begin{equation}
\ln{G(1+z)} = \f{z}{2} \ln{(2\pi)} + \Int{0}{\infty} \f{\dd t}{t} \bigg\{ \f{1 - \ex{-zt} }{ 4 \sinh^2(\tfrac{t}{2}) } + \f{z^2}{2}\ex{-t} - \f{z}{t} \bigg\}
\label{ecriture rep int Barnes}
\end{equation}
and continued to the whole complex plane by means of the functional equation
\begin{equation}
G(z+1) = \Ga(z) G(z) \;,
\end{equation}
where $\Ga$ is the regular Euler-$\Ga$ function. One also has the integral representation
\beq
\f{ G(1+z) }{ G(1-z) } \, = \, (2\pi)^z \exp\bigg\{ - \Int{ 0 }{ z }  \pi s \cot(\pi s) \dd s  \bigg\} \;.
\label{ecriture rep int ratio barnes}
\enq

\section{Proof of the form factor equations}
\label{Appendice Preuve des eqns de FF}

\subsection{The Watson symmetry equation i)}
\label{Appendice SS section eqn symmetrie i}

It follows from the Yang-Baxter equation \eqref{equation de Yang-Baxter pour S twiste} satisfied by $\op{S}^{(\ga)}$ that 
\beq
\wt{\op{T}}_{1,\dots,n;0}^{(\ga)}\big( \bs{\be}_n^{(i+1i)};\la) \op{S}_{i+1 i }^{(\ga)}(\be_{i i+1}) \; = \; 
\op{S}_{i+1 i }^{(\ga)}(\be_{i i+1})  \op{P}_{i i+1} \wt{\op{T}}_{1,\dots,n;0}^{(\ga)}\big( \bs{\be}_n;\la) \op{P}_{i i+1} \;. 
\enq
Moreover, since it holds that 
\beq
\op{O}_{1,\dots,n}^{(\ga)}(\bs{\be}_n) \op{S}_{ij}(\be_{ij}) \; = \;\op{S}_{ij}^{(\ga)}(\be_{ij})   \op{O}_{1,\dots,n}^{(\ga)}(\bs{\be}_n) 
\label{ecriture equation echange operateur spin twist et matrice S}
\enq
one infers that  
\bem
\Om_n \wt{\op{C}}_{1,\dots,n}^{\, (\ga)}\big( \bs{\be}_n^{(i+1i)} ; u_1 \big)  \cdots \wt{\op{C}}_{1,\dots,n}^{\, (\ga)}\big(\bs{\be}_n^{(i+1i)} ; u_m\big)
\cdot \op{O}_{1,\dots,n}^{ (\ga)}\big(\bs{\be}_n^{(i+1i)}\big) \cdot \op{S}_{i+1 i }(\be_{i i+1}) \\
\; = \; 
a(\be_{i i+1})\Om_n \wt{\op{C}}_{1,\dots,n}^{\, (\ga)}(\bs{\be}_n ; u_1)  \cdots \wt{\op{C}}_{1,\dots,n}^{\, (\ga)}(\bs{\be}_n ; u_m) \, \op{P}_{i i+1}  \cdot \op{O}_{1,\dots,n}^{ (\ga)}\big(\bs{\be}_n^{(i+1i)}\big)  \\
\; = \; 
a(\be_{i i+1})\Om_n \wt{\op{C}}_{1,\dots,n}^{\, (\ga)}(\bs{\be}_n ; u_1)  \cdots \wt{\op{C}}_{1,\dots,n}^{\, (\ga)}(\bs{\be}_n ; u_m)   \cdot \op{O}_{1,\dots,n}^{  (\ga)}\big(\bs{\be}_n\big)  \, \op{P}_{i i+1} \;. 
\end{multline}
Moreover one has that 
\beq
g\big( \bs{\be}_n^{(i+1i)};\bs{u}_m)  \; = \; \f{ \op{F}(\be_{i+1 i}) }{ \op{F}(\be_{i i+1}) } g\big( \bs{\be}_n;\bs{u}_m)  \, = \, 
a(\be_{i+1 i})g\big( \bs{\be}_n;\bs{u}_m) \;. 
\enq
Hence, all together, this leads to 
\beq
\bs{\Psi}_{1,\dots,n}\big[ p\big]\big( \bs{\be}_n \big)\, \op{P}_{i i+1} \;=\;  \bs{\Psi}_{1,\dots,n}\big[ p\big]\big( \bs{\be}_n^{(i+1 i)} \big)  \,  \op{S}_{i i+1}\big(\be_{i i+1} \big) \;. 
\enq

\subsection{The Watson monodromy equation ii)}
\label{Appendice SS section eqn Watson ii}

The transfer matrix 
\beq
\op{t}_{1,\dots,n}(\bs{\be}_n ; \la) \; = \; \e{Tr}_{0}\Big[ \op{T}_{1,\dots,n;0}(\bs{\be}_n ; \la) \Big]
\enq
admits the explicit product form at $\la=\be_1$
\beq
\op{t}_{1,\dots,n}(\bs{\be}_n ; \be_1)  \, = \, - \op{S}_{21}(\be_{21})\cdots \op{S}_{n 1}(\be_{n1})  \;. 
\label{ecriture forme produit explicite pour matrice de transfer aux inhomogeneites}
\enq
As a consequence, Watson's equations may be recast as 
\beq
\mc{F}^{(\op{O})}\big( \bs{\be}_n+2\i\pi \bs{e}_1 \big) \ex{- 2\i\pi \om_{\op{O}} -\i\pi \ga \sg_1^{z} }  \, = \, - 
\mc{F}^{(\op{O})}\big( \bs{\be}_n \big) \cdot \op{t}_{1,\dots,n}(\bs{\be}_n ; \be_1)   \, . 
\enq
When computing the action of $\op{t}_{1,\dots,n}(\bs{\be}_n ; \be_1)$ on $\bs{\Psi}_{1,\dots,n}\big[ p\big]\big( \bs{\be}_n \big)$, 
one first has to exchange the transfer matrix $\op{t}_{1,\dots,n}(\bs{\be}_n ; \be_1)$ with $ \op{O}_{1,\dots,n}^{(\ga)}\big(\bs{\be}_n\big)$. 
Owing to \eqref{ecriture equation echange operateur spin twist et matrice S} and the explicit product form \eqref{ecriture forme produit explicite pour matrice de transfer aux inhomogeneites}, 
one gets that 
\beq
\op{O}_{1,\dots,n}^{(\ga)}\big(\bs{\be}_n\big) \op{t}_{1,\dots,n}(\bs{\be}_n ; \be_1) \; = \; \op{t}_{1,\dots,n}^{(\ga)}(\bs{\be}_n ; \be_1)   \op{O}_{1,\dots,n}^{(\ga)}\big(\bs{\be}_n\big) \;, 
\enq
what yields 
\beq
\, - \, \bs{\Psi}_{1,\dots,n}\big[ p\big]\big( \bs{\be}_n \big) \cdot \op{t}_{1,\dots,n}(\bs{\be}_n ; \be_1)  \, = \, - \Int{ (\msc{C}_{\bs{\be}_n})^m }{} \dd^{m} u \; 
\wt{\Cx}_{1,\dots,n}^{(\ga)} \big(\bs{\be}_n;\bs{u}_m \big) \cdot \op{t}_{1,\dots,n}^{(\ga)}(\bs{\be}_n ; \be_1)
\cdot \big( g \cdot p \big)\big(\bs{\be}_n;\bs{u}_{m} \big)  \cdot  \op{O}_{1,\dots,n}^{(\ga)}\big(\bs{\be}_n\big) \;.
\enq
At this stage, one should  compute the action of the transfer matrix $\op{t}_{1,\dots,n}^{(\ga)}(\bs{\be}_n ; \be_1)\, = \,  \op{A}_{1,\dots,n}^{(\ga)}(\bs{\be}_n ; \be_1)+ \op{D}_{1,\dots,n}^{(\ga)}(\bs{\be}_n ; \be_1)$ 
on the covector $\wt{\Cx}_{1,\dots,n}^{(\ga)} \big(\bs{\be}_n;\bs{u}_m \big)  $ arising in the integral representation for
$\bs{\Psi}_{1,\dots,n}\big[ p\big]\big( \bs{\be}_n \big)$ by using the formulae \eqref{ecriture recurrence pour A}-\eqref{ecriture recurrence pour D}. 
Observe also, that one may slightly accommodate the integration curve $\msc{C}_{\bs{\be}_n}$ arising in the definition 
of $\bs{\Psi}_{1,\dots,n}\big[ p\big]\big( \bs{\be}_n \big)$ \textit{prior} to computing the action of the $\op{A}_{1,\dots,n}^{(\ga)}(\bs{\be}_n ; \be_1)$ operator. 
Indeed, as already discussed the integrand has its only poles at the points 
\beqa
&& u_b = \be_a-\i\pi +\i k \xi \quad k \in \mathbb{Z}  \\
&& u_b = \be_a +\i p \xi +\i\pi (2\ell + 1) \; , \quad u_b = \be_a - \i p \xi  - 2 \i\pi \ell \quad p, \ell \in \mathbb{N} \;.  
\eeqa
Therefore, one may deform the original integration curve $\msc{C}_{\bs{\be}_n}$ to $\wh{\msc{C}}_{\bs{\be}_n+2\i\pi \bs{e}_1}$ which is such that:
\begin{itemize}

\item $\wh{\msc{C}}_{\bs{\be}_n+2\i\pi \bs{e}_1}$ connects $-\infty$ to $+\infty$;

\item  the points $\be_a + 2\i\pi \de_{a1} - \i p \xi - 2 \i \pi \ell $, $\be_a + 2\i\pi \de_{a1} - \i\pi - \i (p + 1) \xi$, $p, \ell \in \mathbb{N}$,   have index $-1$ in respect to $\wh{\msc{C}}_{\bs{\be}_n+2\i\pi \bs{e}_1}$, 
\textit{viz}. they are located below of $\wh{\msc{C}}_{\bs{\be}_n+2\i\pi \bs{e}_1}$;

\item  the points $\be_a+ \i p \xi + (2 \ell + 1) \i \pi$, $\be_a -\i\pi +\i p \xi$, $p, \ell \in \mathbb{N}$,  have index $1$ in respect to $\wh{\msc{C}}_{\bs{\be}_n+2\i\pi \bs{e}_1}$,
\textit{viz}. they are located above of $\wh{\msc{C}}_{\bs{\be}_n+2\i\pi \bs{e}_1}$.

\end{itemize}
All-in-all, upon using the symmetry properties of the integrand in respect to permutations of the coordinates  of $\bs{u}_m$, one gets that 
\beq
\, - \, \bs{\Psi}_{1,\dots,n}\big[ p\big]\big( \bs{\be}_n \big) \cdot \op{t}_{1,\dots,n}(\bs{\be}_n ; \be_1)  \, = \, 
\sul{\a=1}{3}\bs{\Psi}_{1,\dots,n}^{(\a)}\big[ p\big]\big( \bs{\be}_n \big)
\label{ecriture action matrice transfer sur Psi offShell}
\enq
in which 
\beq
\bs{\Psi}_{1,\dots,n}^{(1)}\big[ p\big]\big( \bs{\be}_n \big) \, = \hspace{-3mm} \Int{ (\wh{\msc{C}}_{\bs{\be}_n+2\i\pi \bs{e}_1})^m }{} \hspace{-4mm} \dd^{m} u \; 
\wt{\Cx}_{1,\dots,n}^{(\ga)}\big(\bs{\be}_n+2\i\pi \bs{e}_1;\bs{u}_m \big)  \cdot \Phi^{(1)}\big( \bs{\be}_n; \bs{u}_m\big)  \;, 
\enq
where 
\beq
\Phi^{(1)}\big( \bs{\be}_n; \bs{u}_m\big) \, = \,
-\, \f{   \pl{k=1}{n} a(\be_{k1}) }{ \pl{k=1}{m} \wt{b}(\be_1-\mu_k+2\i\pi) }  \cdot  \big( g \cdot p \big)\big(\bs{\be}_n;\bs{u}_{m} \big) \cdot  \op{O}_{1,\dots,n}^{(\ga)}\big(\bs{\be}_n\big) \;.
\enq
Here we remind that $\wt{\Cx}_{1,\dots,n}^{(\ga)}\big(\bs{\be}_n;\bs{u}_m \big)$ has been introduced in \eqref{definition tilde C vect de Bethe}, 
while $\wt{g}\big(\bs{\be}_n;\bs{u}_{m} \big)$ in \eqref{definition fonction tilde g}. Further, one has 
\beq
\bs{\Psi}_{1,\dots,n}^{(2)}\big[ p\big]\big( \bs{\be}_n \big) \, = \,   m  \hspace{-3mm} \Int{ (\wh{\msc{C}}_{\bs{\be}_n+2\i\pi \bs{e}_1})^m }{} \hspace{-4mm} \dd^{m} u \; 
\wt{\Cx}^{(\ga)}_{1,\dots,n}\big( \bs{\be}_n  \, ; \, (\bs{u}_{m-1},\be_1)  \big)  \cdot \Phi^{(2)}\big( \bs{\be}_n; \bs{u}_m\big)
\enq
as well as
\beq
\bs{\Psi}_{1,\dots,n}^{(3)}\big[ p\big]\big( \bs{\be}_n \big) \, = \,   m  \hspace{-1mm} \Int{ (\msc{C}_{\bs{\be}_n})^m }{} \hspace{-1mm} \dd^{m} u \; 
\wt{\Cx}^{(\ga)}_{1,\dots,n}\big( \bs{\be}_n  \, ; \, (\bs{u}_{m-1},\be_1)  \big) \cdot 
\Phi^{(3)}\big( \bs{\be}_n; \bs{u}_m\big) \;.
\enq
The two above vectors are defined in terms of 
\beqa
\Phi^{(2)}\big( \bs{\be}_n; \bs{u}_m\big) & =& \f{ c^{(\ga)}( \be_1 + 2\i\pi - u_m ) }{ b(\be_1 + 2\i\pi - u_m ) } \pl{ s= 1   }{ m-1 } \f{ a( u_{ms}) }{ b( u_{ms} ) } 
\pl{k=1}{n} a(\be_{k1}) \cdot
(g  \cdot p) \big(\bs{\be}_n;\bs{u}_{m} \big)  \cdot  \op{O}_{1,\dots,n}^{(\ga)}\big(\bs{\be}_n\big)  \;,   \\
\Phi^{(3)}\big( \bs{\be}_n; \bs{u}_m\big) & =&   \f{ c^{(-\ga)}( u_m - \be_1 ) }{ b( u_m - \be_1 ) } \pl{ s = 1   }{ m-1 } \f{ a( u_{sm} ) }{b( u_{sm} )} 
\pl{k=1}{n} \Big\{ a(\be_{k1}) \wt{b}(\be_k- u_m ) \Big\}  \cdot   (g \cdot  p) \big(\bs{\be}_n;\bs{u}_{m} \big)   \cdot  \op{O}_{1,\dots,n}^{(\ga)}\big(\bs{\be}_n\big)   \;. 
\eeqa
Now, upon using the $2\i\pi$-periodicity properties of $p\big( \bs{\be}_n ; \bs{u}_{m} \big)$ given in $\mathrm{d)}$,  one has that 
\beq
\Phi^{(1)}\big( \bs{\be}_n; \bs{u}_m\big) \, = \,  \pl{k=2}{m} \f{ a(\be_{k1}) \op{F}(\be_{1k})  }{ \op{F}(\be_{1k}+2\i\pi) }
\pl{s=1}{m} \f{  (a\phi)(\be_1-u_s) }{ (a\phi\wt{b})(\be_1-\mu_s+2\i\pi)  }
\big( g \cdot p \big)\big(\bs{\be}_n+2\i\pi \bs{e}_1 ; \bs{u}_{m} \big) \cdot  \op{O}_{1,\dots,n}^{(\ga)}\big( \bs{\be}_n + 2\i\pi \bs{e}_1 \big)  \cdot \ex{-2\i\pi \om -\i\pi \ga \sg_1^{z} } \;. 
\enq
The first term reduces to $1$ by virtue of the functional equation for $\op{F}$ given in \eqref{ecriture eqn fnelle FF minimal},
while the second owing to the identity
\beq
(a \phi)\big( \la + 2\i\pi \big) \, = \, \f{ (a \phi)\big( \la \big)  }{ \wt{b}\big( \la + 2\i\pi \big)   }  \;. 
\label{ecriture shift par 2 i pi de a phi}
\enq
In a similar way
\bem
\Phi^{(3)}\big( \bs{\be}_n; \bs{u}_m-2\i\pi \bs{e}_m\big) \,  =\, \f{ c^{(-\ga)}( u_m-\be_1 - 2\i\pi ) }{ b(u_m-\be_1 - 2\i\pi ) }   \pl{k=1}{n} a(\be_{k1})
\pl{ s= 1   }{ m-1 } \bigg\{  \f{  \tau( u_{sm}+2\i\pi ) }{   \tau(u_{sm}) \wt{b}( u_{sm}+2\i\pi ) } \bigg\} \\ 
\\ \times    \pl{k=1}{n}  \bigg\{ \f{ (a\phi\wt{b})(\be_k-u_m+2\i\pi) }{  (a\phi)(\be_k - u_m)  } \bigg\}
\cdot (g   p) \big(\bs{\be}_n;\bs{u}_{m} \big) \cdot  \op{O}_{1,\dots,n}^{(\ga)}\big(\bs{\be}_n\big) 
\; = \; - \Phi^{(2)}\big( \bs{\be}_n; \bs{u}_m \big)  \;. 
\end{multline}
There, we made use of \eqref{ecriture shift par 2 i pi de a phi}, as well as 
\beq
\f{c^{(\ga)}(\th) }{ b(\th) } \; = \; - \f{c^{(-\ga)}(-\th) }{ b(-\th) } \qquad \e{and} \qquad 
\tau(\th+2\i\pi) \, = \, \tau(\th) \cdot \f{ \wt{b}(\th+2\i\pi) }{  \wt{b}(-\th) } \;. 
\enq
As a consequence,  
one may recast the building blocks arising in \eqref{ecriture action matrice transfer sur Psi offShell} as
\beq
\bs{\Psi}_{1,\dots,n}^{(1)}\big[ p\big]\big( \bs{\be}_n \big) \, =   \hspace{-3mm} \Int{ (\wh{\msc{C}}_{\bs{\be}_n+2\i\pi \bs{e}_1})^m }{} \hspace{-4mm} \dd^{m} u \; 
\wt{\Cx}_{1,\dots,n}^{(\ga)}\big(\bs{\be}_n+2\i\pi \bs{e}_1;\bs{u}_m \big) \cdot  \op{O}_{1,\dots,n}^{(\ga)}\big( \bs{\be}_n + 2\i\pi \bs{e}_1 \big)  \cdot
\big( g \cdot p \big)\big(\bs{\be}_n+2\i\pi \bs{e}_1;\bs{u}_{m} \big) \ex{-2\i\pi \om - \i\pi \ga \sg_1^{z} } \;, 
\enq
while
\beq
\bs{\Psi}_{1,\dots,n}^{(3)}\big[ p\big]\big( \bs{\be}_n \big) \, = \,    - m  \hspace{-1mm} \Int{ (\msc{C}_{\bs{\be}_n})^m }{} \hspace{-1mm} \dd^{m} u \; 
\wt{\Cx}_{1,\dots,n}^{(\ga)}\big( \bs{\be}_n  \, ; \, (\bs{u}_{m-1},\be_1)  \big) 
\cdot \Phi^{(2)}\big( \bs{\be}_n; \bs{u}_m+2\i\pi \bs{e}_m\big) \;. 
\enq
We first justify that $\bs{\Psi}_{1,\dots,n}^{(1)}\big[ p\big]\big( \bs{\be}_n \big)$ has the sought form and then 
prove that the term $\bs{\Psi}_{1,\dots,n}^{(3)}\big[ p\big]\big( \bs{\be}_n \big)$ compensates with $\bs{\Psi}_{1,\dots,n}^{(2)}\big[ p\big]\big( \bs{\be}_n \big)$. 
Regarding $\bs{\Psi}_{1,\dots,n}^{(1)}\big[ p\big]\big( \bs{\be}_n \big)$, according to the previous analysis, its integrand has poles at
\beqa
&& u_b = \be_a+2\i\pi \de_{a1}-\i\pi +\i k \xi \quad k \in \mathbb{Z}  \\
&& u_b = \be_a+2\i\pi \de_{a1} +\i p  \xi +\i\pi (2n+1) \; , \quad u_b = \be_a+2\i\pi \de_{a1} - \i p \xi  - 2 \i\pi \ell \quad p, \ell \in \mathbb{N} \;.  
\eeqa
%
%
%
%
%
%




%
%
%
Since these are the only singularities of the integrand, one may deform $\wh{\msc{C}}_{\bs{\be}_n+2\i\pi \bs{e}_1}$ into $\msc{C}_{\bs{\be}_n+2\i\pi \bs{e}_1}$
in $\bs{\Psi}_{1,\dots,n}^{(1)}\big[ p\big]\big( \bs{\be}_n \big)$, what then yields
\beq
\bs{\Psi}_{1,\dots,n}^{(1)}\big[ p\big]\big( \bs{\be}_n \big) \; = \;   \bs{\Psi}_{1,\dots,n}\big[ p\big]\big( \bs{\be}_n + 2\i\pi \bs{e}_1 \big)\ex{-2\i\pi \om - \i\pi \ga \sg_1^{z} } \;.
\enq
In order to discuss the pole structure, in respect to $u_m$ of the integrand in $\bs{\Psi}_{1,\dots,n}^{(3)}\big[ p\big]\big( \bs{\be}_n \big)$, we first observe that, by construction, 
$\msc{C}_{\bs{\be}_n}$ and $\wh{\msc{C}}_{\bs{\be}_n+2\i\pi \bs{e}_1}$ both partition the poles of the integrand 
in respect to $u_1,\dots, u_m$ into two same sets $\mc{A}_+$ and $\mc{A}_-$ associated with the same index. Therefore, 
in the integrals over $\bs{u}_{m-1}$ one may carry out the contour deformation $\msc{C}_{\bs{\be}_n} \hookrightarrow \wh{\msc{C}}_{\bs{\be}_n+2\i\pi \bs{e}_1}$. 
In what concerns the $u_m$ integral, one has the decomposition 
\beq
\Phi^{(2)}\big( \bs{\be}_n; \bs{u}_m \big) \; = \;   \f{ c^{(\ga)}(  \be_1 - u_m + 2\i\pi) }{ b(  \be_1 - u_m + 2\i\pi ) } 
\pl{ s=1   }{ m-1 } \bigg\{  \, \f{ \tau(u_{sm})  }{ \wt{b}(u_{ms}) }  \bigg\}
\pl{k=1}{n} (a \cdot \phi )(\be_{k} - u_m ) \cdot    \wt{\Phi}^{(2)}\big( \bs{\be}_n; \bs{u}_{m-1}\big) p \big(\bs{\be}_n;\bs{u}_{m} \big)   \;,
\enq
in which $ \wt{\Phi}^{(2)}$ does not depend on $u_m$. 

From there, one infers that the integrand only has simple poles and these are given by 
\beq
u_m \, = \, \be_k -\i \pi \, + \,  \i \xi \ell \; , \quad u_m \, = \, \be_k +\i \pi(2p+1) \, + \,  \i \xi \ell \; , \quad 
u_m \, = \, \be_k - 2 \i \pi (p+1) \, - \,  \i \xi \ell
\enq
and $u_m=\be_1-\i \ell \xi$, with $p, \ell \geq 0$.

In the integral representation for $\bs{\Psi}_{1,\dots,n}^{(3)}\big[ p\big]\big( \bs{\be}_n \big)$, one would like to deform the original contour of the $u_m$-integration form 
$\msc{C}_{ \bs{\be}_n }$ up to $\wh{ \msc{C} }_{ \bs{\be}_n + 2\i\pi \bs{e}_1 }-2\i\pi$. Observe that 
the contour $\wh{ \msc{C} }_{ \bs{\be}_n + 2\i\pi \bs{e}_1 }-2\i\pi$
has the property 
\begin{itemize}

\item  the points $\be_a + 2\i\pi \de_{a1} - \i p \xi - 2 \i \pi (\ell+1) $, $\be_a + 2\i\pi \de_{a1} - 3 \i\pi - \i (p+1) \xi$, $p, \ell \in \mathbb{N}$,   have index $-1$ in respect to 
$\wh{\msc{C}}_{\bs{\be}_n+2\i\pi \bs{e}_1} -2\i\pi$;

\item  the points $\be_a  +\i p \xi + (2 \ell - 1) \i \pi$, $\be_a -3\i\pi +\i p \xi$, $p, \ell \in \mathbb{N}$,  have index $1$ in respect to $\wh{\msc{C}}_{\bs{\be}_n+2\i\pi \bs{e}_1} -2\i\pi$;

\end{itemize}
while the original contour has the property 
\begin{itemize}

\item  the points $\be_a - \i p \xi - 2 \i \pi \ell $, $\be_a  - \i\pi - \i (p+1) \xi$, $p, \ell \in \mathbb{N}$,   have index $-1$ in respect to $\msc{C}_{\bs{\be}_n}$;

\item  the points $\be_a+\i p \xi + (2 \ell + 1) \i \pi$, $\be_a-\i\pi +\i p \xi$, $p, \ell \in \mathbb{N}$,  have index $1$ in respect to $\msc{C}_{\bs{\be}_n}$.

\end{itemize}
Hence, in the sense of homotopy classes, one has 
\bem
\msc{C}_{ \bs{\be}_n } \setminus \Big\{ \wh{ \msc{C} }_{ \bs{\be}_n + 2\i\pi \bs{e}_1 }-2\i\pi \Big\}  \; \simeq \; 
\bigcup\limits_{a=1}^{n} \Bigg\{   \bigg( \bigcup\limits_{\ell, p \geq 0}\Big\{ - \Dp{}\mc{D}_{\be_a-\i p  \xi - 2  \i \pi \ell,\eps}  \Big\}
\bigcup\limits_{\ell, p \geq 0}  \Big\{  \Dp{}\mc{D}_{\be_a+ 2\i\pi\de_{a1}-\i p  \xi - 2  \i (\ell+1) \pi ,\eps}  \Big\}  \bigg)   \\
\bigcup   \bigg( \bigcup \limits_{  p \geq 0}\Big\{  - \Dp{}\mc{D}_{\be_a-\i\pi - \i (p+1) \xi ,\eps}  \Big\}  
\bigcup \limits_{  p \geq 0}\Big\{    \Dp{}\mc{D}_{\be_a+2\i\pi \de_{a1}- 3\i\pi - \i (p+1) \xi ,\eps}  \Big\}  \bigg)\\ 
\bigcup \bigg( \bigcup\limits_{\ell, p \geq 0}\Big\{   \Dp{}\mc{D}_{\be_a + \i p \xi + \i \pi (2\ell+1),\eps}  \Big\}  \bigg) 
\bigcup\limits_{\ell, p \geq 0}\Big\{  - \Dp{}\mc{D}_{\be_a + \i p \xi + \i\pi (2\ell-1) ,\eps}  \Big\}  \bigg)
\bigcup 
\bigg( \bigcup\limits_{p  \geq 0}\Big\{  \Dp{}\mc{D}_{\be_a + \i p  \xi - \i\pi ,\eps}  \Big\}  
\bigcup\limits_{p \geq 0}\Big\{  - \Dp{}\mc{D}_{\be_a + \i p \xi - 3 \i\pi ,\eps}  \Big\}   \bigg) \Bigg\} \\
\; = \;  \bigcup\limits_{a=2}^{n} \Bigg\{ \bigcup\limits_{p  \geq 0}  \bigg( \Big\{  - \Dp{}\mc{D}_{\be_a - \i p \xi  ,\eps} \Big\}
\bigcup \Big\{  - \Dp{}\mc{D}_{\be_a - \i (p+1) \xi -\i\pi ,\eps} \Big\} \bigcup \Big\{   \Dp{}\mc{D}_{\be_a - \i (p+1) \xi -3\i\pi ,\eps} \Big\} \bigg) \Bigg\}\\
\bigcup\limits_{a=1}^{n} \Bigg\{ \bigcup\limits_{p  \geq 0}  \bigg( \Big\{  - \Dp{}\mc{D}_{\be_a + \i p \xi -\i\pi ,\eps} \Big\}
\bigcup \Big\{   \Dp{}\mc{D}_{\be_a + \i p \xi -\i\pi ,\eps} \Big\} \bigcup \Big\{ -  \Dp{}\mc{D}_{\be_a + \i p \xi -3\i\pi ,\eps} \Big\} \bigg) \Bigg\}\\ 
\; = \; \bigcup\limits_{a=2}^{n} \Bigg\{ \bigcup\limits_{p  \geq 0}  \bigg( \Big\{  - \Dp{}\mc{D}_{\be_a - \i p \xi  ,\eps} \Big\}
\bigcup \Big\{  - \Dp{}\mc{D}_{\be_a - \i (p+1) \xi -\i\pi ,\eps} \Big\} \bigcup \Big\{   \Dp{}\mc{D}_{\be_a - \i (p+1) \xi -3\i\pi ,\eps} \Big\} \bigg)\Bigg\} 
\bigcup\limits_{a=1}^{n} \bigcup\limits_{p  \geq 0}  \Big\{ -  \Dp{}\mc{D}_{\be_a + \i p \xi -3\i\pi ,\eps} \Big\}  \;. 
\end{multline}
Therefore, the only additional points that are grasped by such a deformation are
\beq
\be_a - \i p \xi\, , \, \be_a + \i\xi p - \i\pi  \quad \e{and} \quad \be_a - \i\pi - \i \xi (p+1)  
\quad \e{and} \quad \be_a -3\i\pi - \i \xi (p+1) \quad \e{with} \quad p \in \mathbb{N} \;,  
\enq
and for $a \in \intn{2}{n}$, as well as 
\beq
\be_a -3\i\pi + \i p \xi \quad \e{with} \quad p \in \mathbb{N}\quad \e{and} \quad a \in \intn{1}{n} \;.   
\enq
However, these are not poles of the integrand, so there is no additional contribution coming from the deformation.
Finally, regarding the other integration variables $u_1,\dots, u_{m-1}$, one may deform the contours $\msc{C}_{\bs{\be}_n} \hookrightarrow \wh{\msc{C}}_{\bs{\be}_n+2\i\pi \bs{e}_1} $
as, again, no poles are crossed in the procedure. As a consequence, 
\beqa
\bs{\Psi}_{1,\dots,n}^{(3)}\big[ p\big]\big( \bs{\be}_n \big) & = &  m  \hspace{-3mm} \Int{ (\wh{\msc{C}}_{\bs{\be}_n+2\i\pi \bs{e}_1})^{m-1}   }{} \hspace{-3mm} \dd^{m-1} u \hspace{-5mm}
\Int{ \wh{\msc{C}}_{\bs{\be}_n+2\i\pi \bs{e}_1} -2\i\pi }{} \hspace{-6mm} \dd u_m \; 
\wt{\Cx}^{(\ga)}_{1,\dots,n}\big( \bs{\be}_n  \, ; \, (\bs{u}_{m-1},\be_1)  \big) 
\cdot \Phi^{(2)}\big( \bs{\be}_n; \bs{u}_m+2\i\pi \bs{e}_m\big)  \\
& = & - \bs{\Psi}_{1,\dots,n}^{(2)}\big[ p\big]\big( \bs{\be}_n \big) \;.
\eeqa

\qed

\subsection{The Lorentz invariance equation iii)}
\label{Appendice SS section eqn inv Lorentz iii}

One starts with the expression of the form factor with a constant shift $\th$:
\beq
\bs{\Psi}_{1,\dots,n}\big[ p\big]\big( \bs{\be}_n + \th \, \ov{\bs{e}}_n \big) \, = \hspace{-4mm} \Int{ (\msc{C}_{\bs{\be}_n + \th \, \ov{\bs{e}}_n})^m }{} \hspace{-4mm} \dd^{m} u \;
\wt{\Cx}_{1,\dots,n}^{\, (\ga)}\big(\bs{\be}_n + \th \, \ov{\bs{e}}_n;\bs{u}_m \big) \cdot \op{O}_{1,\dots,n}^{(\ga)}(\bs{\be}_n + \th \, \ov{\bs{e}}_n)  \cdot \big( g \cdot p \big)\big(\bs{\be}_n + \th \, \ov{\bs{e}}_n;\bs{u}_{m} \big)   \;.
\enq
From \eqref{monodromy matrix}, \eqref{definition fonction g} as well as the Lorentz invariance property for the $p$-function, one may rewrite the above as
\beq
\bs{\Psi}_{1,\dots,n}\big[ p\big]\big( \bs{\be}_n + \th \, \ov{\bs{e}}_n \big) \, = \hspace{-4mm} \Int{ (\msc{C}_{\bs{\be}_n + \th \, \ov{\bs{e}}_n})^m }{} \hspace{-4mm} \dd^{m} u \;
\wt{\Cx}_{1,\dots,n}^{\, (\ga)}\big(\bs{\be}_n ;\bs{u}_m - \th \, \ov{\bs{e}}_m \big) \cdot \op{O}_{1,\dots,n}^{(\ga)}(\bs{\be}_n + \th \, \ov{\bs{e}}_n)  \cdot \big( g \cdot p \big)\big(\bs{\be}_n ;\bs{u}_{m} - \th \, \ov{\bs{e}}_m \big)   \cdot \ex{\th s} \;.
\enq
The change of variables $\bs{u}_m \mapsto \bs{u}_m + \th \, \ov{\bs{e}}_m$ entails
\beq
\bs{\Psi}_{1,\dots,n}\big[ p\big]\big( \bs{\be}_n + \th \, \ov{\bs{e}}_n \big) \, = \, \Int{ (\msc{C}_{\bs{\be}_n + \th \, \ov{\bs{e}}_n} - \th)^m}{} \dd^{m} u \; 
\wt{\Cx}_{1,\dots,n}^{\, (\ga)}\big(\bs{\be}_n ;\bs{u}_m \big) \cdot \op{O}_{1,\dots,n}^{(\ga)}(\bs{\be}_n + \th \, \ov{\bs{e}}_n)  \cdot \big( g \cdot p \big)\big(\bs{\be}_n ;\bs{u}_{m}  \big)   \cdot \ex{\th s} \;.
\enq
Furthermore, one may once again deform the contours back to the original $\msc{C}_{\bs{\be}_n + \th \, \ov{\bs{e}}_n} - \th \hookrightarrow \msc{C}_{\bs{\be}_n}$ without crossing any poles in the
procedure. Now, one notices that
\begin{equation}
\op{O}_{1,\dots,n}^{(\ga)}(\bs{\be}_n + \th \, \ov{\bs{e}}_n) = \ex{\f{\ga}{2} \th \Sigma^z_{1,\dots,n}} \cdot \op{O}_{1,\dots,n}^{(\ga)}(\bs{\be}_n)
\end{equation}
with $\Sigma^z_{1,\dots,n} = \sul{k=1}{n} \sg^z_k$, and makes use of the commutation relation
\begin{equation}
	\left[ \wt{\op{C}}_{1,\dots,n}^{\, (\ga)}\big(\bs{\be}_n ;u \big) \, , \, \Sigma^z_{1,\dots,n}\right] = -2 \wt{\op{C}}_{1,\dots,n}^{\, (\ga)}\big(\bs{\be}_n ; u \big)
\end{equation}
to get
\begin{equation}
\wt{\Cx}_{1,\dots,n}^{\, (\ga)}\big(\bs{\be}_n ;\bs{u}_m \big)  \cdot \Sigma^z_{1,\dots,n} = (n-2m) \cdot \wt{\Cx}_{1,\dots,n}^{\, (\ga)}\big(\bs{\be}_n ;\bs{u}_m \big)
= q \cdot \wt{\Cx}_{1,\dots,n}^{\, (\ga)}\big(\bs{\be}_n ;\bs{u}_m \big)\;.
\end{equation}
All of the above thus entails the claim. \qed

\subsection{The kinematic pole equation iv)}
\label{subsection axiom iv}

The kinematical poles in $\bs{\Psi}_{1,\dots,n}\big[ p\big]\big( \bs{\be}_n \big)$ arise from the pinching of the contour $\msc{C}_{\bs{\be}_n}$ 
by the poles it is supposed to separate when $\be_1 \tend \be_n+\i\pi$. One may check that in this case, there are three pinchings occurring simultaneously:
\begin{itemize}
\item The pole at $\be_n$ located below of  $\msc{C}_{\bs{\be}_n}$  pinches the pole at $\be_1-\i\pi$ located above  $\msc{C}_{\bs{\be}_n}$;
\item The pole at $\be_1-2\i\pi$ located below of  $\msc{C}_{\bs{\be}_n}$  pinches the pole at $\be_n -\i\pi$ located above  $\msc{C}_{\bs{\be}_n}$;
\item The pole at $\be_1$ located below of  $\msc{C}_{\bs{\be}_n}$  pinches the pole at $\be_n + \i\pi$ located above  $\msc{C}_{\bs{\be}_n}$. 
\end{itemize}
One may depict the original contour, locally around the points $\be_1, \be_n$, as given in Figure
\ref{contours Gamma GammaPrime et C BetaN Prime}. Call this contour $\Ga$ for definiteness  and
denote $\Ga^{\prime}$ its deformation as given in Figure \ref{contours Gamma GammaPrime et C BetaN Prime}. Note that the contour $\Ga^{\prime}$ does not enjoy the
pinching property when $\be_{1n}\tend \i\pi$. 
Below, we establish a technical lemma clarifying the appearance of pole through the pinching mechanism.

\begin{figure}[h!]
\begin{center}

\begin{minipage}[t]{0.31\textwidth}
	\centering
		\begin{tikzpicture}[scale=1]
		\coordinate (A) at (4.00, 1.90);
		\coordinate (B) at (4.00, 3.10);
		\coordinate (C) at (4.00, 3.90);
		\coordinate (D) at (4.00, 5.10);
		\coordinate (E) at (4.00, 5.90);
		\coordinate (F) at (4.00, 7.10);
		\coordinate (G) at (3.00, 1.30);
		\coordinate (H) at (5.00, 3.00);
		\coordinate (I) at (3.00, 2.75);
		\coordinate (J) at (5.00, 2.00);
		\coordinate (K) at (4.00, 4.50);
		\coordinate (L) at (5.03, 3.34);
		\coordinate (M) at (5.00, 4.50);
		\coordinate (N) at (4.99, 6.49);
		\coordinate (O) at (1.99, 4.50);
		\coordinate (P) at (2.25, 6.50);
		\coordinate (Q) at (6.79, 6.50);
		\coordinate (R) at (6.00, 6.50);
		\coordinate (S) at (6.24, 6.53);
		
		\draw[thick] (G) .. controls (3.00, 2.75) and (5.00, 2.00) .. (H);
		\draw[thick] (H) .. controls (5.03, 3.34) and (5.00, 4.50) .. (K);
		\draw[thick] (K) .. controls (1.99, 4.50) and (2.25, 6.50) .. (N);
		\draw[
		thick,
		postaction={decorate},
		decoration={
			markings,
			mark=at position 0.6 with {\arrow{>}}
		}
		] (N) .. controls (6.00, 6.50) and (6.24, 6.53) .. (Q);
		\fill[blue] (A) circle (2pt);
		\node[right, blue] at (A) {$\beta_1 - 2\i\pi$};
		\fill[blue] (B) circle (2pt);
		\node[left, blue] at (B) {$\beta_n- \i\pi$};
		\fill[blue] (C) circle (2pt);
		\node[left, blue] at (C) {$\beta_1 - \i\pi$};
		\fill[blue] (D) circle (2pt);
		\node[right, blue] at (D) {$\beta_n$};
		\fill[blue] (E) circle (2pt);
		\node[right, blue] at (E) {$\beta_1$};
		\fill[blue] (F) circle (2pt);
		\node[left, blue] at (F) {$\beta_n + \i\pi$};
		\node[above right, black] at (S) {$\Gamma$};

	\end{tikzpicture}
\end{minipage}
\hfill 
\begin{minipage}[t]{0.31\textwidth}
	\centering

	\begin{tikzpicture}[scale=1]
	\coordinate (A) at (4.00, 1.90);
	\coordinate (B) at (4.00, 3.10);
	\coordinate (C) at (4.00, 3.90);
	\coordinate (E) at (4.00, 5.90);
	\coordinate (F) at (4.00, 7.10);
	\coordinate (I) at (4.00, 5.10);
	\coordinate (J) at (6.00, 0.99);
	\coordinate (K) at (2.00, 4.11);
	\coordinate (L) at (6.00, 6.00);
	\coordinate (M) at (2.00, 1.75);
	\coordinate (N) at (6.00, 7.31);
	\coordinate (O) at (2.00, 6.00);
	\coordinate (P) at (6.00, 4.75);
	
	\draw[thick,
	postaction={decorate},
	decoration={
		markings,
		mark=at position 0.4 with {\arrow{<}}
	},red] (A) circle (0.25);
	\draw[thick,
	postaction={decorate},
	decoration={
		markings,
		mark=at position 0.3 with {\arrow{<}}
	},red] (E) circle (0.25);
	\draw[thick,
	postaction={decorate},
	decoration={
		markings,
		mark=at position 0.1 with {\arrow{>}}
	},
	red] (C) circle (0.25);
	\draw[thick] (J) .. controls (6.00, 6.00) and (2.00, 1.75) .. (K);
	\draw[thick,
	postaction={decorate},
	decoration={
		markings,
		mark=at position 0.6 with {\arrow{>}}
	}] (K) .. controls (2.00, 6.00) and (6.00, 4.75) .. (N);
	\fill[blue] (A) circle (2pt);
	\node[right, blue] at (4.20,1.90) {$\beta_1 - 2\i\pi$};
	\fill[blue] (B) circle (2pt);
	\node[right, blue] at (B) {$\beta_n- \i\pi$};
	\fill[blue] (C) circle (2pt);
	\node[right, blue] at (4.20,3.90) {$\beta_1 - \i\pi$};
	\fill[blue] (E) circle (2pt);
	\node[right, blue] at (4.20,5.90) {$\beta_1$};
	\fill[blue] (F) circle (2pt);
	\node[right, blue] at (F) {$\beta_n + \i\pi$};
	\fill[blue] (I) circle (2pt);
	\node[right, blue] at (I) {$\beta_n$};
	
	\node[below right, black] at (N) {$\Gamma^{\prime}$};

\end{tikzpicture}
\end{minipage}
\hfill 
\begin{minipage}[t]{0.31\textwidth}
\centering
\begin{tikzpicture}[scale=1]
	\coordinate (A) at (4.00, 1.90);
	\coordinate (C) at (4.00, 3.90);
	\coordinate (E) at (4.00, 5.90);
	\coordinate (J) at (7.52, 0.97);
	\coordinate (K) at (2.99, 3.99);
	\coordinate (L) at (6.00, 6.00);
	\coordinate (M) at (3.00, 1.75);
	\coordinate (N) at (6.20, 6.89);
	\coordinate (O) at (3.00, 6.00);
	\coordinate (P) at (7.00, 3.25);
	
	\draw[red] (A) circle (0.00);
	\draw[red] (E) circle (0.00);
	\draw[red] (C) circle (0.00);
	\draw[thick] (J) .. controls (6.00, 6.00) and (3.00, 1.75) .. (K);
	\draw[thick,
	postaction={decorate},
	decoration={
		markings,
		mark=at position 0.6 with {\arrow{>}}
	}] (K) .. controls (3.00, 6.00) and (7.00, 3.25) .. (N);
	\fill[blue] (A) circle (2pt);
	\node[right, blue] at (A) {$\beta_1 - 2\i\pi = \beta_n - \i\pi$};
	\fill[blue] (C) circle (2pt);
	\node[right, blue] at (C) {$\beta_1 - \i\pi = \beta_n$};
	\fill[blue] (E) circle (2pt);
	\node[right, blue] at (E) {$\beta_1 = \beta_n + \i\pi$};
	\fill[blue] (A) circle (2pt);
	
	\node[below right, black] at (6.50,6.00) {$\mathscr{C}^{\prime}_{\boldsymbol{\beta}_n}$};

\end{tikzpicture}

\end{minipage}

\caption{Contours $\Ga, \Ga^{\prime}$ and $\msc{C}_{\bs{\be}_n}^{\prime}$.}
\label{contours Gamma GammaPrime et C BetaN Prime}
\end{center}
\end{figure}

\begin{lemme}
\label{Lemme engendrement des poles par pincement}
Let  $\Ga$ be a curve in $\Cx$ enjoying the properties as given in Figure~\ref{contours Gamma GammaPrime et C BetaN Prime}. Let
\beq
\mc{G}( \bs{\mu}_m \big) \, = \,  \f{ G(\bs{\mu}_{m}) \,  \psi(\bs{\mu}_m) }{ \pl{a=1}{m} P(\mu_a) } \;, 
\enq
be a symmetric function of $\bs{\mu}_m$ such that $G$ is holomorphic in some open neighbourhood of $\Ga$, 
\beq
\psi(\bs{\mu}_m) \, = \, \pl{a<b}{m} \Big\{ \sinh(\mu_{ab} )  \cdot \sinh\big[ \tfrac{ \pi }{ \xi } \mu_{ab}  \big] \Big\} 
\enq
while 
\beq
P(\mu) \; = \; (\mu-\be_n)(\mu-\be_1+\i\pi) \cdot (\mu-\be_1+2\i\pi)(\mu-\be_n+\i\pi) \cdot  (\mu-\be_1)(\mu-\be_n-\i\pi) \;. 
\label{definition polynome P engendrant les poles en be 1n egal a i pi}
\enq
Then, 
\beq
\mc{I}_{m}(\be_1,\be_n) \; = \; \Int{ \Ga^m }{}   \dd^m  \mu  \; \mc{G}( \bs{\mu}_m \big)  \;, 
\enq
seen as a meromorphic function of $\be_1$ admits simple poles at $\be_{1}=\be_{n}+\i\pi$ and one has 
\bem
\Res\Big( \mc{I}_{m}(\be_1,\be_n) \dd \be_1, \be_1=\be_n+\i\pi\Big) \; = \;  m \hspace{-3mm} \Int{  (\Ga^{\prime})^{m-1} }{} \hspace{-3mm} \dd^{m-1}\! \mu  \; \Bigg\{ 
\Res\Big(  2\i\pi \e{Res}\Big( \mc{G} ( \bs{\mu}_{m} ) \dd \mu_m \, , \, \mu_m=\be_1-\i\pi  \Big) \dd \be_1, \be_{1n}=\i\pi \Big)  \\
\hspace{-8mm} \, - \, \Res\Big( 2\i\pi \e{Res}\Big( \mc{G}( \bs{\mu}_{m}  ) \dd \mu_m \, , \, \mu_m=\be_1 \Big) \dd \be_1,  \be_{1n}=\i\pi \Big) 
\, - \, \Res\Big( 2\i\pi \e{Res}\Big( \mc{G} ( \bs{\mu}_{m} ) \dd \mu_m \, , \, \mu_m=\be_1 -2\i\pi \Big) \dd \be_1,   \be_{1n}=\i\pi \Big)  \Bigg\} \;. 
\label{ecriture residues de I beta1 betan}
\end{multline}
There, $\Ga^{\prime}$ is as depicted in Figure~\ref{contours Gamma GammaPrime et C BetaN Prime}.


\end{lemme}

\Proof

By using the symmetry of the integrand, it is easy to see that 
\bem
\mc{I}_{m}(\be_1,\be_n) = \sul{ \substack{ k+p_1+p_2+p_3  \\ = m} }{} \f{ m! }{ k! p_1! p_2! p_3! } \Int{ (\Ga^{\prime})^{k} }{} \hspace{-2mm} \dd^k \nu  
\hspace{-3mm} \Int{ (- \Dp{} \mc{D}_{\be_1-2\i\pi, \eps})^{p_1}  }{} \hspace{-5mm} \dd^{p_1} \mu^{(1)}
\hspace{-3mm} \Int{  (\Dp{} \mc{D}_{\be_1-\i\pi, \eps})^{p_2}  }{} \hspace{-4mm} \dd^{p_2} \mu^{(2)}
\hspace{-3mm} \Int{ (- \Dp{} \mc{D}_{\be_1, \eps})^{p_3}  }{} \hspace{-4mm} \dd^{p_3} \mu^{(3)}  \\
\times \f{ \big(\psi G  \big)\big(\bs{\nu}_k, \bs{\mu}^{(1)}_{p_1} , \bs{\mu}^{(2)}_{p_2} , \bs{\mu}^{(3)}_{p_3} \big)  }
{ \pl{a=1}{k} P(\nu_a)  \cdot \pl{s=1}{3} \pl{a=1}{p_s} P(\mu_a^{(s)})  } \;.
\end{multline}
Further, introduce for short
\beq
\wt{G}(\bs{\nu}_{m-k}, \bs{\mu}_{k}) \; = \; G(\bs{\nu}_{m-k}, \bs{\mu}_{k}) \psi(\bs{\nu}_{m-k}) 
\pl{a=1}{m-k} \pl{b=1}{k} \Big\{  \sinh(\nu_a-\mu_b )  \cdot \sinh\big[ \tfrac{ \pi }{ \xi } (\nu_a-\mu_{b})  \big]   \Big\}
\pl{a=1}{m-k} \f{ 1 }{ P(\nu_a) } \;. 
\enq
Then, the contour integrals can be evaluated explicitly, leading to 
\bem
\mc{I}_{m}(\be_1,\be_n) \; = \;  \sul{ \substack{ k+p_1+p_2+p_3  \\ = m} }{} \f{ m! }{ k! p_1! p_2! p_3! } \Int{  (\Ga^{\prime})^{k} }{} \hspace{-2mm} \dd^k \nu  
\; \wt{G}\Big(\bs{\nu}_k,  \big( (\be_1-2\i\pi)\ov{\bs{e}}_{p_1} , (\be_1-\i\pi)\ov{\bs{e}}_{p_2}, \be_1\ov{\bs{e}}_{p_3} \big) \Big) \\
\times \psi\Big((\be_1-2\i\pi)\ov{\bs{e}}_{p_1} , (\be_1-\i\pi)\ov{\bs{e}}_{p_2}, \be_1\ov{\bs{e}}_{p_3}  \Big)
\cdot \Big\{ -2\i\pi \e{Res}\Big( \tfrac{ \dd \mu }{ P(\mu) }\, , \, \mu=\be_1-2\i\pi  \Big) \Big\}^{p_1} \\
\times \Big\{ 2\i\pi \e{Res}\Big( \tfrac{ \dd \mu }{ P(\mu) }\, , \, \mu=\be_1-\i\pi  \Big) \Big\}^{p_2}
\cdot \Big\{ -2\i\pi \e{Res}\Big( \tfrac{ \dd \mu }{ P(\mu) }\, , \, \mu=\be_1  \Big) \Big\}^{p_3} \;. 
\end{multline}
Also, we used the notation $\ov{\bs{e}}_n=(1,\dots, 1)\in \R^n$. Since $\psi$ contains a term that is $\i\pi$ anti-periodic in each of its variables and vanishes 
when two of these coincide, the only non-vanishing contributions to the above sum are obtained, either, by taking $p_a=0$  for any $a$ or 
$p_a=1$ and  $p_b=0$ if $b\not=a$. However, only the latter three give rise to poles at $\be_{1n}=\i\pi$. This leads to \eqref{ecriture residues de I beta1 betan}. \qed  \\
We now apply this lemma to the calculation of the residue of interest. Recalling the definition \eqref{definition hat C}, one has 
\beq
-\i \Res\Big( \bs{\Psi}_{1,\dots,n}\big[ p\big]\big( \bs{\be}_n \big) \dd \be_1, \be_{1n}=\i\pi\Big) 
\; = \; -\i m \Int{ (\msc{C}_{\bs{\be}_n}^{\prime})^m}{} \hspace{-1mm} \dd^{m-1} u  \sul{a=1}{3} \bs{\Phi}_{1,\dots,n}^{(a)}\big( \bs{\be}_{n}^{\prime}; \bs{u}_{m-1} \big)
\enq
where 
\beqa
\bs{\Phi}_{1,\dots,n}^{(1)}\big( \bs{\be}_{n}^{\prime}; \bs{u}_{m-1} \big) & = & 
\Res\Big(  2\i\pi \e{Res}\Big( \wh{\Cx}_{1,\dots,n}^{(\ga)}\big( \bs{\be}_n \, ; \, \bs{u}_m\big) \cdot  \dd u_m \, , \, u_m=\be_1-\i\pi  \Big) \cdot  \dd \be_1, \be_{1n}=\i\pi \Big)  \vspace{2mm} \; ,  \\
\bs{\Phi}_{1,\dots,n}^{(2)}\big( \bs{\be}_{n}^{\prime}; \bs{u}_{m-1} \big) & = & 
\, - \, \Res\Big( 2\i\pi \e{Res}\Big(\wh{\Cx}_{1,\dots,n}^{(\ga)}\big( \bs{\be}_n \, ; \, \bs{u}_m\big) \cdot  \dd u_m \, , \, u_m=\be_1 \Big) \cdot  \dd \be_1,  \be_{1n}=\i\pi \Big) \vspace{2mm} \; , \\
\bs{\Phi}_{1,\dots,n}^{(3)}\big( \bs{\be}_{n}^{\prime}; \bs{u}_{m-1} \big) & = & 
\, - \, \Res\Big( 2\i\pi \e{Res}\Big(\wh{\Cx}_{1,\dots,n}^{(\ga)}\big( \bs{\be}_n \, ; \, \bs{u}_m\big) \cdot \dd u_m \, , \, u_m=\be_1 -2\i\pi \Big) \cdot  \dd \be_1,   \be_{1n}=\i\pi \Big) \;. 
\eeqa

\subsubsection{Calculation of $\bs{\Phi}_{1,\dots,n}^{(1)}$}

As follows from Lemma \ref{Lemme engendrement des poles par pincement}, the contribution to $\bs{\Phi}_{1,\dots,n}^{(1)}$ will stem from the 
pole at $u_{m}=\be_1-\i\pi$ which is generated by the $\wt{b}$ and $\wt{c}^{\, (\pm \ga)}$ coefficients present in the building blocks for the covector 
$\wt{\Cx}_{1,\dots,n}^{(\ga)}\big( \bs{\be}_n \, ; \, \bs{u}_m\big)$. \\
Since 
\beq
\e{Res}\Big(\wt{b}(\th) \dd \th, \th=\i\pi \Big) \, = \, -\f{\xi}{\pi} \sinh\Big( \i \f{ \pi^2 }{ \xi }  \Big) \qquad \e{and} \qquad
\e{Res}\Big(\wt{c}^{\, (\pm \ga)}(\th) \dd \th, \th=\i\pi \Big) \, = \,-\f{\xi}{\pi} \sinh\Big( \i \f{ \pi^2 }{ \xi }  \Big) \ex{\pm \i\pi \ga } \;,
\enq
then 
\beq
\e{Res}\Big(\wt{\op{S}}_{01}^{\, (\ga)}(\be_1-\th) \dd \th, \th=\be_1 + \i\pi \Big) \; = \; \f{\xi}{\pi} \sinh\Big( \i \f{ \pi^2 }{ \xi }  \Big) \;
\ex{ - \i \ga \f{\pi}{2} \sg_0^z} \, \op{C}_{10}^{ \op{t} } \cdot \op{C}_{10}\, \ex{ \i \ga \f{\pi}{2} \sg_0^z}
\label{ecriture residu de S tilde}
\enq
where 
\beq
\op{C}  \, = \, \Big\{ \op{v}^{-} \otimes \op{v}^{+} \,   +  \,  \op{v}^{+} \otimes \op{v}^{-}    \Big\}^{ \op{t} } \;. 
\enq
One has that 
\bem
\e{Res}\Big( \wt{\op{C}}^{\,(\ga)}_{1,\dots,n}\big( \bs{\be}_n \, ; \, u\big) \cdot  \dd u  \, , \, u =\be_1-\i\pi  \Big)_{\mid \be_{1n} = \i\pi }  \\
\; = \; \f{\xi}{\pi} \sinh\Big( \i \f{ \pi^2 }{ \xi }  \Big) \; 
\big( \op{v}_0^- \big)^{\op{t}} \ex{ - \i \ga \f{\pi}{2} \sg_0^z} \, \op{C}_{10}^{ \op{t} } \cdot \op{C}_{10} \cdot \ex{ \i \ga \f{\pi}{2} \sg_0^z} \cdot  \wt{\op{S}}_{20}^{\,(\ga)}(\be_{2n}) \cdots  \wt{\op{S}}_{n-10}^{\,(\ga)}(\be_{n-1 n}) \op{P}_{0n} \op{v}_0^+  \\
\; = \;   \f{\xi}{\pi} \sinh\Big( \i \f{ \pi^2 }{ \xi }  \Big) \;  \op{K}_{1n} \cdot \op{C}_{1n}  \cdot \ex{ \i \ga \f{\pi}{2} (\sg_n^z+1)} \cdot  \wt{\op{S}}_{2n}^{\,(\ga)}(\be_{2n}) \cdots  \wt{\op{S}}_{n-1n}^{\,(\ga)}(\be_{n-1 n})
\end{multline}
where 
\beq
\op{K}_{1n} \, = \,  \big( \op{v}_0^- \big)^{\op{t}}  \op{P}_{0n} \op{v}_0^+ \op{C}_{1n}^{ \op{t} } \, = \,  \big( \op{v}_0^- \big)^{\op{t}}  
\Big\{ \op{v}^{-}_0 \op{v}^{+}_1\op{v}^{+}_n \,   +  \,  \op{v}^{+}_0 \op{v}^{-}_1 \op{v}^{+}_n   \Big\} \; = \; \op{v}_1^+\op{v}_n^+ \; .
\enq
Further, one has that 
\bem
\wt{\op{C}}^{\,(\ga)}_{1,\dots,n}\big( \bs{\be}_n \, ; \, u \big) \cdot \op{K}_{1n} \, = \, 
\big( \op{v}_0^- \big)^{\op{t}} \cdot   \wt{\op{S}}_{1 0}^{\,(\ga)}(\be_{1}-u) \cdots  \wt{\op{S}}_{n 0}^{\,(\ga)}(\be_{n}-u)   \op{v}_0^+  \op{v}_1^+\op{v}_n^+ \\
\, = \, \op{v}_n^+ \big( \op{v}_0^- \big)^{\op{t}} \cdot  \wt{\op{S}}_{1 0}^{\,(\ga)}(\be_{1}-u) \op{v}_1^+  \wt{\op{S}}_{2 0}^{\,(\ga)}(\be_{2}-u) \cdots  \wt{\op{S}}_{n-1 0}^{\,(\ga)}(\be_{n-1}-u)  \op{v}_0^+  \;. 
\end{multline}
It is direct to check that $ \big( \op{v}_0^- \big)^{\op{t}} \cdot   \wt{\op{S}}_{1 0}^{\,(\ga)}(\be_{1}-u) \op{v}_1^+ =   \wt{b}(\be_1 - u) \op{v}_1^+ \big( \op{v}_0^- \big)^{\op{t}}$, what yields 
\beq
\wt{\op{C}}^{\,(\ga)}_{1,\dots,n}\big( \bs{\be}_n \, ; \, u\big) \cdot  \op{K}_{1n} \, = \,  \wt{b}(\be_1 - u) \,  \op{K}_{1n} \cdot  \wt{\op{C}}^{\,(\ga)}_{2,\dots, n-1}\big( \bs{\be}_{n-1}^{\prime} \, ; \, u\big) \;. 
\enq
All-in-all, this yields that 
\bem
\e{Res}\Big( \wt{\Cx}_{1,\dots,n}^{\,(\ga)}\big( \bs{\be}_n \, ; \, \bs{u}_m\big) \cdot  \dd u_m  \, , \, u_m =\be_1-\i\pi  \Big)_{\mid \be_{1n} = \i\pi }  \\ 
\; = \;  \f{\xi}{\pi} \sinh\Big( \i \f{ \pi^2 }{ \xi }  \Big) \cdot  \pl{k=1}{m-1}\wt{b}(\be_1-u_k) \cdot 
\Big( \op{C}_{1n}  \cdot \ex{ \i \ga \f{\pi}{2} (\sg_n^z+1)}  \Big)\otimes \wt{\Cx}_{2,\dots,n-1}^{\,(\ga)}\big( \bs{\be}_{n-1}^{\prime} \, ; \, \bs{u}_{m-1}\big)
\cdot  \wt{\op{S}}_{2n}^{\,(\ga)}(\be_{2n}) \cdots  \wt{\op{S}}_{n-1n}^{\,(\ga)}(\be_{n-1 n}) \;. 
\end{multline}
Thence, since
\bem
g(\bs{\be}_n ; \bs{u}_m) \, = \, g(\bs{\be}_{n-1}^{\prime} ; \bs{u}_{m-1}) \op{F}(\be_{1n}) \pl{b=2}{n-1} \Big\{ \op{F}(\be_{1b}) \op{F}(\be_{bn}) \Big\}
\pl{k=1}{m-1}\pl{s = 1,n}{} \big( a \phi\big)(\be_s-u_k) \\
\times \pl{s=1}{n} \big( a \phi\big)(\be_s-u_m) \cdot \pl{k=1}{m-1} \tau(u_{km})\;,
\label{ecriture decomposition fct g en poles en um et sans beta 1 et n}
\end{multline}
one has that 
\bem
\bs{\Phi}_{1,\dots,n}^{(1)}\big( \bs{\be}_{n}^{\prime}; \bs{u}_{m-1} \big) \; = \;  g(\bs{\be}_{n-1}^{\prime} ; \bs{u}_{m-1}) \cdot 
\Big( \op{C}_{1n}  \cdot \ex{ \i \ga \f{\pi}{2}  \sg_n^z }  \Big) \otimes \wt{\Cx}_{2,\dots,n-1}^{\,(\ga)}\big( \bs{\be}_{n-1}^{\prime} \, ; \, \bs{u}_{m-1}\big) 
\cdot  \wt{\op{S}}_{2n}^{\,(\ga)}(\be_{2n}) \cdots  \wt{\op{S}}_{n-1n}^{\,(\ga)}(\be_{n-1 n}) \\
\times \op{O}_{1,\dots,n}^{(\ga)}(\be_n+\i\pi,\bs{\be}_n^{\prime} ) \cdot  p\Big( \big( \be_n + \i\pi, \bs{\be}_n^{\prime} \big) ;  \big(\bs{u}_{m-1},\be_1-\i\pi \big) \Big)_{\mid \be_{1n} = \i\pi } \cdot \mf{R}^{(1)}
\end{multline}
in which  $\mf{R}^{(1)} \, = \, 2\i\pi \e{Res}\Big(  H\big(\bs{\be}_n ;   (\bs{u}_{m-1}, \be_1-\i\pi) \big) \cdot  \dd \be_1  \, , \,  \be_1 = \be_n + \i\pi  \Big)_{\mid \be_{1n} = \i\pi }$ where 
\bem
H\big(\bs{\be}_n ;   \bs{u}_{m}) \, = \, \ex{\i \ga \f{\pi}{2} }  \f{\xi}{\pi} \sinh\Big( \i \f{ \pi^2 }{ \xi }  \Big) \, \op{F}(\be_{1n}) \cdot  \pl{k=1}{m-1}\wt{b}(\be_1-u_k) \cdot \pl{b=2}{n-1} \Big\{ \op{F}(\be_{1b}) \op{F}(\be_{bn}) \Big\}
\cdot \pl{k=1}{m-1}\pl{s = 1,n}{} \big( a \phi\big)(\be_s-u_k) \\
\times \pl{s=1}{n} \big( a \phi\big)(\be_s-u_m) \cdot \pl{k=1}{m-1} \tau(u_{km}) \;. 
\end{multline}
This yields
\bem
\mf{R}^{(1)} \, = \, 2\i\pi \ex{\i \f{\pi \ga }{2} } \big( a \phi \op{F}\big)(\i \pi) \f{\xi}{\pi} \sinh\Big( \i \f{ \pi^2 }{ \xi }  \Big) \cdot \e{Res}\Big(  \big( a \phi\big)(\be_{n1}+\i\pi)   \cdot  \dd \be_1  \, , \,  \be_1 = \be_n + \i\pi  \Big)
\cdot \pl{b=2}{n-1} \Big\{ \op{F}(\be_{nb}+\i\pi) \op{F}(\be_{bn})   (a \phi)(\be_{bn})  \Big\} \\
\times \pl{k=1}{m-1}\Big\{ \wt{b}(\be_n+\i\pi-u_k) \cdot \tau(u_{k}-\be_{n}) \cdot \big( a \phi\big)(\be_n-u_k) \cdot \big( a \phi\big)(\be_n+\i\pi-u_k) \Big\} \;. 
\end{multline}
Now, the definition of $\phi$ implies
\beq
\op{F}(\i\pi-\la) \op{F}(\la )     \phi( \la )  \; = \; 1 
\label{ecriture simplification produit Fs et phi}
\enq
while direct calculations lead to 
\beq
\big( a \phi\big)(\la)\, \big( a \phi\big)(\la+\i\pi) \; = \; \Bigg\{  \f{ 1 }{ \op{F}(\i\tfrac{\pi}{2}) \varpi\big( \i \tfrac{\xi-\pi}{2} \big) } \Bigg\}^{4}
\f{  - 1 }{ \sinh(\la) \sinh\big[\tfrac{\pi}{\xi}(\la+\i\pi) \big]}  \;. 
\enq
In its turn, this entails 
\beq
\big( a \phi\big)(\la)\, \big( a \phi\big)(\la+\i\pi) \, \wt{b}(\la+\i\pi)\,  \tau(-\la) \; = \; 1 \;. 
\label{ecriture reduction produit aphi aphi et b tilde i pi shift et tau}
\enq
Thus, all-in-all, one gets that 
\beq
\mf{R}^{(1)} \, = \,  \pl{k=2}{n-1}a(\be_{kn} ) \cdot \f{2 \ex{\i \ga \f{\pi}{2} } \xi \sinh\Big( \i \tfrac{ \pi^2 }{ \xi }  \Big)  }{ \Big\{ \op{F}(\i\tfrac{\pi}{2}) \varpi\big( \i \tfrac{\xi-\pi}{2} \big) \Big\}^{4} } 
\cdot \varpi\left( \ba{c} \i \tf{\xi}{2}  \\ 2\i\pi -\i \tf{\xi}{2}  \ea \right) \; .
\enq
Upon using that 
\beq
\f{ \i }{ \sinh\big( \i \tfrac{ \pi^2 }{ \xi }  \big) } \, = \, \varpi\left( \ba{c} \i \tf{\xi}{2}  \\ 2\i\pi -\i \tf{\xi}{2}  \ea \right) \; ,
\enq
the above reduces to
\beq
\mf{R}^{(1)} \, = \,   \varkappa \cdot \pl{k=2}{n-1}a(\be_{kn} ) \;.  
\enq
Here, I remind that $\varkappa$ has been introduced in \eqref{definition cste varkappa}. In its turn, owing to \eqref{ecriture equation echange operateur spin twist et matrice S}
and after absorbing the $a$ prefactors into the definition of $\op{S}_{kn}(\be_{kn})$, this entails that 
\bem
\bs{\Phi}_{1,\dots,n}^{(1)}\big( \bs{\be}_n^{\prime}; \bs{u}_{m-1} \big) \; = \;  \varkappa \cdot  g(\bs{\be}_{n-1}^{\prime} ; \bs{u}_{m-1}) \cdot 
\Big( \op{C}_{1n}  \cdot \ex{ \i \ga \f{\pi}{2}  \sg_n^z }  \Big) \otimes \wt{\Cx}_{2,\dots,n-1}^{\,(\ga)}\big( \bs{\be}_{n-1}^{\prime} \, ; \, \bs{u}_{m-1}\big) \cdot  \op{O}_{1,\dots,n}^{(\ga)}(\be_n+\i\pi,\bs{\be}_n^{\prime} )  \\
\times  \op{S}_{2n}(\be_{2n}) \cdots  \op{S}_{n-1n}(\be_{n-1 n}) \cdot p\Big( \big( \be_n + \i\pi, \bs{\be}_n^{\prime} \big) ;  \big(\bs{u}_{m-1},\be_n \big) \Big) \;. 
\end{multline}
At this stage, one observes that 
\beq
\op{O}_{1,\dots,n}^{(\ga)}(\be_n+\i\pi,\bs{\be}_n^{\prime} )  \; = \;  \ex{ \f{\ga}{2} \be_n (\sg_n^z+\sg_1^z) + \i \ga \f{\pi}{2} \sg_1^z} \op{O}_{2,\dots,n-1}^{(\ga)}(\bs{\be}_{n-1}^{\prime} )
\label{ecriture reduction sous residus operateur O du twist}
\enq
and that 
\beq
\op{C}_{1n}  \cdot \ex{\a (\sg_n^z+\sg_1^z) } \; = \; \op{C}_{1n}  \;. 
\label{identite vecteur vp1 pour C1n}
\enq
The use of the reduction equation $\mathrm{e)}$ satisfied by $p$ then leads to 
\bem
-\i m \Int{ (\msc{C}_{\bs{\be}_n}^{\prime})^{m-1} }{} \hspace{-1mm} \dd^{m-1} u \, \bs{\Phi}_{1,\dots,n}^{(1)}\big( \bs{\be}_n^{\prime}; \bs{u}_{m-1} \big) \\
\; = \;\Int{ (\msc{C}_{\bs{\be}_n}^{\prime})^{m-1} }{} \hspace{-1mm} \dd^{m-1} u \, \op{C}_{1n} \otimes \wt{\Cx}_{2,\dots,n-1}^{\,(\ga)}\big( \bs{\be}_{n-1}^{\prime} \, ; \, \bs{u}_{m-1}\big)
\cdot  \op{O}_{2,\dots,n-1}^{(\ga)}( \bs{\be}_{n-1}^{\prime} ) p\big(  \bs{\be}^{\prime}_{n-1} ;  \bs{u}_{m-1} \big)  \cdot \op{S}_{2n}(\be_{2n}) \cdots  \op{S}_{n-1n}(\be_{n-1 n}) \\
\, -\, \i \, m \, \de_{n,2m} \, \vp^{(+)}\big( \bs{\be}_{n-1}^{\prime} \big)  \varkappa\Int{ (\msc{C}_{\bs{\be}_n}^{\prime})^{m-1} }{} \hspace{-1mm} \dd^{m-1} u \,
\op{C}_{1n} \otimes \wt{\Cx}_{2,\dots,n-1}^{\,(\ga)}\big( \bs{\be}_{n-1}^{\prime} \, ; \, \bs{u}_{m-1}\big)
\cdot  \op{O}_{2,\dots,n-1}^{(\ga)}( \bs{\be}_{n-1}^{\prime} ) \op{S}_{2n}(\be_{2n}) \cdots  \op{S}_{n-1n}(\be_{n-1 n})  \\
\; = \; \op{C}_{1n} \otimes
\bs{\Psi}_{2,\dots, n-1}[p]\big( \bs{\be}_{n-1}^{\prime} \big)  \cdot \op{S}_{2n}(\be_{2n}) \cdots  \op{S}_{n-1n}(\be_{n-1 n}) \\
\, -\, \i \, m \, \de_{n,2m} \,  \varkappa \cdot  \op{C}_{1n} \otimes \bs{\Psi}_{2,\dots, n-1}[1]\big( \bs{\be}_{n-1}^{\prime} \big) \vp^{(+)}\big( \bs{\be}_{n-1}^{\prime} \big)
\cdot \op{S}_{2n}(\be_{2n}) \cdots  \op{S}_{n-1n}(\be_{n-1 n}) \;. 
\end{multline}
Note that, in order to get to the last equality, we have used that the integrands do not depend on $\be_{n}$ and $\be_1$, so that one may easily deform the contour $\msc{C}_{\bs{\be}_n}^{\prime}$
to $\msc{C}_{\bs{\be}_{n-1}^{\prime}}$, what then produces the last equality. Note that the last line of this equation vanishes due to Lemma \ref{Lemme annulation K transformee}.

\subsubsection{Calculation of $\bs{\Phi}_{1,\dots,n}^{(2)}$}

In the case of this contribution, the pole at $u_m=\be_1$ issues from the $a\phi$ factor. Upon denoting
\beq
C^{(\a)}_{a\phi} \; = \;\e{Res}\Big( a(\th) \phi(\th) \dd \th, \th=\a \Big)
\label{definition residu de a phi}
\enq
and using the decomposition \eqref{ecriture decomposition fct g en poles en um et sans beta 1 et n}, one infers that 
\bem
\e{Res}\Big( g(\bs{\be}_n,\bs{u}_m) \dd u_m, u_m = \be_1 \Big) \, = \, - C^{(0)}_{a\phi}  g(\bs{\be}_{n-1}^{\prime} ; \bs{u}_{m-1}) \op{F}(\be_{1n}) 
\pl{b=2}{n-1} \Big\{ \op{F}(\be_{1b}) \op{F}(\be_{bn}) \Big\}
\\
\times \pl{s=2}{n} \big( a \phi\big)(\be_{s1}) \cdot \pl{k=1}{m-1} \Big\{ \tau(u_{k}-\be_1) \, \pl{s = 1,n}{} \big( a \phi\big)(\be_s-u_k)  \Big\}\;. 
\end{multline}
Moreover, one has 
\bem
\Om_{n} \wt{\op{C}}_{1,\dots,n}^{\,(\ga)}(\bs{\be}_n ;  \be_1) \; = \; \Big(\op{v}^+_1\dots \op{v}^{+}_n\Big)^{\op{t}} 
\big( \op{v}^-_0 \big)^{\op{t}} \,  \op{P}_{01}\wt{\op{S}}_{20}^{\,(\ga)}(\be_{21})\cdots \wt{\op{S}}_{n0}^{\,(\ga)}(\be_{n1}) \op{v}^+_0  \\
\; = \;  \Big(\op{v}^+_0 \, \op{v}^-_1 \, \op{v}^+_2\dots \op{v}^{+}_n\Big)^{\op{t}} \, \wt{\op{S}}_{20}^{\,(\ga)}(\be_{21})\cdots \wt{\op{S}}_{n0}^{\,(\ga)}(\be_{n1}) \op{v}^+_0\; = \; 
\Big(  \op{v}^-_1 \, \op{v}^+_2\dots \op{v}^{+}_n\Big)^{\op{t}} \;. 
\end{multline}
Therefore, 
\beq
\wt{\Cx}_{1,\dots,n}^{\,(\ga)}\Big( \bs{\be}_{n} \, ; \, \big( \bs{u}_{m-1}, \be_1 \big) \Big) \; = \; 
\Big(  \op{v}^-_1 \, \op{v}^+_2\dots \op{v}^{+}_n\Big)^{\op{t}} \cdot \wt{\op{C}}_{1,\dots,n}^{\,(\ga)}(\bs{\be}_n,u_1) \cdots \wt{\op{C}}_{1,\dots,n}^{\,(\ga)}(\bs{\be}_n ; u_{m-1}) \;. 
\enq
Since the pole at $\be_{1}=\be_n + \i\pi$ issues from the pole at $\th=-\i\pi$ of $a\phi$, one gets that  
\bem
\bs{\Phi}_{1,\dots,n}^{(2)}\big( \bs{\be}_n^{\prime} ; \bs{u}_{m-1} \big) \; = \; 2\i\pi \, \Big( C^{(0)}_{a\phi}  \Big)^2 \, \op{F}(\i\pi) \, 
g(\bs{\be}_{n-1}^{\prime} ; \bs{u}_{m-1}) \cdot \Big(  \op{v}^-_1 \, \op{v}^+_2\dots \op{v}^{+}_n\Big)^{\op{t}} \cdot \pl{p=1}{m-1}\wt{\op{C}}_{1,\dots,n}^{\,(\ga)}(\bs{\be}_n ; u_p)   
\cdot  \op{O}_{1,\dots,n}^{(\ga)}(\bs{\be}_n) \\
\times \pl{s=2}{n-1} \Big\{ \op{F}(\be_{1s}) \op{F}(\be_{sn}) \big( a \phi\big)(\be_{s1})  \Big\} \cdot 
\pl{k=1}{m-1} \Big\{ \tau(u_{k}-\be_1) \, \big( a \phi\big)(\be_1-u_k) \, \big( a \phi\big)(\be_1 - u_k - \i \pi )  \Big\}
\cdot p\Big( \bs{\be}_{n}, \big(\bs{u}_{m-1} ; \be_1\big) \Big)_{\mid \be_{1n}=\i\pi} \;. 
\nonumber
\end{multline}
Note that above, we have used that $C^{(-\i\pi)}_{a\phi} \, = \, - C^{(0)}_{a\phi}$. Moreover, since $\op{F}(\i\pi)=1$, 
\beq
\op{F}(\be_{1s}) \op{F}(\be_{sn}) \big( a \phi\big)(\be_{s1})_{\mid \be_{1n}=\i\pi} \, = \, \op{F}(\be_{1s}) \op{F}(\be_{s1}+\i\pi) \big( a \phi\big)(\be_{s1}) \; ,
\label{identite reduction produits de F et phi}
\enq
this product reduces to $1$ by virtue of \eqref{ecriture simplification produit Fs et phi}. Further, owing to \eqref{ecriture shift par 2 i pi de a phi}, one has that
\beq
\tau(-\la ) \, \big( a \phi\big)(\la) \, \big( a \phi\big)(\la - \i \pi ) \; = \; 
\tau(-\la ) \, \big( a \phi\big)(\la) \, \big( a \phi\big)(\la + \i \pi ) \wt{b}\big( \la +  \i\pi \big) \, = \, 1
\label{ecriture reduction produit tau et al}
\enq
as follows from \eqref{ecriture reduction produit aphi aphi et b tilde i pi shift et tau}. Finally, a direct calculation yields that 
\beq
2\i\pi \, \Big( C^{(0)}_{a\phi}  \Big)^2 \; = \; - \varkappa  \ex{-\i\f{\pi}{2} \ga }\;, 
\label{ecriture reduction cste C0 a phi vers varkappa}
\enq
where $\varkappa$ has been introduced in \eqref{definition cste varkappa}. This yields 
\bem
\bs{\Phi}_{1,\dots,n}^{(2)}\big( \bs{\be}_n^{\prime} ; \bs{u}_{m-1} \big) \; = \;  - \varkappa \ex{-\i\f{\pi}{2} \ga } \, g(\bs{\be}_{n-1}^{\prime} ; \bs{u}_{m-1}) 
\cdot \Big(  \op{v}^-_1 \, \op{v}^+_2\dots \op{v}^{+}_n\Big)^{\op{t}}  \\
\times \pl{p=1}{m-1}\wt{\op{C}}_{1,\dots,n}^{\,(\ga)}(\bs{\be}_n,u_p) \cdot  \op{O}_{1,\dots,n}^{(\ga)}(\bs{\be}_n) 
\cdot  p\Big( \bs{\be}_{n} ;  \big(\bs{u}_{m-1},\be_1\big) \Big)_{\mid \be_{1n}=\i\pi} \;. 
\end{multline}

\subsubsection{Calculation of $\bs{\Phi}_{1,\dots,n}^{(3)}$}

In that case, the pole at $u_m=\be_1-2\i\pi$ stems from the factor $a\phi$ present in $g$. Following the decomposition  
\eqref{ecriture decomposition fct g en poles en um et sans beta 1 et n}, one infers that 
\bem
\e{Res}\Big( g(\bs{\be}_n ; \bs{u}_m) \dd u_m, u_m = \be_1 -2\i\pi \Big) \, = \, - C^{(2\i\pi)}_{a\phi} \, g(\bs{\be}_{n-1}^{\prime} ; \bs{u}_{m-1}) \, \op{F}(\be_{1n})  \, 
\pl{b=2}{n-1} \Big\{ \op{F}(\be_{1b}) \op{F}(\be_{bn}) \Big\}
\\
\times \pl{s=2}{n} \big( a \phi\big)(\be_{s1}+2\i\pi) \cdot \pl{k=1}{m-1} \Big\{ \tau(u_{k}-\be_1+2\i\pi) \, \pl{s = 1,n}{} \big( a \phi\big)(\be_s-u_k)  \Big\}\;. 
\end{multline}
The remaining factor do not exhibit poles at $\be_1 = \be_n + \i \pi$ and the latter issues from the covector part of the integrand. 
Owing to \eqref{ecriture residu de S tilde}, one gets 
\bem
\e{Res}\Big( \Om_{n} \wt{\op{C}}_{1,\dots,n}^{\, (\ga)}(\bs{\be}_n ; \be_1-2\i\pi) \dd \be_1, \be_{1n}=\i\pi \Big) \\
\, = \, \f{\xi}{\pi} \sinh\Big( \i \f{ \pi^2 }{ \xi }\Big) \cdot 
\Om_{n} \cdot \big( \op{v}^-_0 \big)^{\op{t}} \, \wt{\op{S}}_{01}^{\, (\ga)}(2\i\pi) \wt{\op{S}}_{20}^{\, (\ga)}(\be_{21}+2\i\pi)\cdots \wt{\op{S}}_{n-10}^{\, (\ga)}(\be_{n-11}+2\i\pi)
\cdot \ex{-\i \f{\pi}{2} \ga \sg_0^z } \big( \op{v}^+_0 \op{v}_n^{-} \, + \, \op{v}^+_n \op{v}_0^{-} \big) \cdot \big( \op{v}_n^{-} \big)^{\op{t}}  \ex{\i \f{\pi}{2} \ga } \\
\, = \, \ex{\i   \pi  \ga }  \f{\xi}{\pi} \sinh\Big( \i \f{ \pi^2 }{ \xi }\Big) \big( \op{v}_1^{+}\cdots \op{v}_{n-1}^{+} \op{v}_n^{-} \big)^{\op{t}}
\underbrace{ \big( \op{v}^-_0 \big)^{\op{t}} \, \wt{\op{S}}_{01}^{\, (\ga)}(2\i\pi) \wt{\op{S}}_{20}^{\, (\ga)}(\be_{21}+2\i\pi)\cdots \wt{\op{S}}_{n-10}^{\, (\ga)}(\be_{n-11}+2\i\pi)  \op{v}_0^{-} }_{ = \op{D}_{1,\dots, n-1}^{\, (\ga)}\big( \bs{\be}_{n-1};\be_1 - 2 \i \pi \big) } \\
\, = \, \ex{\i   \pi  \ga }  \f{\xi}{\pi} \sinh\Big( \i \f{ \pi^2 }{ \xi }\Big) \pl{k=1}{n-1} \wt{b}\big( \be_{k1}+2\i\pi)
\cdot  \big( \op{v}_1^{+}\cdots \op{v}_{n-1}^{+} \op{v}_n^{-} \big)^{\op{t}} \;. 
\end{multline}
All together one arrives to  
\bem
\bs{\Phi}_{1,\dots,n}^{(3)}\big( \bs{\be}_n^{\prime} ; \bs{u}_{m-1} \big) \; = \; 2\i\pi \, \ex{\i \pi  \ga }   \,   \f{\xi}{\pi} \sinh\Big( \i \f{ \pi^2 }{ \xi }\Big) \cdot \wt{b}(2\i\pi) \, C^{(2\i\pi)}_{a\phi}    \cdot \big( a \phi \op{F}\big)(\i\pi) \, 
\cdot g(\bs{\be}_{n-1}^{\prime} ; \bs{u}_{m-1})  \\ 
\times  \Big(  \op{v}^+_1 \dots  \op{v}^+_{n-1} \,  \op{v}^{-}_n\Big)^{\op{t}} \cdot \pl{p=1}{m-1}\wt{\op{C}}_{1,\dots,n}^{\, (\ga)}(\bs{\be}_n ; u_p) \cdot  \op{O}_{1,\dots,n}^{(\ga)}(\bs{\be}_n)
\cdot \pl{s=2}{n-1} \Big\{ \op{F}(\be_{1s}) \op{F}(\be_{sn}) \, \big( a \phi\big)(\be_{s1} + 2\i\pi)   \,   \wt{b}(\be_{s1}+2\i\pi)  \Big\} \\
\times  \pl{k=1}{m-1} \Big\{ \tau(u_{k}-\be_1+2\i\pi) \, \big( a \phi\big)(\be_1-u_k) \, \big( a \phi\big)(\be_1 - u_k - \i \pi )  \Big\}
\cdot p\Big( \bs{\be}_{n} ;  \big(\bs{u}_{m-1},\be_1-2\i\pi\big) \Big)_{\mid \be_{1n}=\i\pi} \;. 
\end{multline}
By using \eqref{ecriture shift par 2 i pi de a phi} but also the identities
\beq
C^{(0)}_{a\phi} \, = \, \wt{b}(2\i\pi) \, C^{(2\i\pi)}_{a\phi}    \qquad \e{and} \qquad 
\f{\xi}{\pi} \sinh\Big( \i \f{ \pi^2 }{ \xi }\Big) \cdot \big( a \phi \big)(\i\pi) \, = \, C^{(0)}_{a\phi} \;, 
\label{ecriture forme alternative C0 a phi}
\enq
one arrives to 
\bem
\bs{\Phi}_{1,\dots,n}^{(3)}\big( \bs{\be}_n^{\prime} ; \bs{u}_{m-1} \big) \; = \; 2\i\pi \, \ex{\i \pi  \ga }   \, \Big( C^{(0)}_{a\phi} \Big)^2    \cdot   \op{F}(\i\pi) \, 
\cdot g(\bs{\be}_{n-1}^{\prime} ; \bs{u}_{m-1})  \cdot \pl{s=2}{n-1} \Big\{ \op{F}(\be_{1s}) \op{F}(\be_{sn}) \, \big( a \phi\big)(\be_{s1} )    \Big\} \\ 
\times  \pl{k=1}{m-1} \Big\{ \tau(u_{k}-\be_1) \, \big( a \phi\big)(\be_1-u_k) \, \big( a \phi\big)(\be_1 - u_k - \i \pi )  \Big\}
\cdot \pl{k=1}{m-1}  \bigg\{ \f{ \sinh\big[ \tfrac{\pi}{\xi} (u_k-\be_1+2\i\pi) \big] }{ \sinh\big[ \tfrac{\pi}{\xi} (u_k-\be_1) \big] } \bigg\}
\\
\times \Big(  \op{v}^+_1 \dots  \op{v}^+_{n-1} \, \op{v}^{-}_n\Big)^{\op{t}} \cdot \pl{p=1}{m-1}\wt{\op{C}}_{1,\dots,n}^{\,(\ga)}(\bs{\be}_n ; u_p)  \cdot  \op{O}_{1,\dots,n}^{(\ga)}(\bs{\be}_n)
\cdot p\Big( \bs{\be}_{n} ;  \big(\bs{u}_{m-1},\be_1-2\i\pi\big) \Big)_{\mid \be_{1n}=\i\pi} \;. 
\end{multline}
At this stage, upon employing \eqref{identite reduction produits de F et phi}, \eqref{ecriture reduction produit tau et al}, \eqref{ecriture reduction cste C0 a phi vers varkappa} property $\mathrm{b)}$, 
\bem
\bs{\Phi}_{1,\dots,n}^{(3)}\big( \bs{\be}_n^{\prime} ; \bs{u}_{m-1} \big) \; = \; -\varkappa \ex{\i \f{\pi}{2} \ga } \cdot g(\bs{\be}_{n-1}^{\prime} ; \bs{u}_{m-1}) 
\cdot \pl{k=1}{m-1}  \Bigg\{ \f{ \sinh\big[ \tfrac{\pi}{\xi} (u_k-\be_1+2\i\pi) \big] }{ \sinh\big[ \tfrac{\pi}{\xi} (u_k-\be_1) \big] } \Bigg\} \\
\times \Big(  \op{v}^+_1 \dots  \op{v}^+_{n-1} \, \op{v}^{-}_n\Big)^{\op{t}}\cdot \pl{p=1}{m-1}\wt{\op{C}}_{1,\dots,n}^{\,(\ga)}(\bs{\be}_n ; u_p)  \cdot  \op{O}_{1,\dots,n}^{(\ga)}(\bs{\be}_n)
\cdot p\Big( \bs{\be}_{n} ;  \big(\bs{u}_{m-1},\be_1 \big) \Big)_{\mid \be_{1n}=\i\pi} \;. 
\end{multline}

\subsubsection{Reduction in $\bs{\Phi}_{1,\dots,n}^{(2)}+ \bs{\Phi}_{1,\dots,n}^{(3)}$}

Upon putting together the above two results, one gets that 

\beq
\bs{\Phi}_{1,\dots,n}^{(2)}\big( \bs{\be}_n^{\prime} ; \bs{u}_{m-1} \big) \; + \; \bs{\Phi}_{1,\dots,n}^{(3)}\big( \bs{\be}_n^{\prime} ; \bs{u}_{m-1} \big)  \; = \; 
-\varkappa   \,  g(\bs{\be}_{n-1}^{\prime} ; \bs{u}_{m-1})  
\cdot \mc{G}\big( \bs{\be}_n^{\prime} ; \bs{u}_{m-1} \big)\cdot  \op{O}_{1,\dots,n}^{(\ga)}(\bs{\be}_n) \cdot p\Big( \bs{\be}_{n} ;  \big(\bs{u}_{m-1},\be_1\big) \Big)_{\mid \be_{1n}=\i\pi} \;,
\enq
where 
\bem
\mc{G}\big( \bs{\be}_n^{\prime} ; \bs{u}_{m-1} \big) \, = \, \ex{\i \f{\pi}{2} \ga } \pl{k=1}{m-1}  \bigg\{ \f{ \sinh\big[ \tfrac{\pi}{\xi} (u_k-\be_1+2\i\pi) \big] }{ \sinh\big[ \tfrac{\pi}{\xi} (u_k-\be_1) \big] } \bigg\}
\cdot \Big(  \op{v}^+_1  \dots \op{v}^+_{n-1} \, \op{v}^{-}_n\Big)^{\op{t}} \cdot \pl{p=1}{m-1}\wt{\op{C}}_{1,\dots,n}^{\,(\ga)}(\bs{\be}_n ; u_p)    \\
\, + \, \ex{- \i \f{\pi}{2} \ga } \, \Big(  \op{v}^-_1 \, \op{v}^+_2\dots \op{v}^{+}_n\Big)^{\op{t}} \cdot  \pl{p=1}{m-1}\wt{\op{C}}_{1,\dots,n}^{\,(\ga)}F(\bs{\be}_n ; u_p)  \;. 
\end{multline}
At this stage, one employs \eqref{ecriture action produit C sur vecteur + ... +-}-\eqref{ecriture action produit C sur vecteur -+ ... +}, uses the permutation symmetry of the $\bs{u}_{m-1}$
integral  so as to get that 
\bem
-\i m \Int{ \msc{C}_{\bs{\be}_n}^{\prime} }{} \hspace{-1mm} \dd^{m-1} u  \sul{a=2}{3} \bs{\Phi}_{1,\dots,n}^{(a)}\big( \bs{\be}_{n}^{\prime}; \bs{u}_{m-1} \big) \; = \;
\i   m \, \varkappa  \Int{ (\msc{C}_{\bs{\be}_n}^{\prime})^{m-1} }{} \hspace{-1mm} \dd^{m-1}   u \,   g(\bs{\be}_{n-1}^{\prime} ; \bs{u}_{m-1})
\, \big(  \op{v}^+_2 \dots \op{v}^{+}_{n-1}\big)^{\op{t}} \cdot \pl{p=1}{m-1}\wt{\op{C}}_{2,\dots,n-1}^{\,(\ga)}(\bs{\be}_{n-1}^{\prime} ; u_p) \\
\times  \Bigg\{ \ex{\i \f{\pi}{2} \ga } \, 
\pl{k=1}{m-1}  \bigg\{ \f{ \sinh\big[ \tfrac{\pi}{\xi} (u_k-\be_1+2\i\pi) \big] }{ \sinh\big[ \tfrac{\pi}{\xi} (u_k-\be_1) \big] } \cdot \wt{b}(\be_1-u_s) \, \wt{b}(\be_n-u_s)  \bigg\} 
\, \big( \op{v}^+_1 \,\op{v}^{-}_n\big)^{\op{t}} \; + \; \ex{-\i \f{\pi}{2} \ga }
\big( \op{v}^-_1 \,\op{v}^{+}_n\big)^{\op{t}} \Bigg\}  \\
\times  \op{O}_{1,\dots,n}^{(\ga)}(\bs{\be}_n) \cdot p\Big( \bs{\be}_{n} ;  \big(\bs{u}_{m-1},\be_1 \big) \Big)_{\mid \be_{1n}=\i\pi} \; +\; \i m (m-1) \varkappa \msc{W}\big( \bs{\be}_{n-1} \big)
\label{ecriture integrale globale contributions termes 2 et 3 aux cas iii}
\end{multline}
where, above, we have introduced 
\bem
\msc{W}\big( \bs{\be}_{n-1} \big) \; = \; \Int{ (\msc{C}_{\bs{\be}_n}^{\prime})^{m-1}  }{}   \hspace{-1mm} \dd^{m-1}   u \,   g(\bs{\be}_{n-1}^{\prime} ; \bs{u}_{m-1})
\big(  \op{v}^-_1 \, \op{v}^+_2 \dots \op{v}^{+}_{n-1}\,\op{v}^{-}_n \big)^{\op{t}} \cdot \mc{W}\big(\bs{\be}_{n-1} ; \bs{u}_{m-1}\big) \\
\times \pl{p=1}{m-2}\wt{\op{C}}_{2,\dots,n-1}^{\,(\ga)}(\bs{\be}_{n-1}^{\prime} ; u_p)  \cdot  \op{O}_{1,\dots,n}^{(\ga)}(\bs{\be}_n)  \cdot  
p\Big( \bs{\be}_{n} ;  \big(\bs{u}_{m-1},\be_1 \big) \Big)_{\mid \be_{1n}=\i\pi} \;.
\end{multline}
in which, upon agreeing that $\be_n=\be_1-\i\pi$,  
\bem
\mc{W}\big(\bs{\be}_{n-1} ; \bs{u}_{m-1}\big) \; = \; \ex{\i \f{\pi}{2} \ga } \, \wt{c}^{\,(\ga)}(\be_1-u_{m-1}) \, \cdot  \,
\pl{k=1}{m-1}  \bigg\{ \f{ \sinh\big[ \tfrac{\pi}{\xi} (u_k-\be_1+2\i\pi) \big] }{ \sinh\big[ \tfrac{\pi}{\xi} (u_k-\be_1) \big] }
\, \cdot \wt{b}(\be_n-u_k) \bigg\} \cdot \pl{  s=1  }{ m-2 } \f{ \wt{b}(\be_1-u_s) }{ \wt{b}(u_{m-1} - u_s) }  \\ 
\; + \; \ex{-\i \f{\pi}{2} \ga } \, 
\wt{c}^{\,(\ga)}(\be_n-u_{m-1}) \cdot  \f{   \pl{p=2}{n-1} \wt{b}(\be_p-u_{m-1}) }{ \pl{  s=1   }{ m-2 } \wt{b}(u_{s} - u_{m-1}) } \;. 
\end{multline}
We first focus on the first term in \eqref{ecriture integrale globale contributions termes 2 et 3 aux cas iii}. 
Here, a direct calculation shows that the product in factor of $\big( \op{v}^+_1 \,\op{v}^{-}_n\big)^{\op{t}}$ simply reduces to $1$. 
Then, upon applying the second reduction equation given in $\mathrm{e)}$ and the identities \eqref{ecriture reduction sous residus operateur O du twist}-\eqref{identite vecteur vp1 pour C1n}, one arrives to 
\bem
-\i m \Int{ (\msc{C}_{\bs{\be}_n})^{m-1} }{} \hspace{-1mm} \dd^{m-1} u  \sul{a=2}{3} \bs{\Phi}_{1,\dots,n}^{(a)}\big( \bs{\be}_{n}^{\prime}; \bs{u}_{m-1} \big) \; = \; 
-  \op{C}_{1n} \otimes \bs{\Psi}_{2,\dots,n-1}[p](\bs{\be}_{n-1}^{\prime} ) \cdot \ex{2\i\pi \om +\i\pi \ga \sg_1^z } \\
\, + \, \i m  \de_{n,2m}\,  \op{C}_{1n} \otimes \bs{\Psi}_{2,\dots,n-1}[1](\bs{\be}_{n-1}^{\prime} ) \cdot \ex{ \i\pi \ga \sg_1^z }   \, \vp^{(-)}(\bs{\be}_{n-1}^{\prime} ) 
\, + \, \i m (m-1) \varkappa \msc{W}\big( \bs{\be}_{n-1} \big) \;. 
\end{multline}
Here, the first term in the second line vanishes by virtue of Lemma \ref{Lemme annulation K transformee}. Below, we establish that the second term 
vanishes as well. 

A direct calculation shows that 
\beq
\mc{W}\big(\bs{\be}_n ; \bs{u}_{m-1}\big)  g(\bs{\be}_{n-1}^{\prime} ; \bs{u}_{m-1})  \; = \;  g(\bs{\be}_{n-1}^{\prime} ; \bs{u}_{m-2}) 
\Bigg\{  \mc{V}\big( \bs{\be}_{n-1} ; u_{m-1}, \bs{u}_{m-2} \big) \, - \,  \mc{V}\big( \bs{\be}_{n-1} ; u_{m-1} + 2\i\pi , \bs{u}_{m-2} \big) \Bigg\}
\enq
in which 
\beq
\mc{V}\big( \bs{\be}_{n-1} ; u , \bs{u}_{m-2} \big)  \; = \;   \f{ c^{m-1}_{\tau} \sinh\big(\i \tfrac{\pi^2}{\xi} \big)    }{  \sinh\big[ \tfrac{\pi}{\xi}(\be_1-u)\big] } \ex{ \ga(\be_1-u+\i\f{\pi}{2}) }
\pl{\ell=1}{m-2} \bigg\{ \sinh(u_{\ell}-u) \sinh\big[ \tfrac{\pi}{\xi}( u - u_{\ell} - \i \pi )\big]  \bigg\} \pl{\ell=2}{n-1} (a \phi)\big( \be_{\ell}-u\big) \; . 
\enq
Above, we have introduced the constant $c_{\tau} \, = \, \op{F}^2\big(\i \tfrac{\pi}{2} \big) \cdot  \op{F}^2\big(-\i \tfrac{\pi}{2} \big) \cdot 
\varpi^2\big(\i \tfrac{\pi+\xi}{2}, \i \tfrac{\xi-3\pi}{2}\big)$. 
This recasts $\msc{W}\big( \bs{\be}_{n-1} \big) $ as
\bem
\msc{W}\big( \bs{\be}_{n-1} \big) \; = \; \Int{ (\msc{C}_{\bs{\be}_n}^{\prime})^{m-2}  }{}   \hspace{-1mm} \dd^{m-2}   u \,   g(\bs{\be}_{n-1}^{\prime} ; \bs{u}_{m-2})
\big(  \op{v}^-_1 \, \op{v}^+_2 \dots \op{v}^{+}_{n-1}\,\op{v}^{-}_n \big)^{\op{t}}   \\
\times \pl{p=1}{m-2}\wt{\op{C}}_{2,\dots,n-1}^{\,(\ga)}(\bs{\be}_{n-1}^{\prime} ; u_p)  \cdot  \op{O}_{1,\dots,n}^{(\ga)}(\bs{\be}_n)  \cdot  \chi(\bs{\be}_{n-1} ; \bs{u}_{m-2})  
p\Big( \bs{\be}_{n} ;  \big(\bs{u}_{m-1},\be_1 \big) \Big)_{\mid \be_{1n}=\i\pi} 
\end{multline}
where, by using the $2\i\pi$ periodicity of $p$ in respect to the $u$-variables, 
\beq
\label{diff contour cancel}
\chi(\bs{\be}_{n-1} ; \bs{u}_{m-2})  \; = \; \Bigg\{ \Int{  \msc{C}_{\bs{\be}_n}^{\prime}  }{} \, - \,  \Int{  \msc{C}_{\bs{\be}_n}^{\prime} +2\i\pi }{}  \Bigg\} \dd u
\mc{V}\big( \bs{\be}_{n-1} ; u  ;  \bs{u}_{m-2} \big)   p\Big( \bs{\be}_{n} ;  \big(\bs{u}_{m-2}, u ,\be_1 \big) \Big)_{\mid \be_{1n}=\i\pi} \;. 
\enq
A direct inspection shows that the integrand gives rise to a convergent integral at infinity and that it has poles at 
\beq
u=\be_1+ \i k \xi \;, k \in \mathbb{Z} \qquad   u=\be_{s} -\i p \xi-2\i\pi \ell \quad \e{and} \quad   u=\be_{s} + \i p \xi+ \i\pi (2\ell + 1) \;\; p, \ell \in \mathbb{N}
\enq
with $s=2,\dots, n-1$. Moreover, $ \msc{C}_{\bs{\be}_n}^{\prime} \setminus  \msc{C}_{\bs{\be}_n}^{\prime} +2\i\pi $ is homotopic to
\bem
\msc{C}_{ \bs{\be}_n }^{\prime} \setminus \Big\{   \msc{C} _{ \bs{\be}_n   }^{\prime}+ 2\i\pi \Big\}  \; \simeq \; 
\bigcup\limits_{a=2}^{n} \Bigg\{   \bigg( \bigcup\limits_{\ell, p \geq 0}\Big\{ - \Dp{}\mc{D}_{\be_a-\i p  \xi - 2  \i \pi \ell ,\eps}  \Big\} 
\bigcup\limits_{\ell, p \geq 0}  \Big\{  \Dp{}\mc{D}_{\be_a - \i p  \xi - 2  \i \pi (\ell-1) ,\eps}  \Big\}  \bigg)   \\
\bigcup   \bigg( \bigcup \limits_{  p \geq 0}\Big\{  - \Dp{}\mc{D}_{\be_a-\i\pi - \i (p+1) \xi ,\eps}  \Big\}  
\bigcup \limits_{  p \geq 0}\Big\{    \Dp{}\mc{D}_{\be_a + \i\pi - \i (p+1) \xi ,\eps}  \Big\}  \bigg)\\ 
\bigcup \bigg( \bigcup\limits_{\ell, p \geq 0}\Big\{   \Dp{}\mc{D}_{\be_a + \i p \xi + \i \pi (2\ell+1),\eps}  \Big\}  \bigg) 
\bigcup\limits_{\ell, p \geq 0}\Big\{  - \Dp{}\mc{D}_{\be_a + \i p \xi + \i\pi (2\ell+3) ,\eps}  \Big\}  \bigg)
\bigcup 
\bigg( \bigcup\limits_{p  \geq 0}\Big\{  \Dp{}\mc{D}_{\be_a + \i p  \xi - \i\pi ,\eps}  \Big\}  
\bigcup\limits_{p \geq 0}\Big\{  - \Dp{}\mc{D}_{\be_a + \i p \xi + \i\pi ,\eps}  \Big\}   \bigg) \Bigg\} \\
\; = \;  \bigcup\limits_{a=2}^{n} \Bigg\{ \bigcup\limits_{p  \geq 0}  \bigg( \Big\{    \Dp{}\mc{D}_{\be_a - \i p \xi + 2\i \pi ,\eps} \Big\}
\bigcup \Big\{    \Dp{}\mc{D}_{\be_a + \i p \xi + \i\pi ,\eps} \Big\} \bigcup \Big\{  \Dp{}\mc{D}_{\be_a + \i p \xi - \i\pi ,\eps}  \Big\} \bigg) 
\bigcup \Big\{  - \Dp{}\mc{D}_{\be_a + \i p \xi + \i\pi ,\eps} \Big\}  \\ 
\bigcup \Big\{   \Dp{}\mc{D}_{\be_a - \i (p +1) \xi + \i\pi ,\eps} \Big\} \bigcup \Big\{  - \Dp{}\mc{D}_{\be_a - \i (p +1)\xi - \i\pi ,\eps} \Big\} \Bigg\}\\ 
\; = \; \bigcup\limits_{a=2}^{n} \Bigg\{ \bigcup\limits_{p  \geq 0}  \bigg( \Big\{    \Dp{}\mc{D}_{\be_a - \i p \xi + 2\i \pi ,\eps} \Big\}
\bigcup \Big\{    \Dp{}\mc{D}_{\be_a + \i p \xi -  \i\pi ,\eps} \Big\} 
\bigcup \Big\{   \Dp{}\mc{D}_{\be_a - \i (p +1) \xi + \i\pi ,\eps} \Big\} \bigcup \Big\{  - \Dp{}\mc{D}_{\be_a - \i (p +1)\xi - \i\pi ,\eps} \Big\} \Bigg\} \;,
\end{multline}
which entails that no poles are involved in the difference of contours in \eqref{diff contour cancel}, and thus the two contributions cancel each other. \qed

\section{The breather/soliton-antisoliton form factors}
\label{Appendix preuve decomposition breathers}

\proof of Proposition \ref{Proposition FF dans le secteur breather}.

\vspace{3mm}
The representation is established by induction. We thus start by obtaining the representation in the
case where only one breather particle $\mf{b}_k$  with $ k  \in \intn{ 1 }{ n_{\xi} }$ is present.
The form factor is deduced from the purely solitonic one $\bs{\Psi}_{1,\dots,n}\big[ p\big]$ by taking the residue as imposed by Axiom $V)$,
\textit{c.f.} \eqref{recursion breather 2}:
\beq
\bs{\Phi}^{\mf{b}_k}_{1,\dots,n}[p]\big(\alpha, \bs{\be}_n \big) \; = \;
\e{Res}\Big( \bs{\Psi}_{-1,0,\dots,n}\big[ p\big]\big(\bs{\de}_2, \bs{\be}_n \big) \dd \de_1 \, , \, \de_{12} \, = \, \i \mf{u}_k \Big)_{ \mid \ov{\de}_{12} = 2 \a }
\cdot \Ga^{(k)}_{-1\, 0}\;.
\enq
Here, on the \textit{rhs} we chose to start labelling the spaces building up $(\mf{h}^*)^{\otimes n+2}$
with the indices $-1, 0,\dots,n$.
Observe that the integrand  $\wh{\Cx}_{1,\dots,n}^{(\ga)}$, see \eqref{definition hat C},
has, among others, poles at
\beq
u_{p} = \de_1-\i\pi + \i t \xi\, , \quad  t \in \intn{0}{k} \quad \e{and} \quad
u_{p} \, = \, \de_2- \i \ell \xi \, , \quad  \ell \in \intn{0}{k} \;.
\enq
The first set of poles has index $+1$ in respect to $\msc{C}_{\bs{\de}_2,\bs{\be}_n}$
while the second one has index $-1$. In the limit when $\de_{12} \tend \i \mf{u}_k$ and $\ov{\de}_{12}=2\a$,
these poles pinch the integration curve $\msc{C}_{\bs{\de}_2,\bs{\be}_n}$ and thus generate a pole of
$\bs{\Psi}_{-1,0,\dots,n}\big[ p\big]\big(\bs{\de}_2, \bs{\be}_n \big) $ at $\de_{12} = \i \mf{u}_k$. One may check that these
are the only poles pinching the contour in this specific limit.
The residue at $\de_{12} = \i \mf{u}_k$ can be computed by using similar handlings to those outlined
in Lemma  \ref{Lemme engendrement des poles par pincement} of Appendix \ref{subsection axiom iv}. All-in-all, one gets
\bem
\bs{\Phi}^{\mf{b}_k}_{1,\dots,n}[p]\big(\alpha, \bs{\be}_n \big)  \, = \,
- (m+1) \sul{\ell=0}{k} \Int{ (\Ga^{\prime})^{m} }{} \hspace{-2mm} \dd^{m} u \\
\times \Res\Big(  2\i\pi \e{Res}\Big( \wh{\Cx}_{-1,\dots,n}^{(\ga)}\big( \bs{\de}_2, \bs{\be}_n \, ; \, \bs{u}_{m+1}\big) \cdot  \dd u_{m+1} \, ,
\, u_{m+1}=\de_2-\i \ell \xi  \Big) \cdot  \dd \de_1, \de_{12}=\i \mf{u}_k \Big){}_{ _{\mid \ov{\de}_{12} = 2\alpha} }
\hspace{-3mm} \cdot \Ga^{(k)}_{-1\, 0}  \;.
\label{ecriture eqn breather apres evaluation pincements}
\end{multline}
Here, the new integration contour is defined as
\beq
\Ga^{\prime} \, = \, \Big\{ \msc{C}_{\bs{\de}_2,\bs{\be}_n} \bigcup\limits_{s=0}^{k} \Dp{}\mc{D}_{\de_2 - \i s \xi, \eps}
\Big\}_{ \de_{12}=\i \mf{u}_k, \ov{\de}_{12}=2\a }
\enq
with $\eps$ small enough. The contour on the \textit{rhs} is such that the poles of $\wh{\Cx}_{-1,\dots,n}^{(\ga)}\big( \bs{\de}_2, \bs{\be}_n \, ; \, \bs{u}_{m+1}\big)$
in $u_{\ell}$ at $\de_1-\i\pi + \i t \xi$, $t \in \intn{0}{k}$, and $\de_2- \i \ell \xi$,  $\ell \in \intn{0}{k}$,
all have index $+1$ in respect to it. In particular, there is no problem to evaluate it at
$\bs{\de}_2 =  \big( \a + \tfrac{\i}{2} \mf{u}_k ,  \a - \tfrac{\i}{2} \mf{u}_k \big)$. A direct inspection then shows that
$\Ga^{\prime}=\msc{C}^{(k)}_{\a, \bs{\be}_n}$.

We first address the computation of the residue at the pole $u_1=\de_2-\i \ell \xi$. The latter stems from $g \big(\bs{\be}_n;\bs{u}_{m} \big)$
and, more precisely,
issues from the factor $\big( a \phi)(\de_2-u_{m+1})$.  Recalling the definition of this function's residue \eqref{definition residu de a phi}, it holds
\begin{multline}
\e{Res}\Big( g\big( \bs{\de}_2, \bs{\be}_n ; \bs{u}_{m+1}\big) \cdot  \dd u_{m+1} \, , \, u_{m+1}  = \de_2-\i \ell \xi  \Big) \,  =  \,
- C^{(\i \ell \xi)}_{a\phi} \, g(\bs{\be}_{n},\bs{u}_{m}) \, \big( a \phi\big)(\de_{12}+\i \ell \xi) \,   \op{F}(\de_{12}) \\
\times \pl{b=1}{n} \Big\{ \op{F}(\de_1 -\be_{b}) \op{F}(\de_2-\be_{b}) \big( a \phi\big)(\be_{b}-\de_2 +\i \ell \xi) \Big\}
\cdot \pl{t=1}{m} \Big\{  \tau(\de_2-\i \ell \xi -u_{t})  \big( a \phi\big)(\de_1-u_{t}) \big( a \phi\big)(\de_2-u_{t})  \Big\}\;.
\end{multline}
After that, it remains to compute the residue at $\de_{1}= \de_2 + \i \mf{u}_{k}$.
This issues from the off-shell Bethe vector, and more precisely the operator
$\wt{\op{C}}_{-1,\dots,n}^{\,(\ga)}(\bs{\de}_2,\bs{\be}_n;\de_2-\i \ell \xi)$.
To compute the corresponding residue, one uses the multisite expansion:
\begin{equation}
\wt{\op{C}}_{-1,\dots,n}^{\,(\ga)}(\bs{\de}_2,\bs{\be}_n;\de_2-\i \ell \xi) \, = \,
\wt{\op{C}}_{-1,0}^{\,(\ga)}(\de_1,\de_2 ;\de_2-\i \ell \xi)\cdot\wt{\op{A}}_{1,\dots,n}^{\,(\ga)}(\bs{\be}_n,\de_2-\i \ell \xi) +
\wt{\op{D}}_{-1,0}^{\,(\ga)}(\de_1,\de_2;\de_2-\i \ell \xi)\cdot\wt{\op{C}}_{1,\dots,n}^{\,(\ga)}(\bs{\be}_n,\de_2-\i \ell \xi) \;.
\end{equation}
Then, the explicit form of the pseudo-vaccum $\Om_{-1\dots n} = \big(\op{v}^+_{-1} \otimes \op{v}^+_0\big)^{\op{t}} \otimes \Om_{n}$ leads to:
\begin{equation}
\Om_{-1\dots n} \wt{\op{C}}_{-1,\dots,n}^{\,(\ga)}(\bs{\be}_n;\de_2-\i \ell \xi) \, = \,
\Om_{-1\dots n} \cdot \wt{\op{C}}_{-1,0}^{\,(\ga)}(\de_1,\de_2 ;\de_2-\i \ell \xi)
+ \wt{b}\big(\de_{12}+\i \ell \xi\big) \wt{b}\big(\i \ell \xi\big) \, \Om_{-1\dots n} \cdot \wt{\op{C}}_{1,\dots,n}^{\,(\ga)}(\bs{\be}_n;\de_2-\i \ell \xi) \; .
\end{equation}
The second term vanishes because of $ \wt{b}\big(\i \ell \xi)=0$, while, for the first term we get :
\begin{equation}
\big(\op{v}^+_{-1} \otimes \op{v}^+_0\big)^{\op{t}} \cdot \wt{\op{C}}_{-1,0}^{\,(\ga)}(\de_1,\de_2 ;\de_2-\i \ell \xi) \, =\,  \wt{b}\big(\de_{12}+\i \ell \xi\big)
\, \wt{c}^{\, (\ga)}\big(\i \ell \xi \big)\,  \big(\op{v}^+_{-1} \otimes \op{v}^-_0\big)^{\op{t}}
+ \wt{c}^{\, (\ga)}\big(\de_{12}+\i \ell \xi\big)\, \big(\op{v}^-_{-1} \otimes \op{v}^+_0\big)^{\op{t}} \;.
\end{equation}
Taking the residue, all-in-all, yields
\bem
\label{apparition intertwiner twist}
\e{Res} \Big(\Om_{-1\dots n} \wt{\op{C}}_{-1,\dots,n}^{\,(\ga)}(\bs{\de}_2,\bs{\be}_n ;\de_2-\i \ell \xi) \cdot  \dd \de_1, \de_{12}=\i \mf{u}_{k} \Big)   \\
\, = \,  \f{\xi}{\pi}(-1)^{\ell} \sinh\Big( \i \tfrac{ \pi^2 }{ \xi }\Big) \ex{\i\ga \ell \xi}
\Big(\op{v}_{-1}^+\op{v}_{0}^- + (-1)^k \ex{\i\ga\mf{u}_{k} } \op{v}_{-1}^-\op{v}_{0}^+\Big)^{\op{t}}  \otimes \Om_{n}  \;.
\end{multline}
Notice that the 2-covector in the right-hand side of \eqref{apparition intertwiner twist} is related to the (anti)soliton/breather intertwiner \eqref{exp intertwiner} as:
\begin{equation}
\Big(\op{v}_{-1}^+\op{v}_{0}^- + (-1)^k \ex{\i\ga\mf{u}_{k} } \op{v}_{-1}^-\op{v}_{0}^+\Big)^{\op{t}}
= \sqrt{2}
\ex{  \f{ \i }{ 2 } \ga \mf{u}_k } \cdot \big( \bs{\vp}_{-10}^{(\mf{b}_k)} \big)^{\mathtt{t}}  \cdot \ex{ -\f{ \i  }{2} \ga \mf{u}_{k} \sigma_{-1}^z} \;.
\end{equation}
This entails that
\bem
\mf{R} \, = \,  \Res\Big\{
\wt{\Cx}_{-1,\dots,n}^{(\ga)}\Big( \bs{\de}_2, \bs{\be}_n \, ; \, \big( \bs{u}_{m}, \de_2-\i \ell \xi, \big) \Big)   \dd \de_1,
\de_{1}= \de_2 + \i \mf{u}_k \Big\}_{ _{\mid \ov{\de}_{12} = 2\alpha} } \\
 =\; \sqrt{2} \f{\xi}{\pi}(-1)^{ \ell + 1 } \sinh\Big( \i \tfrac{ \pi^2 }{ \xi }\Big) \ex{ \i\ga (\ell \xi + \f{ 1 }{ 2 } \mf{u}_k )  }
 \Big(  \big( \bs{\vp}_{-10}^{(\mf{b}_k)} \big)^{\mathtt{t}}  \otimes \Om_{n}   \Big) \cdot  \ex{ - \f{\i}{2} \ga \mf{u}_{k} \sigma_{-1}^z}
\pl{t=1}{m} \wt{\op{C}}_{-1,\dots,n}^{\,(\ga)}\Big( \a + \tfrac{\i}{2} \mf{u}_k, \a - \tfrac{\i}{2} \mf{u}_k, \bs{\be}_n; u_t \Big) \;.
\nonumber
\end{multline}
To simplify the above expression, one moves the covector $\big( \bs{\vp}_{-10}^{(\mf{b}_k)} \big)^{\mathtt{t}}$ through
$\pl{t=1}{m}\wt{\op{C}}_{-1,\dots,n}^{\,(\ga)}\Big( \a + \tfrac{\i}{2} \mf{u}_k, \a - \tfrac{\i}{2} \mf{u}_k, \bs{\be}_n; u_t \Big)$,
because of the definition of the (anti)soliton/breather S-matrix. A short calculation yields:
\begin{equation}
\big( \bs{\vp}_{ab}^{(\mf{b}_k)} \big)^{\mathtt{t}} \cdot \ex{ - \f{\i}{2} \ga \mf{u}_{k}   \sigma_a^z} \cdot
\wt{\op{S}}_{a c}^{\, (\ga)}\Big( \alpha + \tfrac{\i}{2} \mf{u}_k -u_{t} \Big)
\cdot \wt{\op{S}}_{b c}^{\, (\ga)}\Big(\alpha -  \tfrac{\i}{2} \mf{u}_k -u_{t} \Big) \, = \,
\f{ \op{S}_{c}^{(\mf{b}_k)} ( \a - u_{t} )  \, \big( \bs{\vp}_{ab}^{(\mf{b}_k)} \big)^{\mathtt{t}} }
{ a\big(   \a - u_{t}  + \tfrac{\i}{2} \mf{u}_k \big) a\big(   \a - u_{t}  - \tfrac{\i}{2} \mf{u}_k \big)    }
 \cdot \ex{ - \f{\i}{2} \ga \mf{u}_{k}   \sigma_a^z}
\end{equation}
with $\op{S}_{c}^{(\mf{b}_k)}$ as introduced in \eqref{definition matrice S breather k avec le reste}. This leads to
\bem
\mf{R} \,  =\; \sqrt{2} \f{\xi}{\pi}(-1)^{ \ell + 1 } \sinh\Big( \i \tfrac{ \pi^2 }{ \xi }\Big) \ex{ \i\ga (\ell \xi + \f{  1 }{ 2 } \mf{u}_k )  }
\pl{s=1}{m} \Big\{ \wt{S}_{\mf{b}_k, \mf{s} } \big(\a - u_{s} \big) \Big\} \,
\big( \bs{\vp}_{-10}^{(\mf{b}_k)} \big)^{\mathtt{t}}  \otimes \wt{\Cx}_{ 1,\dots,n}^{(\ga)}\Big(   \bs{\be}_n \, ; \,   \bs{u}_{m}  \Big)
\ex{ - \f{\i}{2} \ga \mf{u}_{k}  \sigma_{-1}^z} \;,
\end{multline}
in which
\begin{equation}
\wt{S}_{\mf{b}_k, \mf{s} } (\th) = (-1)^k
\f{  \sinh \big[\tfrac{\pi}{\xi}(\th +  \tfrac{ \i }{2} \mf{u}_k ) \big] }{ \sinh \big[\tfrac{\pi  }{\xi}(  \tfrac{ \i }{2} \mf{u}_k  - \th )\big] } \; .
\end{equation}
Finally, one may push the covector to the very right of the expression by using that
\begin{equation}
	\big( \bs{\vp}_{-10}^{(\mf{b}_k)} \big)^{\mathtt{t}}  \cdot \ex{ - \f{ \i }{ 2 }  \ga \mf{u}_{k} \sigma_{-1}^z}  \cdot
\op{O}_{-1,\dots,n}^{(\ga)}\Big( \a + \tfrac{\i}{2} \mf{u}_k, \a - \tfrac{\i}{2} \mf{u}_k, \bs{\be}_n; u_t \Big)
	= \op{O}_{1,\dots,n}^{(\ga)}(\bs{\be}_n) \cdot  \big( \bs{\vp}_{-10}^{(\mf{b}_k)} \big)^{\mathtt{t}} \;.
\end{equation}
By putting all these partial results together, one represents the  $\mf{b}_k$-breather/$n$-soliton-antisoliton form factor as
\begin{multline}
\bs{\Phi}^{\mf{b}_k}_{1,\dots,n}[p]\big(\alpha, \bs{\be}_n \big)  \, = \,
-\i \f{ \sqrt{2} \xi }{ r_k }  \sinh\big( \i \tfrac{ \pi^2 }{ \xi }\big) \op{F}(\i \mf{u}_k)
\sul{\ell=0}{k} (-1)^{\ell}  \big( a \phi\big)(\i \mf{u}_k +\i \ell \xi) \,  C^{(\i \ell \xi)}_{a\phi}
\,  \ex{\i\ga (\ell \xi+\f{1}{2} \mf{u}_k ) }   \\
\times
\pl{b=1}{n} \bigg\{ \op{F}\big(\alpha +\tfrac{  \i  }{2}\mf{u}_k  -\be_{b}\big) \op{F}\big(\alpha - \tfrac{  \i  }{2}\mf{u}_k  -\be_{b}\big)\,
\big( a \phi\big)\big(\be_{b}-\alpha + \tfrac{  \i  }{2}\mf{u}_k  + \i \ell \xi \big) \bigg\} \\
\times \Int{ ( \msc{C}^{(k)}_{\a, \bs{\be}_n} )^m }{} \hspace{-1mm} \dd^{m} u \,
\pl{t=1}{m} \bigg\{
\tau\big( \alpha - \tfrac{  \i  }{2}\mf{u}_k  -\i \ell \xi -u_{t} \big) \,
\big( a \phi\big)\big(\alpha +\tfrac{  \i  }{2}\mf{u}_k  -u_{t}\big) \,
\big( a \phi\big)\big(\alpha - \tfrac{  \i  }{2}\mf{u}_k  -u_{t}\big) \,
\wt{S}_{\mf{b}_k, \mf{s}} \big( \alpha - u_{t} \big)  \bigg\}  \\
\times \wt{\Cx}_{1,\dots,n}^{\,(\ga)}\big( \bs{\be}_n  \, ; \, \bs{u}_{m} \big)
\cdot \op{O}_{1,\dots,n}^{(\ga)}(\bs{\be}_n) \cdot g(\bs{\be}_{n},\bs{u}_{m}) \cdot
(m+1)  \,
p \big(\alpha + \tfrac{  \i  }{2}\mf{u}_k  , \alpha - \tfrac{  \i  }{2}\mf{u}_k  ,\bs{\be}_n;
\alpha - \i \tfrac{  \mf{u}_k  }{2}-\i \ell \xi,\bs{u}_{m} \big)  \;.
\end{multline}
It now remains to recast the above expression in terms of the functions arising in the statement of the proposition.
First, one observes that
\beq
\big( a \phi\big)( \la +\i \ell \xi)  \, = \, (-1)^{\ell}
\big( a \phi\big)(\la )  \pl{t=0}{\ell-1}  \f{ \cosh\big[\tfrac{1}{2}(\la + \i t \xi) \big] }{ \sinh\big[\tfrac{1}{2}(\la + \i (t+1) \xi) \big] } \;.
\enq
This identity, along with \eqref{ecriture forme alternative C0 a phi}, \eqref{ecriture reduction cste C0 a phi vers varkappa}
and the explicit expression for $r_k$ \eqref{exp intertwiner} yields:
\beq
  -\i \f{ \sqrt{2} \xi }{ r_k }  \sinh\big( \i \tfrac{ \pi^2 }{ \xi }\big) \op{F}(\i \mf{u}_k)
\big( a \phi\big)(\i \mf{u}_k +\i \ell \xi) \,  C^{(\i \ell \xi)}_{a\phi}  \,  \ex{\i\ga (\ell \xi + \f{1}{2}\mf{u}_k) } \, = \, d_{ \ell; k}  \;.
\enq
Further, by construction from \eqref{s breathers p function}, one gets:
\beq
 d_{\ell;k}
\cdot
(m+1)  \,
p \big(\alpha + \tfrac{  \i  }{2}\mf{u}_k  , \alpha - \tfrac{  \i  }{2}\mf{u}_k  ,\bs{\be}_n;
\alpha - \i \tfrac{  \mf{u}_k  }{2}-\i \ell \xi,\bs{u}_{m} \big)  \, =  \,	p_{\ell;k} \big(\alpha ; \bs{\be}_n ;\bs{u}_{m} \big)   \;.
\enq
Next, by using the shift properties of the dilogarithms and recalling \eqref{chi function}, we get
\beq
\tau\big( \th  - \tfrac{  \i  }{2}\mf{u}_k  -\i \ell \xi   \Big) \big( a \phi\big)\big(\th +\tfrac{  \i  }{2}\mf{u}_k   \big)
\big( a \phi\big)\big(\th - \tfrac{  \i  }{2}\mf{u}_k   \big) \wt{S}_{\mf{b}_k, \mf{s}} \big( \th \big)
\, =  \, \chi_{\ell;k}(\th ) \; .
\end{equation}
Finally, recalling the definitions of $\op{F}_{\mf{b}_k, \mf{s}}$ \eqref{minimal bs form factor} and $\rho_{\ell;k}$ \eqref{rho function}
one observes the identity:
\begin{equation}
\op{F}\big(\th +\tfrac{  \i  }{2}\mf{u}_k  \big) \op{F}\big(\th - \tfrac{  \i  }{2}\mf{u}_k  \big)
\big( a \phi\big)\big(\th + \tfrac{  \i  }{2}\mf{u}_k  + \i \ell \xi \big)  \, = \,  \op{F}_{\mf{b}_k, s} (\th) \cdot \rho_{\ell;k}(\th) \;.
\end{equation}
All-in-all, this thus reproduces the $1$-breather/$n$-soliton-antisolito formula.
\vspace{2mm}

Now assume that the $s$-breathers/$n$-soliton-antisoliton form factor is given by \eqref{exp generale form factor}, this for any $n$.
To obtain the expression for $s+1$ breathers with indices $k_1,\dots,k_{s+1}$ one applies the recursion formula \eqref{recursion breather 2}
to the $s$ breather/$n+2$-soliton-antisoliton form factor and computes the residue at $\de_1=\de_2 + \i \mf{u}_{k_{s+1}}$
analogously to the previous handlings, by again dealing with the pinching of the integration contour. Note that, owing
to the breathers's rapidities $\bs{\a}_s$ being generic, as much as those of the $n$ solitons/antisolitons $\bs{\be}_n$, only the very same kind
of residues arises in the computation for $s=1$. All-in-all, by following very similar steps, one obtains:
%
%
\begin{multline}
\bs{\Phi}^{\bs{\mf{b}}_{s}}_{1,\dots,n}\big[ p\big]\big(\bs{\a}_{s+1}, \bs{\be}_{n} \big) \, = \,
\pl{j=1}{s} \Big\{ \op{F}_{\mf{b}_{k_j}, \mf{s}} \Big(\alpha_j -  \tfrac{\i}{2} \mf{u}_{k_{s+1}}  - \alpha_{s+1}\Big) \cdot
\op{F}_{\mf{b}_{k_j}, \mf{s}} \Big(\alpha_j + \tfrac{\i}{2} \mf{u}_{k_{s+1}}  - \alpha_{s+1}\Big)
\pl{b=1}{n} \op{F}_{\mf{b}_{k_j}, \mf{s}} (\alpha_j - \be_b)  \Big\} \\
\times \hspace{-3mm}
\sul{ \bs{\ell}_{s+1} \in \mc{I}_{\bs{k}_{s+1}} }{ }   \hspace{-3mm}  (-1)^{ \ov{\bs{\ell}}_{s+1} }
\pl{j<t}{s} \Big\{ \op{F}_{\mf{b}_{k_j}, \mf{b}_{k_{t}}} (\alpha_j - \alpha_{t}) \La\Big( \,  ^{k_j,k_t}   _{\ell_j,\ell_t}  \mid\alpha_j - \alpha_{t} \Big)   \Big\}
\pl{j=1}{s} \pl{b=1}{n} \rho_{\ell_j;k_j}(\alpha_j - \be_b)  \\
\times \pl{j=1}{s} \Big\{
\rho_{\ell_j;k_j}\Big(\alpha_j - \tfrac{\i}{2} \mf{u}_{k_{s+1}}  - \alpha_{s+1}\Big) \cdot \rho_{\ell_j;k_j}\Big(\alpha_j + \tfrac{\i}{2} \mf{u}_{k_{s+1}}  - \alpha_{s+1}\Big) \Big\}
\cdot \pl{b=1}{n} \Big\{ \op{F}_{\mf{b}_{k_{s+1}}, \mf{s}} (\alpha_{s+1} - \be_b) \cdot\rho_{\ell_{s+1};k_{s+1}} (\alpha_{s+1} - \be_b) \Big\} \\
	\times  \hspace{-3mm} \Int{ (\msc{C}^{(\bs{k}_{s+1})}_{\bs{\a}_{s+1}, \, \bs{\be}_n} )^{m} }{} \hspace{-4mm} \dd^{m} u \,
\pl{t=1}{m} \chi_{\ell_{s+1};k_{s+1}}(\alpha_{s+1} - u_{t}) \cdot
\pl{j=1}{s} \pl{t=1}{m}   \chi_{\ell_j;k_j}(\alpha_j - u_{t}) \cdot  \pl{j=1}{s} \chi_{\ell_j;k_j}\Big(\alpha_j +\tfrac{\i}{2} \mf{u}_{k_{s+1}} - \alpha_{s+1} + \i \ell_{s+1}\xi \Big) \\
\times \wt{\Cx}_{1,\dots,n}^{\,(\ga)}\big( \bs{\be}_n \, ; \, \bs{u}_m\big)
\cdot \op{O}_{1,\dots,n}^{(\ga)}(\bs{\be}_n) \cdot g(\bs{\be}_{n},\bs{u}_{m}) \cdot
p_{\bs{\ell}_{s+1}; \bs{k}_{s+1}} \big(\bs{\alpha}_{s+1} ; \bs{\be}_n;\bs{u}_m \big) \;.
\end{multline}
It is then a sole matter of bookkeeping to recognise that the above exactly reduces to \eqref{exp generale form factor} at $s+1$. \qed

\section{Vanishing terms in the free fermion limit}
\label{appendix vanish}

\begin{lemme}

Given $\mc{R}_r\big( \bs{\ell}_r ;\bs{\be}_n; \bs{\nu}_{m-r} \big) $ as introduced in \eqref{def succ res}, in the proof of Proposition \ref{prop free fermion limit}, it holds
\beq
\lim_{\xi \tend \, \pi} \, \Int{ \big( \widetilde{\msc{C}}_{\bs{\be}_n} \big)^r }{}  \hspace{-3mm} \dd^{m-r} \nu
\, \mc{R}_r\big( \bs{\ell}_r ;\bs{\be}_n; \bs{\nu}_{m-r} \big) \; = \; 0
\enq
as soon as $r<m$.

\end{lemme}

\Proof For given $\bs{j}_r \in \{0,1\}^r$, let $q = \# \{ s \in \intn{1}{r} \, : \, j_s = 0\}$ and further let $\sg_i \in \intn{1}{r}$, $i=1,\dots q$
be pairwise distinct and such that $j_{\sg_i}=0$, and likewise, $\tilde{\sg}_i \in \intn{1}{r}$, $i=1,\dots r-q$
be pairwise distinct and such that $j_{\tilde{\sg}_i}=1$. Recalling the definitions \eqref{definition beta 0 et beta 1}, denote
\beq
\la_i \, = \, \be_{\ell_{\sg_i}}^{(0)} \qquad \e{and} \qquad
\mu_i \, = \, \be_{\ell_{\tilde{\sg}_i}}^{(1)} \; .
\enq
Then observe that the residues at $\be^{(0)}_a = \be_a - \i\pi + \i\xi$ stem from the poles of $\wt{\op{C}}_{1,\dots,n}^{(\ga)}\big( \bs{\be}_n \, ; \,u\big)$
while the residues at $\be^{(1)}_a = \be_a - \i \xi$ stem from the poles of $g\big(\bs{\be}_n; \bs{u}_r,\bs{\nu}_{m-r} \big)$.
We start by computing the residues at  $u_{\tilde{\sg}_s}=\mu_s$:
\begin{multline}
	\label{residu g function}
\e{Res}\Big( g\big(\bs{\be}_n; \bs{u}_r,\bs{\nu}_{m-r} \big) \pl{s=1}{r-q} \dd  u_{\tilde{\sg}_s}, u_{\tilde{\sg}_s}=\mu_s \Big)_{\mid u_{\sg_s}=\la_s }
\, =  \, \Big\{ -  C^{(\i  \xi)}_{a\phi} \Big\}^{r-q} g\big(\bs{\be}_n; \bs{\la}_q,\bs{\nu}_{m-r} \big) \\
\times \pl{t=1}{r-q} \pl{\substack{k=1 \\k \neq \ell_{\tilde{\sg}_t} } }{ n }  (a\phi)(\be_k - \mu_t)
\pl{k<t}{r-q} \tau(\mu_{kt}) \pl{t=1}{r-q} \Big\{ \pl{k=1}{q} \tau(\la_k-\mu_t)  \pl{k=1}{m-r} \tau(\nu_k-\mu_t)  \Big\} \;.
\end{multline}
where $C^{(\a)}_{a\phi}$ is as defined in \eqref{definition residu de a phi}.

Further, oberve that
\begin{equation}
	\label{residu matrice S}
	\e{Res} \Big( \wt{\op{S}}^{\, (\ga)}(u) \dd u , u = \be-\i(\pi - \xi) \Big)  = \f{\xi}{\pi} \cdot  \sinh\left(\i\tfrac{\pi^2}{\xi}\right) \cdot \left(\ba{cccc}  0 & 0 & 0 & 0  \\
	0  & 1  & - \ex{\i\ga (\pi-\xi)} & 0  \\ 
	0 &  -\ex{-\i\ga (\pi-\xi)} &  1  &  0  \\ 
	0 &  0      &   0       &   0  \ea \right) \;,
\end{equation}
and set $\e{Res} \Big( \wt{\op{C}}_{1,\dots,n}^{(\ga)}\big( \bs{\be}_n \, ; \,u\big) \dd u , u = \be_t-\i(\pi - \xi) \Big) \, = \,
\mf{C}_{a}^{(\ga)}\big( \bs{\be}_n  \big)$. Thus,
\beq
\e{Res} \Big(  \wh{\Cx}_{1,\dots,n}^{(\ga)}\big( \bs{\be}_n \, ; \, \bs{u}_r, \bs{\nu}_{m-r} \big) \pl{s=1}{q} \dd  u_{\sg_s}, u_{\sg_s}=\la_s \Big)_{\mid u_{\tilde{\sg}_s}=\mu_s }
 \; = \;
\wh{\Cx}_{1,\dots,n}^{(\ga)}\big( \bs{\be}_n \, ; \, \bs{\mu}_{r-q}\big) \,   \pl{s=1}{q}\mf{C}_{ \ell_{\sg_s} }^{(\ga)}\big( \bs{\be}_n  \big) \;.
\enq
Further, one observe that $\wt{\op{S}}^{\, (\ga)}_{p0}(u) \tend  \e{id}_0\e{id}_p$ and $\xi \tend \pi$, and thus
$\wt{\op{T}}_{1,\dots,n;0}^{(\ga)}\big( \bs{\be}_n \, ; \,\la\big) \tend  \e{id}_{1,\dots , n;0}$. In particular, this implies that
\begin{equation}
\wt{\op{C}}_{1,\dots,n}^{(\ga)}\big( \bs{\be}_n \, ; \,u_{j}\big) \underset{\xi \rightarrow \pi}{\longrightarrow} (\op{v}_0^-)^{\op{t}} \cdot \e{id}_{1,\dots,n} \cdot \op{v}_0^+ = 0 \;.
\end{equation}
In its turn this yields that pointwise in $\bs{\nu}_{m-r} \in \widetilde{\msc{C}}_{\bs{\be}_n}$
\begin{equation}
\e{Res} \Big( \wh{\Cx}_{1,\dots,n}^{(\ga)}\big( \bs{\be}_n \, ; \, \bs{u}_r, \bs{\nu}_{m-r} \big)  \dd^r u \, , \, u_s=\be_{\ell_s}^{(j_s)} \Big)
\underset{\xi \rightarrow \pi}{\longrightarrow} 0 \;.
\end{equation}
where $\wh{\Cx}_{1,\dots,n}^{(\ga)}$ is as it has been defined in \eqref{definition hat C}.
In order to conclude about the vanishing of the integral, we need to establish a uniform in $\xi$ close to $\pi$ domination of the integrand.
By virtue of the $\Re(\la)\tend \pm \infty$, $|\Im(\la)| < \pi + \tf{ \xi }{ 2 }$  asymptotics of the quantum dilogarithm
\beq
\varpi(\la) \; = \; \ex{ \mp \f{\i}{4\xi} \big( \la^2+ \tfrac{ \xi^2 + 4\pi^2  }{12}  \big) \, + \, \e{O}(\la^{-\infty})  }
\enq
one gets that
\beq
\label{asymptotic aphi}
(a\phi)(\la) \; = \; \f{  \mp 2 \i  \ex{ \pm \f{\i}{4\xi} \big( \f{\xi^2}{4}-(\pi -\f{\xi}{2})^2  \big)  }   }{  \op{F}^2\big(\i \tfrac{\pi}{2}\big) \varpi^2\big( \i \tfrac{\xi-\pi}{2} \big)   }  
\cdot \ex{ \mp \f{\la}{2} \mp \f{\pi \la }{ 2 \xi} }  \cdot \Big( 1 \, + \,   \e{O}(\la^{-\infty}) \Big) \;.
\enq
Then, a direct estimation building on $\xi$ being close to $\pi$ and the fact that $2m=n$ and $r<m$ leads to:
\begin{equation}
\Big| \mc{R}_r\big( \bs{\ell}_r ;\bs{\be}_n; \bs{\nu}_{m-r} \big) \Big|  \, <  \, K \cdot  \pl{p=1}{m-r}\ex{ \f{\ups_p}{2}\Re(\nu_{p}) } \quad
\e{as} \quad \Re(\nu_{p}) \tend \ups_p \infty \;, \; \ups_p \in \{\pm\},
\end{equation}
where $K>0$ is a constant (although depending on $\bs{\be}_n$ fixed). \qed

\section{Some auxiliary results}
\label{Appendice Resultats auxiliaires}
\subsection{Algebraic relations for the monodromy matrix}

The monodromy matrix satisfies the Yang-Baxter equation 
\beq
\op{S}_{00^{\prime}}^{(\ga)}(\la-\mu) \cdot \op{T}_{1,\dots,n; 0^{\prime} }^{(\ga)}(\bs{\be}_n ; \mu)  \cdot \op{T}_{1,\dots,n; 0 }^{(\ga)}(\bs{\be}_n ; \la) \; = \; 
\op{T}_{1,\dots,n; 0 }^{(\ga)}(\bs{\be}_n ; \la) \cdot \op{T}_{1,\dots,n; 0^{\prime} }^{(\ga)}(\bs{\be}_n ; \mu) \cdot  \op{S}_{00^{\prime} }^{(\ga)}(\la-\mu) \;. 
\enq
The latter gives rise to the exchange relations:
\beqa
\op{C}_{1,\dots,n}^{(\ga)}(\bs{\be}_n ; \mu) \cdot  \op{A}_{1,\dots,n}^{(\ga)}(\bs{\be}_n ; \la)  	& = & \f{ a(\la-\mu) }{ b(\la-\mu) }  \, \op{A}_{1,\dots,n}^{(\ga)}(\bs{\be}_n ; \la)
\cdot \op{C}_{1,\dots,n}^{(\ga)}(\bs{\be}_n ; \mu)  
\, - \,  \f{ c^{(\ga)}(\la-\mu) }{ b(\la-\mu) } \,  \op{A}_{1,\dots,n}^{(\ga)}(\bs{\be}_n ; \mu) \cdot \op{C}_{1,\dots,n}^{(\ga)}(\bs{\be}_n ; \la) \nonumber \\
\op{C}_{1,\dots,n}^{(\ga)}(\bs{\be}_n ; \mu) \cdot  \op{D}_{1,\dots,n}^{(\ga)}(\bs{\be}_n ; \la)  	& = & \f{ a(\mu-\la) }{ b(\mu-\la) } \, 
\op{D}_{1,\dots,n}^{(\ga)}(\bs{\be}_n ; \la) \cdot \op{C}_{1,\dots,n}^{(\ga)}(\bs{\be}_n ; \mu)  
\, - \,  \f{ c^{(-\ga)}(\mu-\la) }{ b(\mu-\la) } \,  \op{D}_{1,\dots,n}^{(\ga)}(\bs{\be}_n ; \mu) \cdot \op{C}_{1,\dots,n}^{(\ga)}(\bs{\be}_n ; \la)    \;. 
\nonumber 
\eeqa
Similarly, one may establish the exchange relation
\beq
\op{S}_{00^{\prime}}^{(\ga)}(\be_1-\mu + 2\i\pi) \cdot \op{T}_{1,\dots,n; 0^{\prime} }^{(\ga)}(\bs{\be}_n ; \mu)  \cdot \op{T}_{1,\dots,n; 0 }^{(\ga)}(\bs{\be}_n ; \be_1) \; = \; 
\op{T}_{1,\dots,n; 0 }^{(\ga)}(\bs{\be}_n ; \be_1) \cdot \op{T}_{1,\dots,n; 0^{\prime} }^{(\ga)}(\bs{\be}_n +2\i\pi \bs{e}_1; \mu) \cdot  \op{S}_{00^{\prime} }^{(\ga)}(\be_1-\mu) \;,  
\label{ecriture YBE avec des shifts de rapidities}
\enq
which yields 
\beqa
\op{C}_{1,\dots,n}^{(\ga)}(\bs{\be}_n ; \mu) \cdot  \op{A}_{1,\dots,n}^{(\ga)}(\bs{\be}_n ; \be_1)  	& = & \f{ a(\be_1-\mu) }{ b(\be_1-\mu+2\i\pi) }  \, \op{A}_{1,\dots,n}^{(\ga)}(\bs{\be}_n ; \be_1)
\cdot \op{C}_{1,\dots,n}^{(\ga)}(\bs{\be}_n +2\i\pi \bs{e}_1 ; \mu)    \\
&& \hspace{4cm} \, - \,  \f{ c^{(\ga)}( \be_1 + 2\i\pi - \mu ) }{ b(\be_1 + 2\i\pi - \mu ) } \,  \op{A}_{1,\dots,n}^{(\ga)}(\bs{\be}_n ; \mu) \cdot \op{C}_{1,\dots,n}^{(\ga)}(\bs{\be}_n ; \be_1) \nonumber \\
\op{C}_{1,\dots,n}^{(\ga)}(\bs{\be}_n ; \mu) \cdot  \op{D}_{1,\dots,n}^{(\ga)}(\bs{\be}_n ; \be_1)  	& = & \f{ a(\mu-\be_1-2\i\pi ) }{ b(\mu-\be_1) } \,  \op{D}_{1,\dots,n}^{(\ga)}(\bs{\be}_n ; \be_1) 
\cdot \op{C}_{1,\dots,n}^{(\ga)}(\bs{\be}_n +2\i\pi \bs{e}_1; \mu) \\
&& \hspace{4cm}  
\, - \,  \f{ c^{(-\ga)}(\mu-\be_1) }{ b(\mu-\be_1) } \,  \op{D}_{1,\dots,n}^{(\ga)}(\bs{\be}_n ; \mu) \cdot \op{C}_{1,\dots,n}^{(\ga)}(\bs{\be}_n ; \be_1)    \;. 
\nonumber 
\eeqa
Note that the second exchange relation involving $\op{D}_{1,\dots,n}^{(\ga)}$ follows from the overall exchange relation 
\beq
\op{T}_{1,\dots,n; 0^{\prime} }^{(\ga)}(\bs{\be}_n ; \mu)  \cdot \op{T}_{1,\dots,n; 0 }^{(\ga)}(\bs{\be}_n ; \be_1)  \cdot  \op{S}_{00^{\prime} }^{(-\ga)}(\mu-\be_1) \; = \; 
\op{S}_{00^{\prime}}^{(-\ga)}(\mu- \be_1- 2\i\pi) \cdot \op{T}_{1,\dots,n; 0 }^{(\ga)}(\bs{\be}_n ; \be_1) \cdot \op{T}_{1,\dots,n; 0^{\prime} }^{(\ga)}(\bs{\be}_n +2\i\pi \bs{e}_1; \mu) 
\enq
which is obtained from \eqref{ecriture YBE avec des shifts de rapidities} by using that $\op{S}^{(-\ga)}(-\la)\op{S}^{(\ga)}(\la) = \e{id}$. \\
From the above exchange relation, recalling $\Cx_{1,\dots,n}^{(\ga)}\big(\bs{\be}_n;\bs{u}_m \big)$ as introduced in \eqref{definition C vect de Bethe}, one infers that 
\bem
\Cx_{1,\dots,n}^{(\ga)}\big( \bs{\be}_n \, ; \, \bs{u}_m\big) \cdot  \op{A}_{1,\dots,n}^{(\ga)}(\bs{\be}_n ; \be_1)   \; = \; 
\pl{k=1}{m}\f{ a(\be_1 - u_k) }{ b(\be_1 - u_k+2\i\pi) } \cdot \pl{k=1}{n} a(\be_{k1}) \cdot  \Cx_{1,\dots,n}^{(\ga)}\big( \bs{\be}_n +2\i\pi \bs{e}_1\, ; \, \bs{u}_m\big) \\
\, - \, \sul{p=1}{m}  \f{ c^{(\ga)}( \be_1 + 2\i\pi - u_p ) }{ b(\be_1 + 2\i\pi - u_p ) } \pl{\substack{c=1 \\ \not= p } }{ m } \f{a(u_{pc}) }{b(u_{pc})} 
\pl{k=1}{n} a(\be_{k} - u_p) \cdot \Cx_{1,\dots,n}^{(\ga)}\big( \bs{\be}_n  \, ; \, (\wh{\bs{u}}_m^{(p)},\be_1)  \big)
\label{ecriture recurrence pour A intermediaire}
\end{multline}
and similarly for the action of the $\op{D}$ operator
\bem
\Cx_{1,\dots,n}^{(\ga)}\big( \bs{\be}_n \, ; \, \bs{u}_m\big) \cdot  \op{D}_{1,\dots,n}^{(\ga)}(\bs{\be}_n ; \be_1)   \; = \; 
\pl{k=1}{m}\f{ a( u_k-\be_1-2\i\pi) }{ b( u_k - \be_1) } \cdot \pl{k=1}{n} b(\be_{k1}) \cdot  \Cx_{1,\dots,n}^{(\ga)}\big( \bs{\be}_n +2\i\pi \bs{e}_1\, ; \, \bs{u}_m\big) \\
\, - \, \sul{p=1}{m}  \f{ c^{(-\ga)}( u_p - \be_1) }{ b( u_p - \be_1) } \pl{\substack{c=1 \\ \not= p } }{ m } \f{a( u_{cp} ) }{b( u_{cp} )} 
\pl{k=1}{n} b(\be_{k} - u_p) \cdot \Cx_{1,\dots,n}^{(\ga)}\big( \bs{\be}_n  \, ; \, (\wh{\bs{u}}_m^{(p)},\be_1)  \big)
\label{ecriture recurrence pour D intermediaire}
\end{multline}
in which, we agree upon 
\beq
\wh{\bs{u}}_m^{(p)} \, = \, \big( u_1,\dots, u_{p-1}, u_{p+1},\dots, u_{m} \big) \;.
\enq
Note that the first line in the \textit{rhs} of \eqref{ecriture recurrence pour D intermediaire} actually vanishes. Thus, expressing the latter exchange relations in terms of 
$\wt{\Cx}_{1,\dots,n}^{(\ga)}\big( \bs{\be}_n \, ; \, \bs{u}_m\big)$ as introduced in \eqref{definition tilde C vect de Bethe}, one gets 
\bem
\wt{\Cx}_{1,\dots,n}^{(\ga)}\big( \bs{\be}_n \, ; \, \bs{u}_m\big) \cdot  \op{A}_{1,\dots,n}^{(\ga)}(\bs{\be}_n ; \be_1)   \; = \; 
\f{ \pl{k=1}{n} a(\be_{k1}) }{ \pl{k=1}{m} \wt{b}(\be_1 - u_k+2\i\pi) }  \cdot  \wt{\Cx}_{1,\dots,n}^{(\ga)}\big( \bs{\be}_n +2\i\pi \bs{e}_1\, ; \, \bs{u}_m\big) \\
\, - \, \sul{p=1}{m}  \f{ c^{(\ga)}( \be_1 + 2\i\pi - u_p ) }{ b(\be_1 + 2\i\pi - u_p ) } \pl{\substack{c=1 \\ \not= p } }{ m } \f{a(u_{pc}) }{b(u_{pc})} 
\pl{k=1}{n} a(\be_{k1}) \cdot \wt{\Cx}_{1,\dots,n}^{(\ga)}\big( \bs{\be}_n  \, ; \, (\wh{\bs{u}}_m^{(p)},\be_1)  \big)
\label{ecriture recurrence pour A}
\end{multline}
and similarly for the action of the $\op{D}$ operator
\bem
\wt{\Cx}_{1,\dots,n}^{(\ga)}\big( \bs{\be}_n \, ; \, \bs{u}_m\big) \cdot  \op{D}_{1,\dots,n}^{(\ga)}(\bs{\be}_n ; \be_1)    \\
\; =    \, - \, \sul{p=1}{m}  \f{ c^{(-\ga)}( u_p - \be_1) }{ b( u_p - \be_1) } \pl{\substack{c=1 \\ \not= p } }{ m } \f{a( u_{cp} ) }{b( u_{cp} )} 
\pl{k=1}{n} \Big\{ a(\be_{k1}) \wt{b}(\be_{k} - u_p) \Big\} \cdot \wt{\Cx}_{1,\dots,n}^{(\ga)}\big( \bs{\be}_n  \, ; \, (\wh{\bs{u}}_m^{(p)},\be_1)  \big)
\label{ecriture recurrence pour D}
\end{multline}

\begin{lemme}

The below inductive form of the operator action takes place
\bem
\Big(\op{v}_1^{+}\cdots \op{v}_{n-1}^{+} \op{v}_n^{-} \Big)^{\op{t}} \pl{p=1}{m-1} \wt{\op{C}}_{1,\dots, n}^{\,(\ga)}\big(\bs{\be}_{n} ; u_{p}) \; = \; 
\pl{s=1}{m-1} \Big\{ \,  \wt{b}(\be_1-u_s) \, \wt{b}(\be_n-u_s)  \Big\} \, \Big(\op{v}_2^{+}\cdots \op{v}_{n-1}^{+} \Big)^{\op{t}} \pl{p=1}{m-1} \wt{\op{C}}_{2,\dots, n-1}^{\,(\ga)}\big(\bs{\be}_{n-1}^{\prime} ; u_{p}) 
\cdot \Big(\op{v}_1^{+}\, \op{v}_{n}^{-} \Big)^{\op{t}}  \\
\, + \, \sul{k=1}{m-1} \wt{c}^{\,(\ga)}(\be_1-u_k) \cdot \pl{s=1}{m-1} \wt{b}(\be_n-u_s) \cdot \pl{\substack{ s=1 \\ \not=k}  }{ m-1 } \f{ \wt{b}(\be_1-u_s) }{ \wt{b}(u_{k} - u_s) }
\cdot  \Big(\op{v}_2^{+}\cdots \op{v}_{n-1}^{+} \Big)^{\op{t}} \pl{ \substack{ p=1 \\ \not=k}  }{m-1} \wt{\op{C}}_{2,\dots, n-1}^{\,(\ga)}\big(\bs{\be}_{n-1}^{\prime} ; u_{p}) 
\cdot \Big(\op{v}_1^{-}\, \op{v}_{n}^{-} \Big)^{\op{t}} \;.
\label{ecriture action produit C sur vecteur + ... +-}
\end{multline}
Likewise, one has 
\bem
\Big(\op{v}_1^{-} \op{v}_2^{+} \cdots \op{v}_{n}^{+}  \Big)^{\op{t}} \pl{p=1}{m-1} \wt{\op{C}}_{1,\dots, n}^{\,(\ga)}\big(\bs{\be}_{n} ; u_{p}) \; = \; 
\Big(\op{v}_2^{+}\cdots \op{v}_{n-1}^{+} \Big)^{\op{t}} \pl{p=1}{m-1} \wt{\op{C}}_{2,\dots, n-1}^{\,(\ga)}\big(\bs{\be}_{n-1}^{\prime} ; u_{p}) 
\Big(\op{v}_1^{-}\, \op{v}_{n}^{+} \Big)^{\op{t}}  \\
\, + \, \sul{k=1}{m-1} \wt{c}^{\,(\ga)}(\be_n-u_k) \cdot  \f{   \pl{p=2}{n-1} \wt{b}(\be_p-u_k) }{ \pl{\substack{ s=1 \\ \not=k}  }{ m-1 } \wt{b}(u_{s} - u_k) }
\;  \Big(\op{v}_2^{+}\cdots \op{v}_{n-1}^{+} \Big)^{\op{t}} \pl{ \substack{ p=1 \\ \not=k}  }{m-1} \wt{\op{C}}_{2,\dots, n-1}^{\,(\ga)}\big(\bs{\be}_{n-1}^{\prime} ; u_{p}) 
\cdot \Big(\op{v}_1^{-}\, \op{v}_{n}^{-} \Big)^{\op{t}} \;. 
\label{ecriture action produit C sur vecteur -+ ... +}
\end{multline}

\end{lemme}

\Proof 

We start with \eqref{ecriture action produit C sur vecteur + ... +-}. Observe that one has
\bem
\big(\op{v}_n^{-}\big)^{\op{t}} \wt{\op{S}}^{\, (\ga)}_{n0}(\th) \, \op{v}_0^+ \, = \, \big(\op{v}_n^{-}\big)^{\op{t}} \wt{\op{S}}^{\, (\ga)}_{n0}(\th) \, 
\Big\{  \op{v}_0^+ \op{v}_n^{+}\big(\op{v}_n^{+}\big)^{\op{t}}  \, + \, \op{v}_0^+ \op{v}_n^{-}\big(\op{v}_n^{-}\big)^{\op{t}}   \Big\}  \\
\, = \, \big(\op{v}_n^{-}\big)^{\op{t}} \bigg\{    \Big[ \wt{\op{S}}^{\, (\ga)}(\th) \Big]^{++}_{++}  \op{v}_0^+ \op{v}_n^{+}\big(\op{v}_n^{+}\big)^{\op{t}} 
\, + \,  \Big[ \wt{\op{S}}^{\, (\ga)}(\th) \Big]^{-+}_{-+}  \op{v}_0^+ \op{v}_n^{-}\big(\op{v}_n^{-}\big)^{\op{t}}  \, + \, 
\Big[ \wt{\op{S}}^{\, (\ga)}(\th) \Big]^{-+}_{+-}  \op{v}_0^- \op{v}_n^{+}\big(\op{v}_n^{-}\big)^{\op{t}}  \bigg\}  \, = \, \wt{b}(\th) \op{v}_0^+ \big(\op{v}_n^{-}\big)^{\op{t}} \;. 
\label{ecriture sandwich Sno avec Vn+ transpose vO+}
\end{multline}
This first identity entails that 
\beq
\big(\op{v}_n^{-}\big)^{\op{t}} \,  \wt{\op{C}}^{\, (\ga)}_{1,\dots,n}\big( \bs{\be}_n ;\bs{u}_m\big) \, = \, \pl{s=1}{m}\wt{b}(\be_n-u_s) \cdot 
\wt{\op{C}}^{\, (\ga)}_{1,\dots,n-1}\big( \bs{\be}_{n-1} ; \bs{u}_m\big) \, \big(\op{v}_n^{-}\big)^{\op{t}} \;, 
\label{ecriture action transpose vn- sur produit des C}
\enq
in which we agree upon 
\beq
\wt{\op{C}}^{\, (\ga)}_{1,\dots,n}\big( \bs{\be}_n ; \bs{u}_m\big) \, = \, \pl{s=1}{m} \wt{\op{C}}^{\, (\ga)}_{1,\dots,n}\big( \bs{\be}_n ; u_s \big) \;. 
\enq
Also, clearly, one has
\beq
\big(\op{v}_1^{-}\big)^{\op{t}} \,  \wt{\op{C}}^{\, (\ga)}_{1,\dots,n}\big( \bs{\be}_n ; \bs{u}_m\big) \, = \,
\wt{\op{C}}^{\, (\ga)}_{2,\dots,n}\big( \bs{\be}_n^{\prime} ; \bs{u}_m\big) \, \big(\op{v}_1^{-}\big)^{\op{t}} \;. 
\enq
Further, it holds

\bem
\big(\op{v}_1^{+} \op{v}_0^- \big)^{\op{t}} \wt{\op{S}}^{\, (\ga)}_{10}(\th)    \, = \, \big(\op{v}_1^{+} \op{v}_0^- \big)^{\op{t}} \wt{\op{S}}^{\, (\ga)}_{10}(\th) \, 
\Big\{   \op{v}_1^{+} \op{v}_0^-\big(\op{v}_1^{+}\op{v}_0^{-}\big)^{\op{t}}  \, + \, \op{v}_1^{-} \op{v}_0^+\big(\op{v}_1^{-}\op{v}_0^{+}\big)^{\op{t}}   \Big\}  \\
\, = \,      \Big[ \wt{\op{S}}^{\, (\ga)}(\th) \Big]^{+-}_{+-}\big(\op{v}_1^{+}\op{v}_0^{-}\big)^{\op{t}} 
\, + \,  \Big[ \wt{\op{S}}^{\, (\ga)}(\th) \Big]^{-+}_{+-}  \big(\op{v}_1^{-}\op{v}_0^{+}\big)^{\op{t}}    \, = \,
\wt{b}(\th) \big(\op{v}_1^{+}\op{v}_0^{-}\big)^{\op{t}}  \, + \,  \wt{c}^{\,(\ga)}(\th) \big(\op{v}_1^{+}\op{v}_0^{-}\big)^{\op{t}}   \;. 
\end{multline}
From there, one infers that 
\beq
\big(\op{v}_1^{+}  \big)^{\op{t}}\,  \wt{\op{C}}^{\, (\ga)}_{1,\dots,n}\big( \bs{\be}_n ; u \big) \, = \, \wt{b}(\be_1-u) \,  \wt{\op{C}}^{\, (\ga)}_{2,\dots,n}\big( \bs{\be}_n^{\prime} ; u \big) \, \big(\op{v}_1^{+}  \big)^{\op{t}} 
\, + \, \wt{c}^{\,(\ga)}(\be_1-u) \,  \wt{\op{A}}^{\, (\ga)}_{2,\dots,n}\big( \bs{\be}_n^{\prime} ; u \big) \, \big(\op{v}_1^{-}  \big)^{\op{t}} \;.
\enq
This induction immediately leads to 
\bem
\big(\op{v}_1^{+}  \big)^{\op{t}}\,  \wt{\op{C}}^{\, (\ga)}_{1,\dots,n-1}\big( \bs{\be}_n ;  \bs{u}_m \big) \, = \, 
\pl{s=1}{m}\wt{b}(\be_1-u_s) \,  \wt{\op{C}}^{\, (\ga)}_{2,\dots,n-1}\big( \bs{\be}_n^{\prime} ;  \bs{u}_m \big) \, \big(\op{v}_1^{+}  \big)^{\op{t}}  \\
\, + \, \sul{k=1}{m} \wt{c}^{\,(\ga)}(\be_1-u_k) \, \pl{s=1}{k-1}\wt{b}(\be_1-u_s)   
\cdot  \wt{\op{C}}^{\, (\ga)}_{2,\dots,n-1}\big( \bs{\be}_n^{\prime} ;  \bs{u}_{k-1}\big)  \cdot \wt{\op{A}}^{\, (\ga)}_{2,\dots,n-1}\big( \bs{\be}_n^{\prime} ; u_k \big) 
\cdot  \wt{\op{C}}^{\, (\ga)}_{2,\dots,n-1}\big( \bs{\be}_n^{\prime} ;  \bs{u}_m^{(k+1)} \big)  \,  \big(\op{v}_1^{-}  \big)^{\op{t}} 
\label{ecriture action v+ transpose sur tilde C}
\end{multline}
where we have set $\bs{u}_m^{(k+1)}=\big( u_{k+1},\dots, u_m \big)$. Therefore, one gets that there are coefficients $\varphi_1,\dots,\varphi_m$ such that
\bem
\big(\op{v}_1^{+} \cdots \op{v}_{n-1}^{+}   \big)^{\op{t}}\,  \wt{\op{C}}^{\, (\ga)}_{1,\dots,n-1}\big( \bs{\be}_n ;  \bs{u}_m \big) \, = \, 
\pl{s=1}{m}\wt{b}(\be_1-u_s) \, \big(\op{v}_2^{+} \cdots \op{v}_{n-1}^{+}   \big)^{\op{t}}\,  \wt{\op{C}}^{\, (\ga)}_{2,\dots,n-1}\big( \bs{\be}_n^{\prime} ;  \bs{u}_m \big) \,  \big(\op{v}_1^{+}  \big)^{\op{t}}  \\
\, + \, \sul{k=1}{m}    \vp_k \cdot
\big(\op{v}_2^{+} \cdots \op{v}_{n-1}^{+}   \big)^{\op{t}}\,  \wt{\op{C}}^{\, (\ga)}_{2,\dots,n-1}\big( \bs{\be}_n^{\prime} ;  \wh{\bs{u}}_{m}^{(k)}\big)  \,  \big(\op{v}_1^{-}  \big)^{\op{t}} \;.
\end{multline}
The only way to generate the coefficient $\vp_m$ is through the direct action of the operator $ \wt{\op{A}}^{\, (\ga)}_{2,\dots,n-1}\big( \bs{\be}_n^{\prime} ; u_k \big) $ in \eqref{ecriture action v+ transpose sur tilde C}
for $k=m$. Hence,
\beq
\vp_m \, = \, \wt{c}^{\,(\ga)}(\be_1-u_m) \, \pl{ s=1 }{m-1} \f{ \wt{b}(\be_1-u_s) }{  \wt{b}(u_m-u_s) } \;. 
\enq
Therefore, by symmetry and upon applying \eqref{ecriture action transpose vn- sur produit des C}, one gets 
\bem
\big(\op{v}_1^{+} \cdots \op{v}_{n-1}^{+}  \op{v}_{n}^{-}  \big)^{\op{t}}\,  \wt{\op{C}}^{\, (\ga)}_{1,\dots,n}\big( \bs{\be}_n ;  \bs{u}_m \big) \, = \, 
\pl{s=1}{m}\Big\{ \wt{b}(\be_1-u_s)  \wt{b}(\be_n-u_s) \Big\}\, \big(\op{v}_2^{+} \cdots \op{v}_{n-1}^{+}   \big)^{\op{t}}\,
\wt{\op{C}}^{\, (\ga)}_{2,\dots,n-1}\big( \bs{\be}_n^{\prime} ;  \bs{u}_m \big) \,  \big(\op{v}_1^{+}  \op{v}_n^{-} \big)^{\op{t}}  \\
\, + \, \sul{k=1}{m}    \wt{c}^{\,(\ga)}(\be_1-u_k) \, \pl{ \substack{ s=1 \\ \not= k } }{m} \f{ \wt{b}(\be_1-u_s) }{  \wt{b}(u_m-u_s) } \, \pl{s=1}{m} \wt{b}(\be_n-u_s)  \cdot
\big(\op{v}_2^{+} \cdots \op{v}_{n-1}^{+}   \big)^{\op{t}}\,  \wt{\op{C}}^{\, (\ga)}_{2,\dots,n-1}\big( \bs{\be}_n^{\prime} ;  \wh{\bs{u}}_{m}^{(k)}\big)   \, \big(\op{v}_1^{-}  \op{v}_n^{-} \big)^{\op{t}} \;.
\end{multline}
We now move on to the proof of \eqref{ecriture action produit C sur vecteur -+ ... +}. 
\bem
\big(\op{v}_n^{+}\big)^{\op{t}} \wt{\op{S}}^{\, (\ga)}_{n0}(\th) \, \op{v}_0^+ \, = \, \big(\op{v}_n^{+}\big)^{\op{t}} \wt{\op{S}}^{\, (\ga)}_{n0}(\th) \, 
\Big\{  \op{v}_0^+ \op{v}_n^{+}\big(\op{v}_n^{+}\big)^{\op{t}}  \, + \, \op{v}_0^+ \op{v}_n^{-}\big(\op{v}_n^{-}\big)^{\op{t}}   \Big\}  \\
\, = \,    \op{v}_0^+  \big(\op{v}_n^{+}\big)^{\op{t}}  \; + \; 
\Big[ \wt{\op{S}}^{\, (\ga)}(\th) \Big]^{-+}_{+-}  \op{v}_0^-  \big(\op{v}_n^{-}\big)^{\op{t}}   \, = \, \op{v}_0^+  \big(\op{v}_n^{+}\big)^{\op{t}}  \; + \;  
\wt{c}^{\,(\ga)}(\th) \, \op{v}_0^- \big(\op{v}_n^{-}\big)^{\op{t}} \;. 
\end{multline}
Therefore, 
\beq
\big(\op{v}_n^{+}  \big)^{\op{t}}\,  \wt{\op{C}}^{\, (\ga)}_{2,\dots,n}\big( \bs{\be}_n^{\prime} ; u \big) \, = \,   \wt{\op{C}}^{\, (\ga)}_{2,\dots,n-1}\big( \bs{\be}_{n-1}^{\prime} ; u \big) \, \big(\op{v}_n^{+}  \big)^{\op{t}} 
\, + \, \wt{c}^{\,(\ga)}(\be_n-u) \,  \wt{\op{D}}^{\, (\ga)}_{2,\dots,n-1}\big( \bs{\be}_{n-1}^{\prime} ; u \big) \, \big(\op{v}_n^{-}  \big)^{\op{t}} \;.
\enq
Hence, upon using \eqref{ecriture action transpose vn- sur produit des C}, one gets 
\bem
\big(\op{v}_1^{-} \op{v}_2^{+} \cdots \op{v}_n^{+}  \big)^{\op{t}}\,  \wt{\op{C}}^{\, (\ga)}_{1,\dots,n-1}\big( \bs{\be}_n ;  \bs{u}_m \big) \, = \, 
\big(  \op{v}_2^{+} \cdots \op{v}_{n-1}^{+}  \big)^{\op{t}}\,    \wt{\op{C}}^{\, (\ga)}_{2,\dots,n-1}\big( \bs{\be}_n^{\prime} ;  \bs{u}_m \big) \, \big(\op{v}_1^{-}   \op{v}_n^{+} \big)^{\op{t}}  \\
\hspace{-8mm} \, + \, \sul{k=1}{m} \wt{c}^{\,(\ga)}(\be_n-u_k) \, \pl{s=k+1}{m}\wt{b}(\be_n-u_s)   
\cdot  \big(  \op{v}_2^{+} \cdots \op{v}_{n-1}^{+}  \big)^{\op{t}}\, \wt{\op{C}}^{\, (\ga)}_{2,\dots,n-1}\big( \bs{\be}_n^{\prime} ;  \bs{u}_{k-1}\big)  \cdot \wt{\op{D}}^{\, (\ga)}_{2,\dots,n-1}\big( \bs{\be}_n^{\prime} ; u_k \big) 
\cdot  \wt{\op{C}}^{\, (\ga)}_{2,\dots,n-1}\big( \bs{\be}_n^{\prime} ;  \bs{u}_m^{(k+1)} \big)  \,  \big(\op{v}_1^{-} \op{v}_n^{-} \big)^{\op{t}} \;. 
%
\end{multline}
This entails that the coefficient in front of  $\wt{\op{C}}^{\, (\ga)}_{2,\dots,n-1}\big( \bs{\be}_n^{\prime} ;  \bs{u}_{m-1}\big)$ is
\beq
\wt{c}^{\,(\ga)}(\be_n-u_m) \,  \f{ \pl{p=2}{n-1}\wt{b}(\be_p-u_m)   }{ \pl{s=1}{m-1} \wt{b}(u_s-u_m)  } \;. 
\enq
Then, by symmetry, 
\bem
\big(\op{v}_1^{-} \op{v}_2^{+} \cdots \op{v}_n^{+}  \big)^{\op{t}}\,  \wt{\op{C}}^{\, (\ga)}_{1,\dots,n-1}\big( \bs{\be}_n ;  \bs{u}_m \big) \, = \, 
\big(  \op{v}_2^{+} \cdots \op{v}_{n-1}^{+}  \big)^{\op{t}}\,    \wt{\op{C}}^{\, (\ga)}_{2,\dots,n-1}\big( \bs{\be}_n^{\prime} ;  \bs{u}_m \big) \, \big(\op{v}_1^{-}   \op{v}_n^{+} \big)^{\op{t}}  \\
\, + \, \sul{k=1}{m} \wt{c}^{\,(\ga)}(\be_n-u_k) \,\f{ \pl{p=2}{n-1}\wt{b}(\be_p-u_k)   }{ \pl{ \substack{s=1 \\ \not=k } }{m} \wt{b}(u_s-u_k)  } 
\cdot  \big(  \op{v}_2^{+} \cdots \op{v}_{n-1}^{+}  \big)^{\op{t}}\, \wt{\op{C}}^{\, (\ga)}_{2,\dots,n-1}\big( \bs{\be}_n^{\prime} ;  \wh{\bs{u}}_{m}^{(k)}\big)   \,  \big(\op{v}_1^{-} \op{v}_n^{-} \big)^{\op{t}} \;. 
\end{multline}

\subsection{The vanishing lemma}

\begin{lemme} \cite{BabujianKarowskiExactFFSineGordonBootsstrapII}
\label{Lemme annulation K transformee}
Let $n=2m$, then 
\beq
\bs{\Psi}_{1,\dots,n}\big[ 1 \big]\big( \bs{\be}_n \big) \, = \, \Int{ (\msc{C}_{\bs{\be}_n})^m }{} \hspace{-1mm} \dd^{m} u \; 
\wt{\Cx}_{1,\dots,n}\big(\bs{\be}_n;\bs{u}_m \big)  \cdot  g \big(\bs{\be}_n;\bs{u}_{m} \big)   \, = \, 0 \;. 
\enq

\end{lemme}

This result was already proven in \cite{BabujianKarowskiExactFFSineGordonBootsstrapII}. Here, we reproduce the proof for completeness.

\Proof 

It is well known that, for any polynomial $P(X) \; = \; \sul{s=0}{m-1}\a_s X^s$, it holds
\beq
\det\left[ \ba{ccccc} 1 & x_1& \cdots & x_{1}^{m-2} & P(x_1) \\   
\vdots & \vdots & & \vdots & \vdots   \\ 
1 & x_m& \cdots & x_{m}^{m-2} & P(x_m)  \ea \right]  \; = \; \a_{m-1} \pl{a<b}{m}(x_b-x_a) \;. 
\enq
Using that 
\beq
r(z) \, = \, \pl{a=1}{2m} \cosh\big[ \tfrac{1}{2}(\be_a-z) \big] \, - \, \pl{a=1}{2m} \sinh\big[ \tfrac{1}{2}(\be_a-z) \big]  \;= \;
2 \pl{a=1}{2m} \Big\{ \tfrac{1}{2} \ex{ \f{\be_a-z}{2} }  \Big\} \sul{p=0}{m-1} \ex{ (2p+1)z } \sg_{2p+1}\big( \ex{-\be_1},\dots, \ex{-\be_{2m}} \big) \;.
\enq
Above, the $\sg_{2p+1}$ stand for the symmetric elementary polynomials in $2m$ variables:
\begin{equation}
	\pl{s=1}{2m}(X+X_s) \, = \, \sul{k=0}{2m} X^k \sg_k\big(X_1,\dots, X_m \big) \;.
\end{equation}
This thus entails that 
\beq
\pl{a<b}{m} \sinh(u_{ab}) \, = \,   4^m   \f{ \ex{-(m-1) \ov{\bs{u}}_m  \, + \, \f{1}{2} \ov{\bs{\be}}_{2m} }  }{ 2 \cdot 2^{ m\f{m-1}{2} } \sul{a=1}{2m} \ex{\be_a}   }
\cdot 
\det\left[ \ba{ccccc} 1 &   \cdots & \ex{ 2 (m-2) u_1 } & \ex{  (m-1) u_1 }  r(u_1) \\   
\vdots   & & \vdots & \vdots   \\ 
1 &   \cdots &  \ex{ 2 (m-2) u_m } & \ex{  (m-1) u_m }  r(u_m)   \ea \right]  \;. 
\enq
Inserting this identity into one of the factors arising in $\pl{a<b}{m} \tau(u_{ab})$ and then using the antisymmetry of the remaining factor so as to replace the determinant by the product of its diagonals, one gets that 
\bem
\bs{\Psi}_{1,\dots,n}\big[ 1 \big]\big( \bs{\be}_n \big) \, = \, \mc{N}_m(\bs{\be}_n) \cdot  \Int{ (\msc{C}_{\bs{\be}_n})^m }{} \hspace{-1mm} \dd^{m} u \; 
\wt{\Cx}_{1,\dots,n}\big(\bs{\be}_n;\bs{u}_m \big)    \cdot \pl{k=1}{n} \pl{\ell=1}{m} \big( a \cdot \phi)(\be_k-u_{\ell}) \\ 
\times \pl{k<\ell}{m}  \sinh \big[ \tfrac{\pi u_{\ell k} }{\xi}  \big] \cdot  r(u_m) \cdot 
\pl{k=1}{m-1} \ex{ (2(k-1)-m+1)u_k } \;,  
\end{multline}
where, recalling that $n=2m$,  
\beq
\mc{N}_m(\bs{\be}_n) \; = \;  \f{  4^m  \, m! \, \ex{ \f{1}{2} \ov{\bs{\be}}_{2m} }  }{ 2 \cdot 2^{ m\f{m-1}{2} } \sul{a=1}{2m} \ex{\be_a}   } \cdot  \pl{k<\ell}{n} \op{F}(\be_{k\ell}) \cdot 
\bigg\{ \op{F} \big( \i \tfrac{\pi}{2} \big)\, \op{F}\big( - \i \tfrac{\pi}{2} \big) \, \varpi\Big(\i \tfrac{\xi + \pi}{2} \, , \,  \i \tfrac{\xi - 3 \pi}{2} \Big)  \bigg\}^{ m (m-1) } \;. 
\enq
The integral may further be reorganised as
\beq
\bs{\Psi}_{1,\dots,n}\big[ 1 \big]\big( \bs{\be}_n \big) \, = \, \mc{N}_m(\bs{\be}_n) \cdot  \Int{ (\msc{C}_{\bs{\be}_n})^{m-1}  }{} \hspace{-1mm} \dd^{m-1} u \; 
\pl{k=1}{n} \pl{\ell=1}{m-1} \big( a \cdot \phi)(\be_k-u_{\ell}) 
\cdot \pl{k<\ell}{m-1}  \sinh \big[ \tfrac{\pi  }{\xi} (u_{ \ell k } ) \big] 
\pl{k=1}{m-1} \ex{ (2(k-1)-m+1)u_k } \cdot \mc{J}\big( \bs{\be}_n , \bs{u}_{m-1} \big)  \;,  
\enq
in which 
\bem
\mc{J}\big( \bs{\be}_n , \bs{u}_{m-1} \big) \; = \;  \Int{ \msc{C}_{\bs{\be}_n} }{} \hspace{-1mm} \dd u \; 
\wt{\Cx}_{1,\dots,n}\Big(\bs{\be}_n; \big( \bs{u}_{m-1},u\big) \Big) 
\cdot \pl{a=1}{m-1}\sinh \big[ \tfrac{\pi }{\xi} (u-u_a) \big] 
\cdot \pl{k=1}{n} \big( a \cdot \phi)(\be_k-u)  \\
\times \bigg\{ \pl{a=1}{2m} \cosh\big[ \tfrac{1}{2}(\be_a-u) \big] \, - \, \pl{a=1}{2m} \sinh\big[ \tfrac{1}{2}(\be_a-u) \big]  \bigg\} \;. 
\end{multline}
By using that 
\beq
\wt{\op{S}}_{ab}(\la + \i\xi) \, = \, \sg^{z}_{b} \wt{\op{S}}_{ab}( \la ) \sg^{z}_{b}
\enq
one infers that $\wt{\op{T}}_{1,\dots,n;0}(\bs{\be}_n;\la \pm \i \xi ) \, = \, \sg^{z}_{0} \wt{\op{T}}_{1,\dots,n;0}(\bs{\be}_n;\la ) \sg^{z}_{0}$,
and in particular, $\wt{\op{C}}_{1,\dots,n}(\bs{\be}_n;\la \pm \i \xi ) \, = \, - \wt{\op{C}}_{1,\dots,n}(\bs{\be}_n;\la )$. 

Further, one also has 
\beq
\big( a  \phi)(\la - \i \xi )  \, = \, \big( a  \phi)(\la ) \cdot \f{ \i \sinh( \tf{\la}{2} )  }{  \cosh\big[ \f{ \la - \i\xi }{ 2 } \big]  } \;. 
\enq
Upon putting these pieces of information together, one gets that
\bem
\mc{J}\big( \bs{\be}_n , \bs{u}_{m-1} \big) \; = \;  \Int{ \msc{C}_{\bs{\be}_n} }{} \hspace{-1mm} \dd u \; 
\wt{\Cx}_{1,\dots,n}\Big(\bs{\be}_n; \big( \bs{u}_{m-1},u\big) \Big) 
\cdot \pl{a=1}{m-1}\sinh \big[ \tfrac{\pi }{\xi} (u-u_a) \big] 
\cdot \pl{k=1}{n} \Big\{ \big( a \cdot \phi)(\be_k-u) \cosh\big[ \tfrac{1}{2}(\be_k-u) \big] \Big\} \\
\; - \; \Int{ \msc{C}_{\bs{\be}_n} }{} \hspace{-1mm} \dd u \; 
\wt{\Cx}_{1,\dots,n}\Big(\bs{\be}_n; \big( \bs{u}_{m-1},u + \i \xi \big) \Big) 
\cdot \pl{a=1}{m-1}\sinh \big[ \tfrac{\pi }{\xi} (u - u_a + \i \xi ) \big] 
\cdot \pl{k=1}{n} \Big\{ \big( a \cdot \phi)(\be_k-u-\i\xi ) \cosh\big[ \tfrac{1}{2}(\be_k-u-\i \xi) \big] \Big\} \\
\; = \; \Bigg( \Int{ \msc{C}_{\bs{\be}_n} }{} \, -  \hspace{-2mm} \Int{ \msc{C}_{\bs{\be}_n}+\i\xi  }{}  \Bigg)   \dd u \; 
\wt{\Cx}_{1,\dots,n}\Big(\bs{\be}_n; \big( \bs{u}_{m-1},u\big) \Big) 
\cdot \pl{a=1}{m-1}\sinh \big[ \tfrac{\pi }{\xi} (u-u_a) \big] 
\cdot \pl{k=1}{n} \Big\{ \big( a \cdot \phi)(\be_k-u) \cosh\big[ \tfrac{1}{2}(\be_k-u) \big] \Big\} \;.
\label{reecriture integrale J}
\end{multline}
First of all, one may check that the integral is convergent at infinity.
By using the asymptotics \eqref{asymptotic aphi} and the fact that $\wt{\Cx}_{1,\dots,n}\Big(\bs{\be}_n; \big( \bs{u}_{m-1},u\big) \Big) $ is bounded in norm
as $\Re(u) \tend \pm \infty$, one has the upper bound
\beq
\big| \big|  \wt{\Cx}_{1,\dots,n}\Big(\bs{\be}_n; \big( \bs{u}_{m-1},u\big) \Big) 
\cdot \pl{a=1}{m-1}\sinh \big[ \tfrac{\pi }{\xi} (u-u_a) \big] 
\cdot \pl{k=1}{n} \Big\{ \big( a \cdot \phi)(\be_k-u) \cosh\big[ \tfrac{1}{2}(\be_a-u) \big] \Big\}   \big| \big| \; \leq \; 
C \ex{ \mp \f{\pi}{\xi} \Re(\mu) } \;, 
\enq
what ensures the convergence of the integral at infinity. \\
Further, one has that, in terms of the potential poles $ \msc{C}_{\bs{\be}_n}\setminus \big\{ \msc{C}_{\bs{\be}_n} + \i \xi \big\} $
is homotopic to 
\bem
\bigcup\limits_{a=1}^{n} \bigcup\limits_{\ell \geq 0} \bigg\{   \bigcup\limits_{p \geq 0}\Big\{  \Dp{}\mc{D}_{\be_a+\i \ell \xi + \i\pi (2p+1),\eps}  \Big\} \bigcup
\bigcup\limits_{p \geq 0}\Big\{  - \Dp{}\mc{D}_{\be_a+\i (\ell+1) \xi + \i\pi (2p+1),\eps}  \Big\}  \bigcup 
\Big\{  \Dp{}\mc{D}_{\be_a+\i \ell \xi - \i\pi ,\eps}  \Big\} \bigcup 
\Big\{  - \Dp{}\mc{D}_{\be_a+\i (\ell+1) \xi - \i\pi ,\eps}  \Big\}  \\ 
\bigcup\limits_{p \geq 0}\Big\{ - \Dp{}\mc{D}_{\be_a - \i \ell \xi - 2 \i\pi p,\eps}  \Big\} \bigcup 
\bigcup\limits_{p \geq 0}\Big\{   \Dp{}\mc{D}_{\be_a - \i (\ell-1) \xi - 2\i\pi p ,\eps}  \Big\}  \bigcup 
\Big\{ - \Dp{}\mc{D}_{\be_a - \i (\ell+1) \xi - \i\pi ,\eps}  \Big\} \bigcup 
\Big\{   \Dp{}\mc{D}_{\be_a - \i \ell \xi - \i\pi ,\eps}  \Big\}   \bigg\} \\
\; = \; \bigcup\limits_{a=1}^{n} \bigg\{   \bigcup\limits_{ p \geq 0}\Big\{  \Dp{}\mc{D}_{\be_a + \i\pi (2p+1),\eps}  \Big\}
\bigcup   \Dp{}\mc{D}_{\be_a   - \i\pi ,\eps}  
\bigcup\limits_{  p \geq 0}\Big\{   \Dp{}\mc{D}_{\be_a + \i \xi - 2\i\pi p,\eps}  \Big\}  \bigcup 
\Dp{}\mc{D}_{\be_a - \i\pi ,\eps} \bigg\} \;.
\end{multline}
However, the poles of $\wt{\Cx}$ are given by \eqref{series de pole covecteur des C} while those of the $a\phi$ factor in can be read off from \eqref{series de pole a phi}. 
Thus, the integrand in \eqref{reecriture integrale J} has no poles at $\be_a + \i \xi - 2\i\pi p$, $p \in \mathbb{N}$, while the simple poles at $\be_a + \i\pi (2p+1)$ are compensated by the 
hyperbolic cosine factor's zeroes. Thus, all-in-all, $\mc{J}\big( \bs{\be}_n , \bs{u}_{m-1} \big) \; = \; 0$ and the claim follows. \qed


\begin{thebibliography}{10}

\bibitem{ArinshteinFateyevZamolodchikovSMatrixTodaChain}
A.E. Arinshtein, V.A. Fateev, and A.B. Zamolodchikov, \emph{{"Quantum S-matrix
  of the (1+1) dimensional Toda chain."}}, Phys. Lett. B \textbf{\bf{87}}
  (1979), 389--392.

\bibitem{BabujianFringKarowskiZapletalExactFFSineGordonBootsstrapI}
H.~Babujian, A.~Fring, M.~Karowski, and A.~Zapletal, \emph{{"Exact form factors
  in integrable quantum field theories: the sine-Gordon model."}}, Nucl. Phys.
  B \textbf{\bf{538}} (1999), 535--586.

\bibitem{BabujianKarowskiExactFFSineGordonBootsstrapII}
H.~Babujian and M.~Karowski, \emph{{"Exact form factors in integrable quantum
  field theories: the sine-Gordon model (II)."}}, Nucl. Phys. B
  \textbf{\bf{620}} (2002), 407--455.

\bibitem{BabujianKarowskiBreatherFFSineGordon}
\bysame, \emph{{"Sine-Gordon breather form factors and quantum field
  equations."}}, J.Phys.A \textbf{\bf{35}} (2002), 9081--9104.

\bibitem{BabujianKarowskiZapletalSomeDvpmtofOffShellBA}
H.~Babujian, M.~Karowski, and A.~Zapletal, \emph{{"Matrix difference equations
  and a nested Bethe ansatz."}}, J. Phys. A: Math. Gen. \textbf{\bf{30}}
  (1997), 6425--6450.

\bibitem{BaxterPartitionfunction8Vertex-FreeEnergy}
R.J. Baxter, \emph{{"Partition function of the eight vertex lattice model."}},
  Ann. Phys. \textbf{\bf 70} (1972), 193--228.

\bibitem{BernardLeclairDiffEqnForPIII}
D.~Bernard and A.~Leclair, \emph{{"Differential equations for sine-Gordon
  correlation functions at the free fermion point."}}, Nucl. Phys. B
  \textbf{\bf 426} (1994), 534--558.

\bibitem{BostelmannCadamuroFFEqnsInIQFTsWithScalarS}
H.~Bostelmann and D.~Cadamuro, \emph{{"Characterization of local observables in
  integrable quantum field theories."}}, Comm. Math. Phys. \textbf{\bf 337}
  (2015), 1199--1240.

\bibitem{BrazhnikovLukyanovFreeFieldRepMassiveFFIntegrable}
V.~Brazhnikov and S.~Lukyanov, \emph{{"Angular quantization and form factors in
  massive integrable models."}}, Nucl. Phys. B \textbf{\bf{512}} (1998),
  616--636.

\bibitem{DashenHasslacherNeveuIdentificationFromPertThFullPartSPectrSineGordon}
R.F. Dashen, B.~Hasslacher, and A.~Neveu, \emph{{"Particle spectrum in model
  field theories from semiclassical functional integral techniques."}}, Phys.
  Rev. D \textbf{\bf 11} (1975), 3424--3450.

\bibitem{FaddeevIrreducibiliteModularDouble}
L.D. Faddeev, \emph{{"Discete Heisenberg-Weyl group and modular group."}},
  Lett. Math. Phys. \textbf{\bf{34}} (1995), 249--254.

\bibitem{FringMussardoSimonettiFFSOmeLocalObsSinhGordon}
A.~Fring, G.~Mussardo, and P.~Simonetti, \emph{{"Form factors for integrable
  lagrangian field theories, the sinh-Gordon model."}}, Nucl. Phys. B
  \textbf{\bf 393} (1993), 413--441.

\bibitem{GryanikVergelesSMatrixAndOtherStuffForSinhGordon}
V.M. Gryanik and S.N. Vergeles, \emph{{"Two-dimensional quantum field theories
  having exact solutions."}}, J. Nucl. Phys. \textbf{\bf 23} (1976),
  1324--1334.

\bibitem{HeisenbergSomeAspectsofSMatrixIdeas}
W.~Heisenberg, \emph{{"Der mathematische Rahmen der Quantentheorie der
  Wellenfelder."}}, Zeit. f\"{u}r Naturforschung \textbf{\bf{1}} (1946),
  608--622.

\bibitem{JimboMiwaSmirnovFormFactorsSineGNewCOnstructionViaHiddenGrassmann}
M.~Jimbo, T.~Miwa, and F.~Smirnov, \emph{{"Fermionic structure in the
  sine-Gordon model: form factors and null-vectors ."}}, Nucl. Phys. B
  \textbf{852} (2011), 390--440.

\bibitem{KarowskiThunCompleteSMatrixThirring}
M.~Karowski and H.J. Thun, \emph{{"Complete S-matrix of the massive Thirring
  model."}}, Nucl. Phys. B \textbf{\bf 130} (1978), 295--308.

\bibitem{KarowskiWeiszFormFactorsFromSymetryAndSMatrices}
M.~Karowski and P.~Weisz, \emph{{"Exact form factors in (1+1)-dimensional field
  theoretic models with soliton behaviour."}}, Nucl. Phys. B \textbf{\bf 139}
  (1978), 455--476.

\bibitem{Kashaev3termIntegralRelationDfcts}
R.~Kashaev, \emph{{"The Quantum Dilogarithm and Dehn Twists in Quantum
  Teichm\"{u}ller Theory."}}, Integrable Structures of Exactly Solvable
  Two-Dimensional Models of Quantum Field Theory. NATO Science Series (Series
  II: Mathematics, Physics and Chemistry), vol 35. Springer, Dordrecht. Edts:
  Pakuliak S., von Gehlen G. (2001), 211--221.

\bibitem{KirillovSmirnovFirstCompleteSetBootstrapAxiomsForQIFT}
A.N. Kirillov and F.A. Smirnov, \emph{{"A representation of the current algebra
  connected with the SU (2)-invariant Thirring model."}}, Phys. Rev. B
  \textbf{\bf 198} (1987), 506--510.

\bibitem{KirillovSmirnovUseOfBootstrapAxiomsForQIFTToGetMassiveThirringFF}
\bysame, \emph{{"Form-factors in the SU(2)-invariant Thirring model."}}, J.
  Soviet. Math. \textbf{\bf 47} (1989), 2423--2450.

\bibitem{KorepinFaddeevQuantisationOfSolitions}
V.E. Korepin and L.D. Faddeev, \emph{{"Quantisation of solitons."}}, Theor.
  Math. Phys. \textbf{\bf 25} (1975), 1039--1049.

\bibitem{KoubekMussardoFFForMoreOpInSinhGordon}
A.~Koubek and G.~Mussardo, \emph{{"On the operator content of the sinh-Gordon
  model."}}, Phys. Lett. B \textbf{311} (1993), 193--201.

\bibitem{KozConvergenceFFSeriesSinhGordon2ptFcts}
K.K. Kozlowski, \emph{{"On convergence of form factor expansions in the
  infinite volume quantum Sinh-Gordon model in 1+1 dimensions."}}, Inventiones
  mathematicae \textbf{233} (2023), 725--827.

\bibitem{KurokawaDoubleSineIntro}
N.~Kurokawa, \emph{{"Multiple Sine Functions and Selberg Zeta Functions."}},
  Proc. Japan Acad. Ser. A Math. Sci. \textbf{67} (1991), 61--64.

\bibitem{LechnerAlgebraicConstructionIQFTSscalarS}
G.~Lechner, \emph{{"Construction of Quantum Field Theories with Factorizing
  S-Matrices."}}, Comm. Math. Phys. \textbf{\bf{277}} (2008), 821--860.

\bibitem{LukyanovFirstIntroFreeField}
S.~Lukyanov, \emph{{"Free field representation for massive integrable
  models."}}, Comm. Math. Phys. \textbf{\bf{167}} (1995), 183--226.

\bibitem{LukyanovConjectureFFExponentialFieldSineGSolASolAndBReather}
\bysame, \emph{{"Form-factors of exponential fields in the sine-Gordon
  model."}}, Mod. Phys. Lett. \textbf{A \bf{12}} (1997), 2543--2550.

\bibitem{LukyanovZamolodchikovSolitonCreatingOpsSineG}
S.~Lukyanov and Al.B. Zamolodchikov, \emph{{"Form factors of soliton-creating
  operators in the sine-Gordon model."}}, Nucl. Phys. B \textbf{{ \bf 607}}
  (2001), 437--455.

\bibitem{PalmaiMultiSolitionFFSineGField}
T.~P\'{a}lmai, \emph{{"Regularization of multi-soliton form factors in
  sine-Gordon model."}}, Comp. Phys. Comm. \textbf{\bf 183} (2012), 1813--1821.

\bibitem{ReshetikhinOffShellBAforKZ}
N.Yu. Reshetikhin, \emph{{"Jackson-Type Integrals, Bethe Vectors, and Solutions
  to a Difference Analog of the Knizhnik-Zamolodchikov System."}}, Lett. Math.
  Phys. \textbf{\bf 26} (1992), 153--165.

\bibitem{RuijsenaarsFirstIntroQuantumRelatToda}
S.~Ruijsenaars, \emph{{"The relativistic Toda systems."}}, Commun.Math.Phys.
  \textbf{\bf 133} (1990), 217--247.

\bibitem{SchroeTruongFFExpPhyInFreeFermionPtSineG}
B.~Schroer and T.T. Truong, \emph{{"The order/disorder quantum field operators
  associated with the two-dimensional Ising model in the continuum limit."}},
  Nucl. Phys. B \textbf{144} (1978), 80--122.

\bibitem{SmirnovUseGLMEqnsTocomputeSineGordonFF}
F.A. Smirnov, \emph{{"Quantum Gelfand-Levitan-Marchenko equations and form
  factors in the sine-Gordon model."}}, J. Phys. A: Math. Gen. \textbf{\bf 17}
  (1984), L873--L878.

\bibitem{SmirnovGLMEqnsDerivationForSineGordon}
\bysame, \emph{{"Quantum Gelfand-Levitan-Marchenko equations for the
  sine-Gordon model."}}, Theor. Math. Phys. \textbf{\bf 60} (1984), 871--880.

\bibitem{SmirnovUseGLMEqnsForSineGordonFF}
\bysame, \emph{{"Solution of quantum Gel'fand-Levitan-Marchenko equations for
  the sine-gordon model in the soliton sector for $\gamma=\pi /\nu$."}}, Theor.
  Math. Phys. \textbf{\bf 67} (1986), 344--351.

\bibitem{SmirnovIntegralRepSolitonFFSineGordonBootstrap}
\bysame, \emph{{"The general formula for solitons form factors in sine-Gordon
  model."}}, J. Phys. A \textbf{\bf 19} (1986), L575--578.

\bibitem{SmirnovProofIntegralRepSolitonFFSineGordonBootstrapFollowUpJPhysAPaper}
\bysame, \emph{{"Proof of some identities which arise in calculating
  form-factors in the sine-Gordon model."}}, J. Sov. Math. \textbf{\bf 46}
  (1989), 2111--2125.

\bibitem{SmirnovReductionsAndClusterPropertyInSineGordonPlusSomeDiscussionsUVLimit}
\bysame, \emph{{"Reductions of the sine-Gordon model as a perturbation of
  minimal models of conformal field theory."}}, Nucl. Phys. B \textbf{\bf 337}
  (1990), 156--180.

\bibitem{SmirnovFormFactors}
\bysame, \emph{{"Form factors in completely integrable models of quantum field
  theory."}}, Advanced Series in Mathematical Physics, vol.~14, World
  Scientific, 1992.

\bibitem{WheelerFirstIntroConceptSmatrix}
J.A. Wheeler, \emph{{"On the mathematical description of light nuclei by the
  method of resonating group structure."}}, Phys. Rev. B \textbf{\bf 52}
  (1937), 1107--1106.

\bibitem{WoronowiczQuantumExpFcts}
S.L. Woronowicz, \emph{{"Quantum exponential function"}}, Rev. Math. Phys.
  \textbf{\bf{12}} (2000), 873--920.

\bibitem{YangFactorizingDiffusionWithPermutations}
C.N. Yang, \emph{{"Some exact results for the many-body problem in one
  dimension with repulsive delta-function interaction."}}, Phys. Rev. Lett.
  \textbf{\bf 19} (1967), 1312--1315.

\bibitem{YurovZalmolodchikovFirstIntroMutualLocalityIndex}
V.P. Yurov and Al.B. Zamolodchikov, \emph{{"Correlation functions of integrable
  2D models of the relativistic field theory; Ising model."}}, Intl. J. Mod.
  Phys. A \textbf{\bf 6} (1991), 3419--3440.

\bibitem{ZalZalBrosFactorizedSMatricesIn(1+1)QFT}
A.B. Zamolodchikov and Al.B. Zamolodchikov, \emph{{"Factorized S-matrices in
  two dimensions as the exact solutions of certain relativistic quantum field
  theory models."}}, Ann. of Phys. \textbf{\bf 120} (1979), 253--291.

\bibitem{ZalmolodchikovSMatrixSolitonAntiSolitonSineGordon}
Al.B. Zamolodchikov, \emph{{"Exact Two-Particle S-Matrix of Quantum Sine-Gordon
  Solitons."}}, Comm. math. Phys. \textbf{\bf 55} (1977), 183--186.

\bibitem{ZalmolodchikovTwoPointFctsLeeYang}
\bysame, \emph{{"Two-point correlation functions in scaling Lee-Yang model."}},
  Nucl. Phys. B \textbf{\bf 348} (1991), 619--641.

\end{thebibliography}
\end{document}